\documentclass[
		runningheads,
		10pt,
		dvipsnames]
		{llncs}

\usepackage{booktabs} 
\usepackage{amsmath}
\usepackage{amssymb}
\usepackage{amsfonts}
\usepackage{listings}
\usepackage{lscape}
\usepackage{todonotes}
\usepackage{algorithm, algpseudocode}
\usepackage{subfig}
\usepackage{mathtools}
\usepackage{multirow}
\usepackage{tabularx}
\usepackage[hidelinks]{hyperref}
\usepackage{geometry}
\usepackage{tikz}
\usetikzlibrary{shapes,fadings,shadows,trees,matrix,intersections,patterns,decorations.markings,decorations.pathreplacing,arrows,automata,calc,decorations,fit,positioning,backgrounds}
\usetikzlibrary{shapes.geometric,decorations.pathmorphing}

\tikzset{
	aut/.style      = {auto,node distance=1cm,line width=2pt,>=to,thick,align=center,shorten >=1pt},
	location/.style = {state,fill=white!89!black,draw=white!50!black,minimum size=0.7cm,inner sep=0cm},
}

\tikzset{
	accept/.style  = {double},
	reject/.style  = {pattern=north east lines, pattern color=black!30}
}

\tikzset{
	atomic/.style	= {minimum size=6.5cm},
	a-big/.style	= {minimum size=7.5cm},
	composite/.style= {minimum size=2.0cm, minimum height=1cm},
	explode/.style	= {node distance= 3.5cm}
}

\newcommand{\tconnect}[5][]{ %
	\path[->,#1] (#2) edge [         ] node (#5)[label=#4]   {} (#3);%
}
\newcommand{\tconnectt}[6][]{ %
	\path[->,#1] (#2) edge [         ] node (#5)[label={[#6]:#4}]   {} (#3);%
}

\newcommand{\rfig}[1]{Figure~\ref{#1}}
\newcommand{\rfigb}[1]{Figure~\ref{#1}}
\newcommand{\rtbl}[1]{Table~\ref{#1}}
\newcommand{\rdef}[1]{Definition~\ref{#1}}

\newcommand{\rprop}[1]{Proposition~\ref{#1}}

\newcommand{\rlemma}[1]{Lemma~\ref{#1}}
\newcommand{\ralg}[1]{Algorithm~\ref{#1}}

\newcommand{\rsec}[1]{Section~\ref{#1}}

\newcommand{\rex}[1]{Example~\ref{#1}}

\newcommand{\rapp}[1]{Appendix~\ref{#1}}

\newcommand{\gettable}[1]{%
	\begin{table}[t]%
	\centering%
		#1%
	\end{table}%
}

\newcommand{\squishlist}{
\begin{list}{-}
 { \setlength{\itemsep}{0pt}
    \setlength{\parsep}{1pt}
    \setlength{\topsep}{1pt}
    \setlength{\partopsep}{0pt}
    \setlength{\leftmargin}{0.6em}
    \setlength{\labelwidth}{1.5em}
    \setlength{\labelsep}{0.4em} } }
\newcommand{\squishend}{
 \end{list}  }

\newcommand{\wrt}[0]{wrt}

\newcommand{\cblock}[1]{}

\renewcommand\equiv{\Leftrightarrow}
\newcommand{\matchr}{{\kern 0 em \rightarrow\,}}

\newcommand\THEMIS{\texttt{THEMIS}}
\newcommand{\vars}{\texttt}
\newcommand{\func}{\mathrm}
\newcommand{\funca}{\textsc}
\newcommand{\code}[1]{\ensuremath{\mathtt{#1}}}
\newcommand\xv{x}
\DeclareMathOperator{\fdom}{dom}
\DeclareMathOperator{\fcodom}{codom}
\newcommand\pif[1][\>]{\ensuremath{#1 \mbox{\text{if }}}}
\newcommand\pelse[1][\>]{\ensuremath{#1 \mbox{\text{otherwise}}}}
\newcommand\tuple[1]{\langle #1 \rangle}
\newcommand\aut{\mathcal{A}}

\newcommand\mem{\mathcal{M}}
\newcommand\Mem{\mathit{Mem}}
\newcommand\sys{\mathcal{I}}

\newcommand\verdict{\mathbb{B}_3}
\newcommand\verdictb{\mathbb{B}_2}
\newcommand\timestamp{\mathbb{N}}

\newcommand\sysadd{\dagger_{\lor}}
\newcommand\memadd{\dagger_2}

\newcommand\atoms{\mathit{Atoms}}
\newcommand\chorcoref{\mathit{coref}}
\newcommand\chorref{\mathit{ref}}
\newcommand\chorresp{\mathit{respawn}}
\newcommand\scor{\mathrm{scor}}

\newcommand\msgkill{\mathrm{msg}_{\rm kill}}
\newcommand\msgverdict{\mathrm{msg}_{\rm ver}}
\newcommand\expr{\mathit{Expr}}
\newcommand\clookup{\mathrm{lu}}
\newcommand\cchoose{\mathrm{chc}}
\newcommand\csplit{\mathrm{split}}
\newcommand\verdictf{\mathrm{ver}}
\newcommand\verdictAt{\mathrm{verAt}}
\newcommand\seval{\func{eval}}
\newcommand\senc{\func{enc}}

\newcommand\AP{\mathit{AP}}
\newcommand\APmons{{\mathit{AP}_{\mathrm{mons}}}}
\newcommand\smove{\mathrm{mov}}
\newcommand\sreach{\mathrm{sel}}
\newcommand\sto{\mathrm{to}}
\newcommand\snext{\mathrm{next}}

\newcommand\MONID{\mathrm{m}}
\newcommand\paths{\ensuremath{\mathrm{paths}}}
\newcommand\mondep{\mathrm{dep}}
\newcommand\mons{\mathrm{Mons}}
\newcommand\comps{\mathcal{C}}

\newcommand\setof[1]{\{#1\}}
\newcommand\trace{\mathrm{tr}}
\newcommand\vna{\mathtt{\mathbf{?}}}
\newcommand\vt{\ensuremath{\mathtt{true}}}
\renewcommand\vt{\ensuremath{\top}}
\newcommand\vf{\ensuremath{\mathtt{false}}}
\renewcommand\vf{\ensuremath{\bot}}
\newcommand\fundef{\mathrm{undef}}
\newcommand\cons{\mathrm{memc}}
\newcommand\migrchc{\mathrm{choose}}

\newcommand\rw{\mathrm{rw}}
\newcommand\fid{\mathrm{idt}}
\newcommand\fts{\mathrm{ts}}
\newcommand\fsimp{\mathrm{simplify}}
\newcommand\finc{\mathrm{inc}}

\newcommand\fenc[1]{\ensuremath{\func{s}_{#1}}}
\newcommand\monitorable{\mathrm{monitorable}}
\newcommand\compat{\mathrm{compat}}
\newcommand\constraint{\mathrm{cdep}}
\newcommand\reach[1]{\mathrm{reach}_{\mathrm{#1}}}

\newcommand\MDG{\ensuremath{\mathrm{MDG}}}
\newcommand\MDS{\ensuremath{\mathrm{MDS}}}

\newcommand\decent{\ensuremath{\mathcal{D}}}
\newcommand\traces{\ensuremath{\mathcal{T}}}
\newcommand\cmapping{\ensuremath{\mathcal{L}}}
\newcommand\mroot{\ensuremath{\mathrm{rt}}}
\newcommand\emptytrace{\ensuremath{\epsilon}}

\newcommand\paramlabel{L}

\newcommand\paramdelay{\delta_{t}}%
\def\MemRep{\texttt{EHE}}
\newcommand\sstable{\mathrm{kn}}
\newcommand\algorch{Orch}
\newcommand\algmigr{Migr}
\newcommand\algmigrr{Migrr}
\newcommand\algchor{Chor}

\newcommand{\timestamps}[1]{\ensuremath{\mathrm{rounds}({#1})}}

\newcommand{\xelse}{\mbox{ otherwise }}
\newcommand{\funcparts}[2][ll]
{
	\left\{
		\begin{array}{#1}
			#2
		\end{array}
	\right.
}

\newcommand{\twopartdef}[3]
{
	\left\{
		\begin{array}{ll}
			#1 & \mbox{if } #2 \\
			#3 & \mbox{otherwise}
		\end{array}
	\right.
}
\newcommand{\mrm}[1]{\mathrm{#1}}

\begin{document}

\author{
Antoine El-Hokayem\inst{1}
\and
Yli\`es Falcone\inst{1}
\institute{
Univ. Grenoble Alpes, CNRS, Inria, Grenoble INP\thanks{Institute of Engineering Univ. Grenoble Alpes}, LIG, \\ F-38000 Grenoble, France \\ \email{firstname.lastname@univ-grenoble-alpes.fr}
}
}

\title{On the Monitoring of Decentralized Specifications}%
\subtitle{Semantics, Properties, Analysis, and Simulation}
\titlerunning{On the Monitoring of Decentralized Specifications}
\authorrunning{Antoine El-Hokayem and Yl\`es Falcone}

\maketitle

\begin{abstract}
We define two complementary approaches to monitor decentralized systems.
The first relies on those with a centralized specification, i.e, when the specification is written for the behavior of the entire system.
To do so, our approach introduces a data-structure that i) keeps track of the execution of an automaton, ii) has predictable parameters and size, and iii) guarantees strong eventual consistency.
The second approach defines decentralized specifications wherein multiple specifications are provided for separate parts of the system.
We study two properties of decentralized specifications pertaining to monitorability and compatibility between specification and architecture.
We also present a general algorithm for monitoring decentralized specifications.
We map three existing algorithms to our approaches and provide a framework for analyzing their behavior.
Furthermore, we introduce \THEMIS{}, a framework for designing such decentralized algorithms and simulating their behavior.
We show the usage of \THEMIS{} to compare multiple algorithms and verify the trends predicted by the analysis by studying two scenarios: a synthetic benchmark and a real example.
\end{abstract}

\section{Introduction}
%
%

%
Runtime Verification (RV)~\cite{LeuckerS09,RVTutorial,Bartocci2017} is a lightweight formal method which consists in verifying that a run of a system is correct {\wrt} a specification.
The specification formalizes the behavior of the system typically in logics (such as variants of Linear-Time Temporal Logic, LTL) or finite-state machines.
Typically the system is considered as a black box that feeds events to a monitor.
An event usually consists of a set of atomic propositions that describe some abstract operations or states in the system.
The sequence of events transmitted to the monitor is referred to as the trace.
Based on the received events, the monitor emits verdicts in a truth domain that indicate whether the run complies or not with the specification.
A typical truth domain is a set $\setof{\vt, \vf, \vna}$ where verdicts $\vt$ and $\vf$ indicate respectively that a program complies or violates the specification, and verdict $\vna$ indicates that no final verdict could be reached yet.
Truth domains can also include additional verdicts such as currently true and currently false, to indicate a finer grained truth value.
RV techniques have been used for instance in the context of decentralized automotive~\cite{ex:autosar} and medical~\cite{ex:medical} systems.
In both cases, RV is used to verify correct communication patterns between the various components and their adherence to the architecture and their formal specifications.
While RV comprehensively deals with monolithic systems, multiple challenges are presented when scaling existing approaches to decentralized systems, that is, systems with multiple components with no central observation point.
These challenges are inherent to the nature of decentralization; the monitors have a partial view of the system and need to account for communication and consensus.
Our assumptions on the system are as follows:
No monitors are malicious, i.e., messages do not contain wrong information;
No messages are lost, they are eventually delivered in their entirety but possibly out-of-order;
All components share one logical discrete time marked by round numbers indicating relevant transitions in the system specification.
\paragraph{Challenges.}
Several algorithms have been designed~\cite{BauerF12,FalconeCF14,BonakdarpourFRT16,DecentMon} and used~\cite{Bartocci13} to monitor decentralized systems.
Algorithms are primarily designed to address one issue at a time and are typically experimentally evaluated by considering runtime and memory overheads.
However, such algorithms are difficult to compare as they may combine multiple approaches at once.
For example, algorithms that use LTL rewriting~\cite{BauerF12,DecentMon,HavelundR01} not only exhibit variable runtime behavior due to the rewriting, but also incorporate different monitor synthesis approaches that separate the specification into multiple smaller specifications depending on the monitor.
%
In this case, we would like to split the problem of generating equivalent decentralized specifications from a centralized one (synthesis) from the problem of monitoring.
In addition, works on characterizing what one can monitor (i.e., monitorability~\cite{KimVBKLS99,PnueliZ06,FalconeFM12}) for centralized specifications exist~\cite{LTL3Tools,FalconeFM12,Diekert201429}, but do not extend to decentralized specifications.
For example by splitting an LTL formula ad-hoc, it is possible to obtain a non-monitorable subformula\footnote{We use the example from~\cite{DecentMon}: $\boldsymbol{\mathrm{GF}}(a) \land \neg (\boldsymbol{\mathrm{GF}}(a))$ (where $\boldsymbol{\mathrm{GF}}(a)$ means that $a$ should hold infinitely often) is monitorable, but its subformulas are both non-monitorable.} which interferes with the completeness of a monitoring algorithm.
\paragraph{Contributions.}
We tackle the presented challenges using two complementary approaches.
The first approach consists in using the data structure \emph{Execution History Encoding} (\MemRep{}) that encodes automata executions.
Since by using \MemRep{} one only needs to rewrite Boolean expressions, we are able to determine the parameters and their respective effect on the size of expressions, and fix upper bounds.
In addition, \MemRep{} is designed to be particularly flexible in processing, storing and communicating the information in the system.
\MemRep{} operates on an encoding of atomic propositions and guarantees strong-eventual consistency~\cite{CRDT}.
%
The second approach introduces decentralized specifications.
We introduce decentralized specifications,  define their semantics, interdependencies and study some of their properties.
We aim at abstracting the high-level steps of decentralized monitoring.
By identifying these steps, we elaborate a general decentralized monitoring algorithm.
We view a decentralized system as a set of components $\comps$.
A decentralized specification is thus as a set of $n$ finite-state automata with specific properties, which we call \emph{monitors}.
We associate $n$ monitors to these components with the possibility of multiple monitors being associated to a component.
Therefore, we generalize monitoring algorithms to multiple monitors.
Monitoring a centralized system can be seen as a special case with one component, one specification, and one monitor.
%
As such, we present a general decentralized monitoring algorithm that uses two high level steps: setup and monitor.
The setup phase creates the monitors, defines their dependencies and attaches them to components.
As such, the setup phases defines a \emph{topology} of monitors and their dependencies.
The monitor phase allows the monitors to begin monitoring and propagating information to reach a verdict when possible.
Therefore, the two high level operations help decompose monitoring into different subproblems and define them independently.
For example, the problem of generating a decentralized specification from a centralized specification is separated from checking the monitorability of a specification, and also separated from the computation and communication performed by the monitor.
We formulate and solve the problems of deciding \emph{compatibility} and \emph{monitorability} for decentralized specifications.
\emph{Compatibility} ensures that a monitor topology can be deployed on a given system,
\emph{monitorability} ensures that  given a specification, monitors are able to eventually emit a verdict, for all possible traces.
We present \THEMIS{}, a JAVA tool that implements the concepts in this paper; and show how it can be used to design and analyze new algorithms.
We use \THEMIS{} to create new metrics related to load-balancing and our data structures.
We use two scenarios to compare four existing algorithms.
The first scenario is a synthetic benchmark,  using random traces and specifications,
while the second scenario is a real example that uses the publish-subscribe pattern in the \emph{Chiron} graphical user interface system.
The synthetic scenario examines the trends of the analysis, and the \emph{Chiron} scenario examines  more specific differences in behavior.

%
This paper extends the work presented at the ACM SIGSOFT International Symposium on Software Testing and Analysis 2017~\cite{themisissta}, as follows:
\begin{itemize}
  \item adding the property that the \MemRep{} construction guarantees its determinism (\rprop{prop:ehe:invariant});
  \item elaborating and adding properties of decentralized specifications (monitorability, compatibility) as well as the algorithms for checking them (\rsec{sec:prop});
  \item improving \THEMIS{} by optimizing the \MemRep{} performance, and adding distributed and multi-threaded support (\rsec{sec:tool});
  \item {elaborating on the results and providing a discussion of the synthetic benchmarks (\rsec{sec:experiment:synthetic})};
  \item evaluating the algorithms on a new use case based on the \emph{Chiron} example that relies on publish-subscribe and has a formalized specification (\rsec{sec:chiron}); and
  \item extending related work (\rsec{sec:rw}) and formulating additional problems (\rsec{sec:fw}).
\end{itemize}
\paragraph{Overview.}
After presenting related work in \rsec{sec:rw}, we lay out the basic blocks, by introducing our basic data structure (\texttt{dict}), and the basic notions of monitoring with expressions in \rsec{sec:common}.
Then, we present our first approach, a middle ground between rewriting and automata evaluation by introducing the Execution History Encoding (\MemRep{}) data structure in \rsec{sec:cmon}.
We shift the focus on studying decentralized specifications by defining their semantics (\rsec{sec:dmon}), and their properties (\rsec{sec:prop}).
In \rsec{sec:analysis}, we use our analysis of \MemRep{} to study the behavior of three existing algorithms and discuss the situations that advantage certain algorithms over others.
In \rsec{sec:tool}, we present  the \THEMIS{} tool, which we use in \rsec{sec:experiment} to compare the algorithms presented in \rsec{sec:analysis} under two different scenarios: a synthetic random benchmark, and an example of a publish-subscribe system.
In \rsec{sec:fw}, we present future work and formulate additional interesting properties for decentralized specifications.
Finally, we conclude in \rsec{sec:conclusion}.
%
\section{Related Work}\label{sec:rw}
%
%
Several approaches have been taken to handle decentralized monitoring.
The first class of approaches consists in monitoring by rewriting formulae, the second class handles fault-tolerance, and the third class defines specifications for monitoring streams.
\paragraph{Formula rewriting.}
The first class of  approaches consists in monitoring by LTL formula rewriting~\cite{HavelundR01,BauerF12,DecentMon}.
Given an LTL formula specifying the system, a monitor will rewrite the formula based on information it has observed or received from other monitors, to generate a formula that has to hold on the next timestamp.
Typically a formula is rewritten and simplified until it is equivalent to $\vt$ (true) or $\vf$ (false) at which point the algorithm terminates.
Another approach~\cite{MTL} extends rewriting to focus on real-time systems. They use Metric Temporal Logic (MTL), which is an extension to LTL with temporal operators.
This approach also covers lower bound analysis on monitoring MTL formulae.
While these techniques are simple and elegant, rewriting varies significantly during runtime based on observations, thus analyzing the runtime behavior could prove difficult if not unpredictable.
For example, when excluding specific syntactic simplification rules, $\boldsymbol{\mathrm{G}}(\vt)$ could be rewritten $\vt \land \boldsymbol{\mathrm{G}}(\vt)$ and will keep growing in function of the number of timestamps.
To tackle the unpredictability of rewriting LTL formulae, another approach~\cite{FalconeCF14} uses automata for monitoring regular languages, and therefore (i) can express richer specifications, and (ii) has predictable runtime behavior.
These approaches use a centralized specification to describe the system behavior.

\paragraph{Fault-tolerant monitoring.}
Another class of research focuses on handling a different problem that arises in distributed systems.
In~\cite{BonakdarpourFRT16}, monitors are subject to many faults such as failing to receive correct observations or communicate state with other monitors.
Therefore, the problem handled is that of reaching consensus with fault-tolerance, and is solved by determining the necessary verdict domain needed to be able to reach a consensus.
To remain general, we do not impose the restriction that all monitors must reach the verdict when it is known, as we allow different specifications per monitor.
Since we have heterogeneous monitors, we are not particularly interested in consensus.
However, for multiple monitors tasked to monitor the same specification, we are interested in strong eventual consistency.
We maintain the 3-valued verdict domain, and tackle the problem from a different angle by considering eventual delivery of messages.
Similar work~\cite{FAMTL} extends the MTL approach to deal with failures by modeling knowledge gaps and working on resolving these gaps.
We also highlight that the mentioned approaches~\cite{BauerF12,FAMTL,DecentMon}, and other works~\cite{DistVolker,DistSen,DistSchmitz} do in effect define separate monitors with different specifications, typically consisting in splitting the formula into subformulas.
Then, they describe the collaboration between such monitors.
However, they primarily focus on presenting one global formula of the system from which they derive  multiple specifications.
In our approach, we generalize the notions from a centralized to a decentralized specification, and separate the problem of generating multiple specifications equivalent to a centralized specification from the monitoring of a decentralized specification (\rsec{sec:prop:synthesis}).

\paragraph{Specifications over streams.}
Specification languages have been developed that monitor synchronous systems as streams~\cite{LOLA,TESSLA}.
In this setting, events are grouped as a stream, and streams are then aggregated by various operators.
The output domain extends beyond the Boolean domain and encompasses types.
The stream approach to monitoring has the advantage of aggregating types, as such operations such as summing, averaging or pulling statistics across multiple streams is also possible.
Stream combination is thus provided by general-purpose functions, which are more complex to analyze than automata.
This is similar to complex event processing where RV is a special case~\cite{CEPRV}.
Specification languages such as LOLA~\cite{LOLA} even define dependency graphs between various stream information, and have some properties like \emph{well formed}, and \emph{efficiently monitorable} LOLA specifications.
The former ensures that dependencies in the trace can be resolved before they are needed, and the latter ensures that the memory requirement is no more than constant with respect to the length of the trace.
While streams are general enough to express monitoring, they do not address decentralized monitoring explicitly.
As such, there is no explicit assignment of monitors to components and parts of the system, nor consideration of architecture.
Furthermore, there is no algorithmic consideration addressing monitoring in a decentralized fashion, even-though some works such as~\cite{BEEPBEEP3} do provide multi-threaded implementations.

\section{Common Notions} \label{sec:common}
We begin by introducing the dict data structure (\rsec{sec:common:dict}.) used to build more complex data structures, and defining the basic concepts for decentralized monitoring (\rsec{sec:common:mon}).
%
\subsection{The \vars{dict} Data Structure}
\label{sec:common:dict}
%
In monitoring decentralized systems, monitors typically have a state, and attempt to merge other monitor states with theirs to maintain a consistent view of the running system, that is, at no point in the execution, should two monitors receive updates that conflict with one another.
We would like in addition, that any two monitors receiving the same information be in equivalent states.
Therefore, we are interested in designing data structures that can replicate their state under strong eventual consistency (SEC)~\cite{CRDT}, they are known as  state-based convergent replicated data-types (CvRDTs).
We  use a dictionary data structure (noted \vars{dict}) as our basic building block that assigns a value to a given key.
{
 Data structure \vars{dict} will be used to define the memory of a monitor (\rsec{sec:common:mon}), and data structure \MemRep{} which encodes the execution of an automaton (\rsec{sec:cmon:ehe}).
}%

We model \vars{dict} as a partial function $\func{f}$.
{The domain of $\func{f}$ (denoted by $\fdom(\func{f})$) is the set of keys, while the codomain of $\func{f}$ (denoted by $\fcodom(\func{f})$) is the set of values.}
\vars{dict} supports two operations: $\vars{query}$ and $\vars{merge}$.
The $\vars{query}$ operation checks if a key $k \in \fdom(\func{f})$ and returns $\func{f}(k)$.
If $k \not\in \fdom(\func{f})$, then it is undefined.
The $\vars{merge}$ operation of a \vars{dict} $\func{f}$ with another \vars{dict} $\func{g}$, is modeled as function composition.
Two partial functions $\func{f}$ and  $\func{g}$ are composed using operator $\dagger_{\mathit{op}}$ where $\mathit{op} : (\fdom(\func{f}) \times \fdom(\func{g})) \rightarrow (\fcodom(\func{f}) \cup \fcodom(\func{g}))$ is a binary function.
\begin{align*}
	\func{f} \dagger_{\mathit{op}} \func{g} 	 &: \fdom(\func{f}) \cup \fdom(\func{g}) \rightarrow \fcodom(\func{f}) \cup \fcodom(\func{g})\\
		\func{f} \dagger_{\mathit{op}} \func{g} (\xv) 					 &= \funcparts{
		op(\func{f}(\xv), \func{g}(\xv))	 	 & \pif \xv \in 		\fdom(\func{f}) \cap \fdom(\func{g})\\
		\func{g}(\xv)				 & \pif \xv \in \fdom(\func{g}) \setminus \fdom(\func{f})\\
		\func{f}(\xv)				 & \pif \xv \in		\fdom(\func{f}) \setminus \fdom(\func{g})\\
		\fundef				 & \pelse
}
\end{align*}
On sets of functions, $\dagger_{\mathit{op}}$ applies pairwise: $\biguplus^{\mathit{op}}\setof{\func{f_1}, \hdots \func{f_n}} =$ $((\func{f_1} \dagger_{op} \func{f_2}) \hdots \func{f_n})$.
The following two operators are used in the rest of the paper: $\memadd$ and $\sysadd$.
We define both of these operators to be commutative, idempotent, and associative to ensure SEC.
\begin{align*}%
  \memadd(\xv,\xv') = \funcparts{
  \xv' & \pif \xv \prec \xv'\\
  \xv  & \xelse \\
} & \quad
\sysadd(\xv, \xv') = \xv \lor \xv'
\end{align*}
Operator $\memadd$ acts as a replace function based on a total order ($\prec$) between the elements, so that it always chooses the highest element to guarantee idempotence, while $\sysadd$ uses the logical \emph{or} operator to combine elements.
Respectively, we denote the associated pairwise set operators by $\biguplus^{2}$ and $\biguplus^{\lor}$.
Data structure $\vars{dict}$ can be composed by only using operation $\vars{merge}$.
The modifications never remove entries, the state of \vars{dict} is then monotonically increasing using the order provided by $\vars{merge}$.
By ensuring that $\vars{merge}$ is idempotent, commutative, and associative we fulfill the necessary conditions~\cite{CRDT} for our data structure to be a CvRDT.
\begin{proposition}
\label{prop:common:dict}
Data structure \vars{dict} with operations $\vars{query}$ and $\vars{merge}$ is a CvRDT.
\end{proposition}

%
\subsection{Basic Monitoring Concepts} \label{sec:common:mon}
%
We recall the basic building blocks of monitoring.
We consider the set of verdicts $\verdict = \setof{\vt, \vf, \vna}$ to denote the verdicts true, false, not reached (or inconclusive) respectively.
{A verdict in $\verdictb = \setof{\vt, \vf}$ is a \emph{final} verdict.
It indicates that the monitor has concluded its monitoring, and any further input will not change affect it.
Abstract states of a system are represented as a set of \emph{atomic propositions} ($\AP$).
A monitoring algorithm typically includes additional information such as a timestamp associated with the atomic propositions.
We capture this information as an encoding of the atomic propositions ($\atoms$)}, this encoding is left to the monitoring algorithm to specify.
$\expr_\atoms$ (resp. $\expr_\AP$) denotes the set of Boolean expressions over $\atoms$ (resp. $\AP$).
When omitted, $\expr$ refers to $\expr_\atoms$.
An encoder is a function $\senc : \expr_\AP \rightarrow \expr_\atoms$ that encodes the atomic propositions into atoms.
In this paper, we use two encoders: $\fid$ which is the identity function (it does not modify the atomic proposition), and $\fts_t$ which adds a timestamp $t$ to each atomic proposition.
A decentralized monitoring algorithm requires retaining, retrieving and communicating observations.
\begin{definition}[Event]
\label{def:common:event}
	An observation is a pair in $\AP \times \verdictb$ indicating whether or not a proposition has been observed. An event is a set of observations in $2^{\AP \times \verdictb}$.
\end{definition}
\begin{example}[Event]\label{ex:common:event}
 Event $\setof{\tuple{a, \vt}, \tuple{b, \vf}}$ over $\setof{a,b}$ indicates that proposition $a$ has been observed to be true, while $b$ has been observed to be false.
\end{example}
\begin{definition}[Memory]
\label{def:common:mem}
A memory is a \vars{dict}, and is modeled as a partial function $\mem : \atoms \rightarrow \verdict$ that associates an atom to a verdict.
The set of all memories is defined as $\Mem$.
\end{definition}
A monitor stores its events in a memory with some encoding (e.g.,  adding a timestamp).
An event can be converted to a memory by encoding the atomic propositions to atoms, and associating their truth value: $\cons: 2^{\AP \times \verdictb} \times (\expr_\AP \rightarrow \expr_\atoms) \rightarrow \Mem$.
\begin{example}[Memory]\label{ex:common:mem}
Let $\mathit{evt}$ $=$ $\setof{\tuple{a, \vt}, \tuple{b, \vf}}$ be an event at $t = 1$, the resulting memories using encodes $\fid$ and $\fts_1$ are:
$
	\cons(\mathit{evt}, \fid) = [a \mapsto \vt, b \mapsto \vf],
	\cons(\mathit{evt}, \fts_1) = [\tuple{1, a} \mapsto \vt, \tuple{1, b} \mapsto \vf].
$
\end{example}
If we impose that $\atoms$ be a totally ordered set, then two memories $\mem_1$ and $\mem_2$ can be merged by applying operator $\memadd$.
The total ordering is needed for operator $\memadd$.
This ensures that the operation is idempotent, associative and commutative.
Monitors that exchange their memories  and merge them  have a consistent snapshot of the memory, regardless of the ordering.
Since a memory is a \vars{dict} and $\memadd$ is idempotent, associative, and commutative, it follows from \rprop{prop:common:dict} that a memory is a CvRDT.
\begin{corollary}\label{prop:mem-crdt}
A memory with operation $\memadd$ is a CvRDT.
\end{corollary}
In this paper, we perform monitoring by manipulating expressions in $\expr$.
The first operation we provide is $\rw$, which rewrites the expression to attempt to eliminate $\atoms$.
\begin{definition}[Rewriting an expression]
\label{def:common:rw}
An expression $\mathit{e}$ is rewritten with a memory $\mem$ using function $\rw: \expr \times \Mem \rightarrow \expr$ defined as follows:
\begin{align*}
	\rw(\mathit{e}, \mem) &= \texttt{ match $\mathit{e}$ with }\\
				 & \begin{array}{ll}
				\mid a \in \atoms  & \rightarrow \funcparts {
						\mem(a)		& \pif a \in dom(\mem)\\
						a			& \pelse
				}\\
				\mid \neg e'		   & \rightarrow \neg \rw(e', \mem)\\
				\mid e_1 \land e_2 & \rightarrow \rw(e_1,\mem) \land \rw(e_2, \mem)\\
				\mid e_1 \lor e_2 & \rightarrow \rw(e_1,\mem) \lor \rw(e_2, \mem)
\end{array}
\end{align*}
\end{definition}
Using information from a memory $\mem$, the expression is rewritten by replacing atoms with a final verdict (a truth value in $\verdictb$) in $\mem$ when possible.
Atoms that are not associated with a final verdict are kept in the expression.
Operation $\rw$ yields a smaller formula to work with.
\begin{example}[Rewriting]
\label{ex:common:rw}
We consider $\mem = [a \mapsto \vt, b \mapsto \vf];$ and $e = (a \lor b) \land c$.
We have $\mem(a) = \vt$, $\mem(b) = \vf$, $\mem(c) = \vna$.
Since $c$ is associated with $\vna \not\in \verdictb$ then it will not be replaced when the expression is evaluated.
The resulting expression is $\rw(e,\mem) = (\vt \lor \vf)\land c$.
\end{example}
We eliminate additional atoms using Boolean logic.
We denote by $\fsimp(\vars{expr})$ the simplification of expression $\vars{expr}$~\footnote{This is also known as The Minimum Equivalent Expression problem~\cite{CircuitMin}.}.
\begin{example}[Simplification] \label{ex:common:simplify}
Consider $\mem = [a \mapsto \vt]$ and  $e = (a \land b) \lor (a \land \neg b)$.
We have $e' = \rw(e, \mem) = (b \lor \neg b)$.
Atoms can be eliminated with $\fsimp(e')$.
We finally get $\vt$.
\end{example}
We combine both rewriting and simplification in the $\seval$ function which determines a verdict from an expression $\mathit{e}$.
\begin{definition}[Evaluating an expression]
\label{def:common:eval}
The evaluation of a Boolean expression $\mathit{e} \in \expr$ using a memory $\mem$ yields a verdict. Function $\seval: \expr \times \Mem \rightarrow \verdict$ is defined as:
\begin{align*}
	\seval (\mathit{e}, \mem) &= \funcparts{
      \vt & \pif \fsimp(\rw(\mathit{e}, \mem) \Leftrightarrow \vt,\\
      \vf & \pif \fsimp(\rw(\mathit{e}, \mem) \Leftrightarrow \vf,\\
      \vna& \pelse.
   }
\end{align*}
\end{definition}
Function $\seval$ returns the verdict $\vt$ (resp. $\vf$)  if the simplification after rewriting is (Boolean) equivalent to  $\vt$ (resp. $\vf$), otherwise it returns verdict $\vna$.
\begin{example}[Evaluating expressions]\label{ex:common:eval}
Consider $\mem = [a \mapsto \vt, b \mapsto \vf]$ and $e = (a \lor b) \land c$.
We have $\fsimp(\rw(e, \mem)) = \fsimp((\vt \lor \vf) \land c) =  c$, and $\seval(e,\mem) = \vna$ which depends on $c$: we cannot emit a final verdict before observing $c$.
\end{example}
A decentralized system is a set of components $\comps$.
We assign a sequence of events to each component using a decentralized trace function.
\begin{definition}[Decentralized trace]
\label{def:common:dtrace}
A decentralized trace of length $n$ is a function $\trace : [1, n] \times \comps \rightarrow 2^{\AP \times \verdictb}$ (where $[1, n]$ denotes the interval of the $n$ first non-zero natural numbers).%
\end{definition}
Function $\trace$ assigns an event to a component for a given timestamp.
We denote by $\traces$ the set of all possible decentralized traces.
We additionally define function $\clookup : \AP \rightarrow \comps$ to assigns an atomic proposition to a component\footnote{We assume that (1) no two components can observe the same atomic propositions, and (2) a component has at least one observation at all times (a component with no observations to monitor, can be simply considered excluded from the system under monitoring).}.
The function $\clookup$ is defined as
$\clookup(ap) = c \mbox { s.t. } \exists t \in \timestamp, \exists v \in \verdictb : \tuple{ap, v} \in \trace(t, c)$.

We consider timestamp 0 to be associated with the initial state, therefore our traces start at 1.
The length of a trace tr is denoted by $|\trace|$.
An empty trace has length 0 and is denoted by $\emptytrace$.
Monitoring using LTL or finite-state automata relies on sequencing the trace.
Events must be totally ordered.
A timestamp indicates simply the order of the sequence of events.
As such, a timestamp represents a logical time, it can be seen as a \emph{round number}.
Every round consists in a transition taken on the automaton after reading a part of the word.
While $\trace$ gives us a view of what components can locally see, we reconstruct the global trace to reason about all observations.
A global trace of the system is therefore a sequence of events.
A global trace encompasses all observations observed locally by components.
\begin{definition}[Reconstructing a global trace]
\label{def:common:gtrace}
Given a decentralized trace $\trace$ of length $n$, we reconstruct the global trace using function $\rho : \left([1, n] \times \comps \rightarrow 2^{\AP \times \verdictb}\right) \rightarrow \left([1, n] \rightarrow 2^{\AP \times \verdictb}\right)$ defined as $\rho(\trace) = \mathit{evt}_1 \cdot \hdots \cdot \mathit{evt}_{n}$ s.t. $\forall i \in [1, n] : \mathit{evt}_i = \bigcup_{c \in \comps} \trace(i, c)$.
\end{definition}
For each timestamp $i \in [1, n]$, we take all observations of all components and union them to get a global event.
Consequently, an empty trace yields an empty global trace, $\rho(\emptytrace) = \emptytrace$.
\begin{example}[Traces] \label{ex:traces}
	We consider a system of two components $A$ and $B$, that are associated with atomic propositions $a$ and $b$ respectively.
	An example decentralized trace of the system is given by $\trace = [1 \mapsto A \mapsto \setof{\tuple{a, \vt}}, 1 \mapsto B \mapsto \setof{\tuple{b, \vt}}, 2 \mapsto A \mapsto \setof{\tuple{a, \vt}}, 2 \mapsto B \mapsto \setof{\tuple{b, \vf}}]$.
	That is, component $A$ observes  proposition $a$ to be $\vt$ at both timestamps 1 and 2, while $B$ observes $b$ to be $\vt$ at timestamp 1 and $\vf$ at timestamp 2.
	The associated global trace is: $\rho(\trace) = \setof{\tuple{a, \vt}, \tuple{b, \vt}} \cdot \setof{\tuple{a, \vt}, \tuple{b, \vf}}$.
\end{example}
\section{Centralized Specifications}
\label{sec:cmon}
%
%
We now focus on a decentralized system specified by one global automaton.
We consider automata that emit 3-valued verdicts in the domain $\verdict$, similar to those in~\cite{DecentMon,LTL3Tools} for centralized systems.
{Using automata with 3-valued verdicts has been the topic of a lot of the Runtime Verification literature~\cite{LTL3Tools,BauerF12,DecentMon,FalconeCF14,FAMTL}}, we focus on extending the approach for decentralized systems in~\cite{DecentMon} to use a new data structure called Execution History Encoding ($\MemRep{}$).
Typically, monitoring is done by labeling an automaton with events, then playing the trace on the automaton and determining the verdict based on the reached state.
We present the $\MemRep{}$, a data structure that encodes the necessary information from an execution of the automaton.
Monitoring using EHEs ensures strong eventual consistency.
We begin by defining the specification automaton used for monitoring in \rsec{sec:cmon:pre}, then we present the $\MemRep{}$ data structure, illustrate its usage for monitoring in \rsec{sec:cmon:ehe}, and describe its use to reconcile partial observations in \rsec{sec:cmon:reconcile}.
%
\subsection{Preliminaries} \label{sec:cmon:pre}
%
Specifications are similar to the Moore automata generated by~\cite{LTL3Tools}.
We modify labels to be Boolean expressions over atomic propositions (in $\expr_\AP$).
{We choose to label the transitions with Boolean expressions as opposed to events, to  keep a homogeneous representation (with \MemRep{})}\footnote{Indeed, an event can be converted to an expression by the conjunction of  all observations, negating the terms that are associated with the verdict $\vf$.}.
\begin{definition}[Specification] \label{def:cmon:aut}
 The specification is a deterministic Moore automaton $\tuple{Q, q_0, \delta, \verdictf}$ where $q_0 \in Q$ is the initial state, $\delta : Q \times \expr_\AP \rightarrow Q$ is the transition function and $\verdictf : Q \rightarrow \verdict$ is the labeling function.
\end{definition}
The labeling function associates a verdict with each state.
When using multiple automata we use labels to separate them, $\aut_\ell = \tuple{Q_\ell, q_{\ell_0}, \delta_\ell, \verdictf_\ell}$.
We fix $\aut$ to be a specification automaton for the remainder of this section.
For monitoring, we are interested in events (\rdef{def:common:event}), we extend $\delta$ to events, and denote it by $\Delta$\footnote{We note that in this case, we are not using any encoding ($\atoms = \AP$).}.
\begin{definition}[Transition over events]
\label{def:cmon:aut-semantics}
	Given an event $\mathit{evt}$, we build the memory $\mem =\cons(\mathit{evt}, \fid)$.
	Then, function $\Delta : Q \times 2^{\AP \times \verdictb} \rightarrow Q$ is defined as follows:
	\begin{align*}
		\Delta(q, \mathit{evt}) &= \funcparts[ll]{
			q'    & \mbox{ if } \mathit{evt} \neq \emptyset \land \exists q' \in Q, \exists e \in \expr_\AP:
                               \delta(q, e) = q' \land \seval(e, \mem) = \vt,  \\
							q   & \mbox{ otherwise}.
		}
  \end{align*}
\end{definition}
A transition is taken only when an event contains observations (i.e., $\mathit{evt} \neq \emptyset$).
This allows the automaton to wait on observations before evaluating, as such it remains in the same state (i.e., $\Delta(q, \emptyset) = q$).
Upon receiving observations, we use $\mem$ to evaluate each label of an outgoing transition, and determine if a transition can be taken (i.e., $ \exists q' \in Q, \exists e \in \expr_\AP: \delta(q, e) = q' \land \seval(e, \mem) = \vt$).

To handle a trace, we extend $\Delta$ to its reflexive and transitive closure in the usual way, and note it $\Delta^*$.
For the empty trace, the automaton makes no moves, i.e.,  $\Delta^*(\emptytrace) = q_{0}$.
	\begin{figure}[t] %
		\centering%
		\scalebox{1}{\begin{tikzpicture}[aut]
	\node[location] (q0) at (0,0) {$q_0$};
	\node[location, accept, right=of q0] (q1) {$q_1$};

	\tconnect{q0}{q1}{$a \lor b$}{t1}
	\tconnect[loop right]{q1}{q1}{$\vt$}{t2}
	\tconnect[loop left]{q0}{q0}{$\neg a \land \neg b$}{t2}
\end{tikzpicture}} %
		\vspace{-1em}
		\caption{Representing $\mathrm{F}(a \lor b)$}%
		\label{fig:cmon:aut}%
		\vspace{-1em}
	\end{figure}%

\begin{example}[Monitoring using expressions]
\label{ex:cmon:aut}
We consider $\atoms = \AP = \setof{a, b}$ and the specification in \rfig{fig:cmon:aut}, we seek to monitor $\mathrm{F}(a \lor b)$.
The automaton consists of two states: $q_0$ and $q_1$ associated respectively with the verdicts $\vna$ and $\vt$.
We consider  at $t = 1$ the event $\mathit{evt} = \setof{\tuple{a, \vt}, \tuple{b, \vf}}$.
The resulting memory is $\mem = [a \mapsto \vt, b \mapsto \vf]$ (see \rex{ex:common:mem}).
The transition from $q_0$ to $q_1$ is taken since $\seval(a \lor b, \mem) = \vt$.
Thus we have $\Delta(q_0, e) = q_1$ with  verdict $\verdictf(q_1) = \vt$.
\end{example}
\begin{remark}[Properties and normalization]\label{rmk:normalized-aut}
We recall that the specification is a deterministic and complete automaton.
Hence, there are properties on the expressions that label the transition function.
For any $q \in Q$, we have:
\begin{enumerate}
  \item $
      \forall \mem \in \Mem:
        (\exists \tuple{q, e} \in \fdom(\delta): \seval(e, \mem) = \vt)
        \implies
        (\not \exists \tuple{q, e'} \in \fdom(\delta) \setminus \setof{\tuple{q,e}}: \seval(e', \mem) = \vt)
  $; and
  \item the disjunction of the labels of all outgoing transitions results in an expression that is a tautology.
\end{enumerate}
The first property states that for all possible memories encoded with $\fid$ no two (or more) labels can evaluate to $\vt$ at once.
It results from determinism: no two (or more) transitions can be taken at once.
The second property results from completeness: given any input, the automaton must be able to take a move.
Furthermore, we note that for each pair of states $\tuple{q,q'} \in Q \times Q$, we can rewrite $\delta$ such that there exists at most one expression $e \in \expr_\AP$, such that $\delta(q, e) = q'$, without loss of generality.
This is because for a pair of states, we can always disjoin the expressions to form only one expression, as it suffices that only one expression needs to evaluate to $\vt$ to reach $q'$.
By having at most one transition between any pair of states, we simplify the topology of the automaton.
\end{remark}
%
\subsection{Execution History Encoding}
\label{sec:cmon:ehe}
%
The execution of the specification automaton, is in fact, the process of monitoring, upon running the trace, the reached state determines the verdict.
An execution of the specification automaton can be seen as a sequence of states $q_0 \cdot q_1 \cdot \hdots q_t \hdots$.
It indicates that for each timestamp $t \in [0, \infty[$ the automaton is in the state $q_t$.
In a decentralized system, a component receives only local observations and does not necessarily have enough information to determine the state at a given timestamp.
Typically, when sufficient information is shared between various components, it is possible to know the state $q_t$ that is reached in the automaton at $t$ (we say that the state $q_t$ has been found, in such a case).
{The aim of the $\MemRep{}$ is to construct a data structure which follows the current state of an automaton, and in case of partial information, tracks the possible states the automaton can be in.
For that purpose, we need to ensure strong eventual consistency in determining the state $q_t$ of the execution of an automaton.}
That is, after two different monitors share their $\MemRep{}$, they should both be able to find  $q_t$ for  $t$ (if there exists enough information to infer the global state), or if not enough information is available, they both find no state at all.
\begin{definition}[Execution History Encoding - \MemRep{}] \label{def:cmon:ehe}
An Execution History Encoding (\MemRep{}) of the execution of an automaton $\aut$ is a partial function $\sys : \timestamp \times Q \rightarrow \expr$.
\end{definition}
For a given execution, we encode the conditions to be in a state at a given timestamp as an expression in $\expr$.
$\sys(t,q)$ is an expression used to track whether the automaton is in state $q$ at $t$. Using information from the execution stored in a memory $\mem$, if $\seval(\sys(t,q),\mem)$ is $\vt$, then we know that the automaton is indeed in state $q$ at timestamp $t$.
We use the notation $\timestamps{\sys}$, to denote all the timestamps that the \MemRep{} encodes, i.e., $\timestamps{\sys} = \setof{t \in \timestamp \mid \tuple{t,q} \in \fdom(\sys)}$.
Similarly to automata notation, if multiple \MemRep{}s are present, we use a label in the subscript to identify them and their respective operations ($\sys_\ell$ denotes the \MemRep{} of $\aut_\ell$).

To compute $\sys$ for a timestamp range, we will next define some (partial) functions: $\sreach$, $\verdictAt$, $\snext$, $\sto$, and $\smove$.
The purpose of these functions is to extract information from $\sys$ at a given timestamp, which we can use to recursively build $\sys$ for future timestamps.
Given a memory $\mem$ which stores atoms, function $\sreach$ determines if a state is reached at a timestamp $t$.
If the memory does not contain enough information to evaluate the expressions, then the state is $\fundef$.
The state $q$ at timestamp $t$ with a memory $\mem$ is determined by:
\[
	\sreach(\sys, \mem, t) = \funcparts{
				q & \pif \exists q \in Q: \seval(\sys(t,q), \mem) = \vt,\\
				\fundef & \pelse.\\
}\]
Function $\verdictAt$ is a short-hand to retrieve the verdict at a given timestamp $t$:
\[
	\verdictAt(\sys, \mem, t) = \funcparts{
			\verdictf(q)  & \pif \exists q \in Q: q = \sreach(\sys,\mem,t),\\
			\vna		  & \pelse.
}
\]
The automaton is in the first state at $t = 0$.
We start building up $\sys$ with the initial state and associating it with expression $\vt$: $\sys = [0 \mapsto q_0 \mapsto \vt]$.
Then for a given timestamp $t$, we use function $\snext$ to check the next possible states in the automaton by looking at the outgoing transitions for all states in $\sys$ at $t$:
\[
	\snext(\sys, t) = \setof{q'  \in Q \mid \exists \tuple{t,q} \in \fdom(\sys), \exists e \in \expr : \delta(q, e) = q'}.
\]
We now build the necessary expression to reach $q'$ from multiple states $q$ by disjoining the transition labels.
Since the label consists of expressions in $\expr_\AP$ we use an encoder to get an expression in $\expr_\atoms$
To get to the state $q'$ at $t+1$ from $q$ we conjunct the condition to reach $q$ at $t$.
\[
	\sto(\sys, t, q', \senc) = \smashoperator{\bigvee_{\setof{\tuple{q,e'} \mid  \> \delta(q, e') = q'}}}( \sys(t,q) \land \senc(e')  )
\]
By considering the disjunction, we cover all possible paths to reach a given state.
Updating the conditions for the same state on the same timestamp is done by disjoining the conditions:
\[
	\smove(\sys, t_s, t_e) = \funcparts{
		\smove(\sys', t_s + 1, t_e) & \mbox{ if } t_s < t_e,\\
		\sys & \mbox{ otherwise},\\
	}
\]
with $\sys' = \sys \sysadd \smashoperator{\biguplus\limits^{\lor}_{q' \in \snext(\sys,t_s)}} \setof{t_s + 1 \mapsto q' \mapsto \sto(\sys,t_s,q', \fts_{t_s + 1})}$.

Finally, $\sys'$ is obtained by considering the next states and merging all their expressions with $\sys$.
We use superscript to denote the encoding up to a given timestamp $t$ as $\sys^t$.

\gettable{
\small
\centering
\caption{A tabular representation of $\sys^2$}
\label{tbl:cmon:exec}
\begin{tabular}{| r | c | l |}
	\hline \textbf{t} & \textbf{q} & \textbf{e} \\
	\hline
	\hline 0 		  & $q_0$	   & $\vt$\\
	\hline 1		  & $q_0$	   & $\neg \tuple{1, a} \land \neg \tuple{1,b}$\\
	\hline 1 		  & $q_1$	   & $\tuple{1,a} \lor \tuple{1,b}$\\
	\hline 2		  & $q_0$	   & $(\neg \tuple{1, a} \land \neg \tuple{1,b}) \land (\neg \tuple{2, a} \land \neg \tuple{2,b}) $\\
	\hline 2		  & $q_1$	   & $(\tuple{1,a} \lor \tuple{1,b}) \lor ((\neg \tuple{1, a} \land \neg \tuple{1,b}) \land (\tuple{2,a} \lor \tuple{2,b}))$\\
	\hline
\end{tabular}
}

\begin{example}[Monitoring with $\MemRep$] \label{ex:cmon:exec}
We encode the execution of the automaton presented in \rex{ex:cmon:aut}.
{For this example, we use the encoder $\fts_n$ which appends timestamp $n$ to an atomic proposition.}
We have $\sys^0 = [0 \mapsto q_0 \mapsto \vt]$.
From $q_0$, it is possible to go to $q_0$ or $q_1$, therefore $\snext(\sys^0, 0) = \setof{q_0, q_1}$.
To move to $q_1$ at $t = 1$, we must be at $q_0$ at $t = 0$.
The following condition must hold: $\sto(\sys^0, 0, q_1, \fts_1) = \sys^0(0, q_0) \land (\tuple{1, a} \lor \tuple{1,b}) = \tuple{1,a} \lor \tuple{1,b}$.
The encoding up to timestamp $t=2$ is obtained with $\sys^2 = \smove(\sys^0, 0, 2)$ and is shown in \rtbl{tbl:cmon:exec}.
We consider the same event as in \rex{ex:common:mem} at $t = 1$, $\mathit{evt} = \setof{\tuple{a, \vt}, \tuple{b, \vf}}$.
Let $\mem = \cons(\mathit{evt}, \fts_1)  = [\tuple{1, a} \mapsto \vt, \tuple{1,b} \mapsto \vf]$.
It is possible to infer the state of the automaton after computing only $\sys^1 = \smove(\sys^0, 0, 1)$ by using $\sreach(\sys^1, \mem, 1) $, we evaluate:
\[
\begin{array}{lll}
	\seval(\sys^1(1, q_0), \mem) &= \neg \tuple{1, a} \land \neg \tuple{1,b}  &= \vf \\
	\seval(\sys^1(1, q_1), \mem) &= \tuple{1,a} \lor \tuple{1,b}  &= \vt
\end{array}
\]
We find that $q_1$ is the selected state, with verdict $\verdictf(q_1) = \vt$.
\end{example}
Since we are encoding deterministic automata, we recall from Remark~\ref{rmk:normalized-aut} that when a state $q$ is reachable at a timestamp $t$, no other state is reachable at $t$.
Moreover, the \MemRep{} construction using operation $\smove$ and encoder $\fts$ preserves determinism.

\begin{proposition}[Deterministic \MemRep{}]\label{prop:ehe:invariant}
 Given an \MemRep{} $\sys$ constructed with operation $\smove$ using encoder $\fts$, we have:
 \[
 \forall t \in \timestamps{\sys}, \forall \mem \in \Mem, \exists q \in Q: (\seval(\sys(t,q)) = \vt, \mem) \implies (\forall q' \in Q \setminus \{ q \} : \seval(\sys(t,q')) \neq \vt, \mem).
 \]
\end{proposition}

Determinism is preserved since, by using encoder $\fts$, we only change an expression to add the timestamp.
By construction, when there exists a state $q$ s.t. $\seval(\sys(t,q), \mem) = \vt$, such a state is unique, since the \MemRep{} is built using a deterministic automaton.
The full proof is in Appendix~\ref{sec:app:proofs}.

While the construction of a \MemRep{} preserves the determinism found in the automaton, an important property is in ensuring that the \MemRep{} encodes correctly the the automata execution.

\begin{proposition}[Soundness]\label{prop:cmon:soundness}
Given a decentralized trace $\trace$ of length $n$, we reconstruct the global trace $\rho(\trace) = \mathit{evt}_1 \cdot \hdots \cdot \mathit{evt}_n$, we have: $\Delta^*(q_0, \rho(\trace)) = \sreach(\sys^n, \mem^n, n)$, with:\\
$
\begin{array}{rl}
\quad \sys^n &= \smove([0 \mapsto q_0 \mapsto \vt], 0, n), \mbox{and} \\
\quad \mem^n &= \biguplus^{2}_{t \in [1, n]} \setof{\cons(\mathit{evt}_{t}, \fts_{t})}.
\end{array}
$
\end{proposition}
\MemRep{} is sound {\wrt} the specification automaton; both the automaton and $\MemRep{}$ will indicate the same state reached with a given trace.
Thus, the verdict is the same as it would be in the automaton.
The proof is by induction on the reconstructed global trace ($|\rho(\trace)|$). 
\paragraph{Proof sketch.}
We first establish that both the \MemRep{} and the automaton memories evaluate two similar expressions modulo encoding to the same result.
That is, for the given length $i$, the generated memories at $i+1$ with encodings $\fid$ and $\fts_{i+1}$ yield similar evaluations for the same expression $\mathit{e}$.
Then, starting from the same state $q_i$ reached at length $i$, we assume $\Delta^*(q_0, \mathit{evt}_1 \cdot \hdots \cdot \mathit{evt}_i) = \sreach(\sys^i, \mem^i, i) = q_i$ holds.
We prove that it holds at $i+1$, by building the expression (for each encoding) to reach state $q_{i+1}$ at $i+1$, and showing that the generated expression is the only expression that evaluates to $\vt$.
As such, we determine that both evaluations point to $q_{i+1}$ being the next state.
The full proof is in Appendix~\ref{sec:app:proofs}.
%
\subsection{Reconciling Execution History} \label{sec:cmon:reconcile}
%
\MemRep{} provides interesting properties for decentralized monitoring.
Since the data structure $\MemRep{}$ is a reification of \vars{dict}, as it maps a pair in $\timestamp \times Q$ to an expression in $\expr$, and since the combination is done using $\sysadd$, which is idempotent, commutative and associative,  it follows from \rprop{prop:common:dict} that  $\MemRep{}$ is a CvRDT.
\begin{corollary}
\label{prop:ehe-crdt}
An \MemRep{} with operation $\sysadd$ is a CvRDT.
\end{corollary}
Two (or more) components sharing $\MemRep{}$s and merging them will be able to infer the same execution history of the automaton.
That is, components will be able to aggregate the information of various $\MemRep{}$s, and are able to determine the reached state, if possible, or that no state was reached.
Merging two $\MemRep{}$s of the same automaton with $\sysadd$ allows us to aggregate information from two partial histories.

However, two $\MemRep{}$s for the same automaton contain the same expression if constructed with $\smove$.
To incorporate the memory in a \MemRep{}, we generate a new \MemRep{} that contains the rewritten and simplified expressions for each entry.
To do so we define function $\finc$ to apply to a whole \MemRep{} and a memory to generate a new \MemRep{}: $\finc(\sys, \mem) = \biguplus^{2}_{\tuple{t,q} \in \fdom(\sys)} \setof{[\tuple{t,q} \mapsto \fsimp(\rw(\sys(t,q), \mem))]}$.
We note, that for a given $\sys$ and $\mem$, $\finc(\sys, \mem)$ maintains the invariant of \rprop{prop:ehe:invariant}.
We are simplifying expressions or rewriting atoms with their values in the memory which is what $\seval$ already does for each entry in the $\MemRep{}$.
That is, $\finc(\sys, \mem)$ is a valid representation of the same deterministic and complete automaton as $\sys$.
However, $\finc(\sys, \mem)$  incorporates information from memory $\mem$ in addition.
\begin{proposition}[Memory obsolescence]\label{prop:inc}
  \[ \forall \tuple{t,q} \in \fdom(\finc(\sys,\mem)):
    \seval(\sys(t,q), \mem) \equiv \seval(\finc(\sys,\mem)(t,q), []).
  \]
\end{proposition}
\begin{proof}
	Follows directly by construction of $\finc$ and the definition of $\seval$ (which uses functions $\fsimp$ and $\rw$).
\end{proof}
\rprop{prop:inc} ensures that it is possible to directly incorporate a memory in an $\MemRep{}$, making the memory no longer necessary.
This is useful for algorithms that communicate the \MemRep{}, as they do not need to also communicate the memory.

By rewriting the expressions, the \MemRep{}s of two different monitors receiving different observations contain different expressions.
However, since they still encode the same automaton, and observations do not conflict, merging with $\sysadd$ shares useful information.

\begin{corollary}
  Given an \MemRep{} $\sys$ constructed using function $\smove$, and two memories $\mem_1$ and $\mem_2$ that do not observe conflicting observations,
  the two \MemRep{}s $\sys_1 = \finc(\sys, \mem_1)$ and $\sys_2 = \finc(\sys, \mem_2)$ have the following properties $\forall \tuple{t,q} \in \fdom(\sys')$:
  \begin{enumerate}
  \item $\sys' = \sys_1 \sysadd \sys_2$ is deterministic (\rprop{prop:ehe:invariant});
  \item $\seval(\sys'(t,q), []) \implies \seval(\sys(t,q), \mem_1 \memadd \mem_2)$;
  \item $(\seval(\sys'(t,q), []) = \vt) \implies ((\seval(\sys'(t,q), \mem_1) = \vt) \land (\seval(\sys'(t,q), \mem_2) = \vt))$;
  \item $(\seval(\sys'(t,q), []) = \vt) \implies
        (\seval(\sys_1(t,q), \mem_1) \neq \vf \land \seval(\sys_2(t,q), \mem_2) \neq \vf).
  $
\end{enumerate}
\end{corollary}
The first property ensures that the merge of two \MemRep{}s that incorporate memories are still indeed representing a deterministic and complete automaton, this follows from \rprop{prop:ehe:invariant} and \rprop{prop:inc}.
Since operation $\sysadd$  disjoins the two expressions, and since the two expressions come from \MemRep{}s that each maintain the property, the additional disjunction will not affect the outcome of $\seval$.
The second property extends \rprop{prop:inc} to the merging of \MemRep{} with incorporated memories.
It follows directly from \rprop{prop:inc}, and the assumptions that the memories have no conflicts.
The third property adds a stronger condition.
It states that merging two \MemRep{}s with incorporated memories results in an \MemRep{} that does not evaluate differently under the different memories.
This follows from the second property and the fact that the memories do not have conflicting observations.
Finally, the fourth property ensures that merging an \MemRep{} with an entry that evaluates to $\vf$ does not result in an entry that evaluates to $\vt$.
That is, if an \MemRep{} has already determined that a state is not reachable, merging it with another \MemRep{} does not result in the state being reachable.
This ensures the consistency when sharing information.
This property follows from the merging operator $\sysadd$ which uses $\lor$ to merge entries in two \MemRep{}s.
We recall that an entry in $\tuple{t,q} \in \fdom(\sys')$ is constructed as:
$\seval(\sys_1(t,q), \mem_1) \lor \seval(\sys_2(t,q), \mem_2)$.
For $\seval(\sys'(t,q), [])$ to be $\vt$,  either $\seval(\sys_1(t,q), \mem_1)$ or $\seval(\sys_2(t,q), \mem_2)$ has to be $\vt$, if one is already $\vf$, then the other has to be $\vt$.
This leads to a contradiction, since both $\sys_1$ and $\sys_2$ encode the same deterministic automaton, as such, the automaton cannot be in two states at once.

\cblock{%
For the same execution of the automaton and a timestamp $t$, if we have two encodings $\sys(t,q) = e$ and $\sys'(t,q) = e'$, then we know that the automaton is in $q$ at $t$ iff either $e$ or $e'$ evaluates to $\vt$ (\rdef{def:cmon:ehe}).
\edited{
Since $\sys$ and $\sys'$ encode the same automaton, the information needed to reach the same state is similar.
}
Therefore, the new expression $e \lor e'$ can be an effective way to reconcile information from two encodings \edited{of the same automaton}.
The memory $\mem$ can be embedded in an expression $e$ by simply using $\rw(e, \mem)$ (\rdef{def:common:rw}).
Thus, by rewriting expressions and combining $\MemRep{}$s, it is possible to reconcile multiple partial observations of an execution.
}%
\begin{example}[Reconciling information]\label{ex:cmon:combine}
We consider specification  $\boldsymbol{\mathrm{F}}(a \land b)$ (\rfig{fig:cmon:combine-aut}), and two components: $c_0$  and $c_1$ monitored by $m_0$ and $m_1$ respectively.
The monitors can observe the propositions $a$ and $b$ respectively and use one \MemRep{} each: $\sys_0$ and $\sys_1$ respectively.
Their memories are respectively $\mem_0 = [\tuple{1,a} \mapsto \vt]$ and $\mem_1 = [\tuple{1,b} \mapsto \vf]$.
\rtbl{tbl:cmon:combine} shows the $\MemRep{}$s at $t = 1$.
Constructing the \MemRep{} follows similarly from \rex{ex:cmon:exec}.
We show the rewriting for both $\sys_0$ and $\sys_1$ respectively in the next two columns.
Then, we show the result of combining the rewrites using $\sysadd$.
We notice initially that since $b$ is $\vf$, $m_1$ could evaluate $\neg\tuple{1,a} \lor \vt = \vt$ and know that the automaton is in state $q_0$.
However, for $m_0$, this is not possible until the expressions are combined.
By evaluating the combination $(\vf \lor \neg\tuple{1,b}) \lor \vt = \vt$, $m_0$ determines that the automaton is in state $q_0$.
In this case, we are only looking for expressions that evaluate to $\vt$.
We notice that monitor $m_1$ can determine that $q_1$ is not reachable (since $\tuple{1,a} \land \vf = \vf$) while $m_0$ cannot, as the expression $\tuple{1,b}$ cannot yet be evaluated to a final verdict.
This does not affect the outcome, as we are only looking for one expression that evaluates to $\vt$, since both $\sys_0$ and $\sys_1$ are encoding the same execution.
\end{example}
	\begin{figure}[t] %
		\centering%
		\scalebox{.95}{\begin{tikzpicture}[aut]
	\node[location] (q0) at (0,0) {$q_0$};
	\node[location, accept, right=of q0] (q1) {$q_1$};

	\tconnect{q0}{q1}{$a \land b$}{t1}
	\tconnect[loop right]{q1}{q1}{$\vt$}{t2}
	\tconnect[loop left]{q0}{q0}{$\neg a \lor \neg b$}{t2}
\end{tikzpicture}} %
		\vspace{-1em}
		\caption{Representing $\mathrm{F}(a \land b)$}%
		\label{fig:cmon:combine-aut}%
		\vspace{-1em}
	\end{figure}%

\gettable{
\small
\centering
\caption{Reconciling information}
\label{tbl:cmon:combine}
\begin{tabular}{| r | c || l || l || l |}
	\hline \textbf{t} & \textbf{q} & $\mathbf{\sys_0}$ & $\mathbf{\sys_1}$ & $\mathbf{{\sysadd}}$ \\
	\hline  \hline 0 		  & $q_0$	   & $\vt$
			&  $\vt$ & $\vt$\\
	\hline 1		  & $q_0$	   
          & $\vf \lor \neg \tuple{1,b}$
          & $\neg \tuple{1,a} \lor \vt$
          & $\vt$\\
	\hline 1 		  & $q_1$	   
          & $\vt \land \tuple{1,b}$
          & $\tuple{1,a} \land \vf$
          & $\tuple{1,b}$\\
	\hline
\end{tabular}
}

\section{Decentralized Specifications} \label{sec:dmon}

%
In this section, we shift the focus to a specification that is decentralized.
A set of automata represent various requirements (and dependencies) for different components of a system.
In this section, we define the notion of a decentralized specification and its semantics, and in \rsec{sec:prop}, we define various properties on such specifications.
%
\paragraph{Decentralizing a specification.}
%
We recall that a decentralized system consists of a set of components $\comps$.
To decentralize the specification, instead of having one automaton, we have a set of specification automata (Definition~\ref{def:cmon:aut}) $\mons = \setof{\aut_\ell \mid \ell \in \APmons}$, where $\APmons$ is a set of monitor labels.
We refer to these automata as \emph{monitors}.
To each monitor, we associate a component using a function $\cmapping : \mons  \rightarrow \comps$.
However, the transition labels of  a monitor $\mathrm{mon} \in \mons$ are expressions restricted to either observations local to the component the monitor is attached to (i.e., $\cmapping(\mathrm{mon})$), or references to other monitors.
Transitions are labeled over $\APmons \setminus \setof{\mathrm{mon}} \cup \setof{ap \in \AP \mid \clookup(ap) = \cmapping(\mathrm{mon})}$.
This ensures that the monitor is labeled with observations it can locally observe or depend on other monitors.
To evaluate a trace as one would on a centralized specification, we require one of the monitors to be a starting point,  we refer to that monitor as the \emph{root monitor} ($\mroot \in \mons$).
\begin{definition}[Decentralized specification]
	A decentralized specification is a tuple $\langle \APmons$, $\mons$, $\comps$, $\cmapping$, $\mroot \rangle$.
\end{definition}
We note that a centralized specification is a special case of a decentralized specification, with one component (global system, $\mathrm{sys}$), and one monitor ($\mathrm{g}$) attached to the sole component, i.e.
$\tuple{\setof{\mathrm{g}}, \setof{\aut_\mathrm{g}}, \setof{\mathrm{sys}}, [\aut_\mathrm{g} \mapsto \mathrm{sys}], \aut_\mathrm{g}}$.

As automata expressions now include references to monitors, we first define function $\mondep : \expr \rightarrow \mons$, which determines monitor dependencies.
Then, we define the semantics of evaluating (decentralized) specifications with references.
\begin{definition}[Monitor dependency]
\label{def:dmon:dep}
The set of monitor dependencies in an expression $\vars{e}$ is obtained by function $\mondep : \expr \rightarrow \mons$, defined as\footnote{We note that this definition can be trivially extended to any encoding of such expressions that contains the monitor id.}:
$ \mondep (e) = $ \texttt{match} $e$ \texttt{with}: \\
$
\begin{array}{llll}
				 |\> \mathit{id} \in \APmons  &\matchr\setof{\aut_{\it id}} &
				 |\> e_1 \land e_2 &\matchr\mondep(e_1) \cup \mondep(e_2)\\
				 |\> \neg e		   &\matchr\mondep(e) &
				 |\> e_1 \lor e_2  &\matchr\mondep(e_1) \cup \mondep(e_2)
\end{array}$
\end{definition}

Function $\mondep$ finds all monitors referenced by expression $\vars{e}$, by syntactically traversing it.

	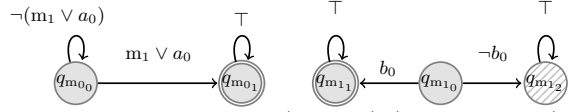
\begin{figure}[t] %
		\centering%
		\scalebox{0.8}{\begin{tikzpicture}[aut]
	\node[location] (m0q0) at (0,0) {$q_{\MONID_{0_0}}$};
	\node[location, accept, right=2 of m0q0] (m0q1) {$q_{\MONID_{0_1}}$};

	\tconnectt{m0q0}{m0q1}{$\MONID_1 \lor a_0$}{m0t1}{yshift=0cm}
	\tconnectt[loop above]{m0q0}{m0q0}{$\neg (\MONID_1 \lor a_0)$}{m0t4}{xshift=-0.3cm, yshift=-0.2cm}
	\tconnectt[loop above]{m0q1}{m0q1}{$\vt$}{m0t5}{yshift=-0.2cm}

	\node[location] (m1q0) at (6,0) {$q_{\MONID_{1_0}}$};
	\node[location, accept, left=1 of m1q0] (m1q1) {$q_{\MONID_{1_1}}$};
	\node[location, reject, right=1 of m1q0] (m1q2) {$q_{\MONID_{1_2}}$};

	\tconnect{m1q0}{m1q1}{$b_0$}{m1t1}
	\tconnect{m1q0}{m1q2}{$\neg b_0$}{m1t2}{yshift=-0.2cm}

	\tconnect[loop above]{m1q1}{m1q1}{$\vt$}{m1t3}{yshift=-0.3cm}
	\tconnect[loop above]{m1q2}{m1q2}{$\vt$}{m1t4}{yshift=-0.3cm}

\end{tikzpicture}} %
		\vspace{-1em}
		\caption{Representing $\mathrm{F}(a_0 \lor b_0)$ (decentralized)}%
		\label{fig:dmon:aut}%
		\vspace{-1em}
	\end{figure}%

\begin{example}[Decentralized specification]
\label{ex:dmon:aut}
\rfigb{fig:dmon:aut} shows a decentralized specification corresponding to the centralized specification in \rex{ex:cmon:aut}.
It consists of 2 monitors  $\aut_{\MONID_0}$ and $\aut_{\MONID_1}$ respectively.
We consider two atomic propositions $a_0$ and $b_0$ which can be observed by component $c_0$ and $c_1$ respectively.
The monitor labeled with $\MONID_0$ (resp. $\MONID_1$) is attached to component $c_0$ (resp. $c_1$).
$\aut_{\MONID_0}$ depends on the verdict from $\MONID_1$ and only observations local to $c_0$, while $\aut_{\MONID_1}$ is only labeled with with observations local to $c_1$.
Given the expression $m_1 \land a_0$, we have $\mondep(m_1 \land a_0) = \setof{\aut_{\MONID_1}}$.
\end{example}
\paragraph{Semantics of a decentralized specification.}
The transition function of the decentralized specification is similar to the centralized automaton with the exception of monitor ids.
\begin{definition}[Semantics of a decentralized specification]
\label{def:sem-decent}
Consider the root monitor $\aut_\mroot$ and a decentralized trace $\trace$ with index $i \in [1, |\trace|]$ representing the timestamps.
Monitoring $\trace$ starting from $\aut_\mroot$ emits the verdict $\verdictf_{\mroot} (\Delta'^*_{\mroot}(q_{\mroot_0}, \trace, 1))$ where for a given monitor label $\ell$:
\begin{align*}
  \Delta_\ell'^*(q, \trace, i) &= \funcparts {
    \Delta_\ell'^*(\Delta_\ell'(q, \trace, i), \trace, i + 1) & \pif i < |\trace|\\
    \Delta_\ell'(q, \trace, i) & \pelse
  }\\
		\Delta'_\ell(q, \trace, i) &=\funcparts{
			q'
				& \mbox{ if } \trace(i, \cmapping(\aut_\ell)) \neq \emptyset \land  \exists e \in \expr_\AP : \delta_{\ell}(q, e) = q'  \land \, \seval(e, \mem) = \vt\\
			q  & \mbox{ otherwise}
		}
\end{align*}
\vspace{-2em}
\begin{align*}
  \text{where}\quad  \mem    &= \cons(\trace(i,  \cmapping(\aut_\ell)), \fid) \memadd \smashoperator{\biguplus^2_{\aut_{\ell'} \in \mondep(e)}} \setof{[\ell' \mapsto \verdictf_{\ell'}(q_{\ell'_f})]}\\
  \text{and}\quad  q_{\ell'_f} &= \Delta'^*_{\ell'}(q_{\ell'_0}, \trace, i)\\
\end{align*}
\vspace{-2em}
\end{definition}
For a monitor $\aut_\ell$, we determine the new state of the automaton starting at $q \in Q_{\ell}$, and running the trace $\trace$ from timestamp $i$ to timestamp $t$ by applying $\Delta_{\ell}'^*(q, \trace, i)$.
To do so, we evaluate one transition at a time using $\Delta_{\ell}'$ as would $\Delta_{\ell}^*$ with $\Delta_{\ell}$ (see \rdef{def:cmon:aut-semantics}).
To evaluate $\Delta_{\ell}'$ at any state $q' \in Q_{\ell}$, we need to evaluate the expressions so as to determine the next state $q''$.
The expressions contain atomic propositions and monitor ids.
For atomic propositions, the memory is constructed using $\cons(\trace(i, \cmapping(\aut_\ell)), \fid)$ which is based on the event with observations local to the component the monitor is attached to (i.e., $\cmapping(\aut_\ell)$).
However, for monitor ids, the memory represents the verdicts of the monitors.
To evaluate each reference $\ell'$ in the expression, the remainder of the trace starting from the current event timestamp $i$ is evaluated recursively on the automaton $\aut_{\ell'}$ from the initial state $q_{\ell'_0} \in \aut_{\ell'}$.
Then, the verdict of the monitor is associated with $\ell'$ in the memory.
\begin{example}[Monitoring of a decentralized specification]
Consider monitors $\aut_{\MONID_0}$ (root) and $\aut_{\MONID_1}$ associated to components $c_0$ and $c_1$ respectively and the trace $\trace = [
  1  \mapsto c_0  \mapsto \setof{\tuple{a, \vf}},
  1  \mapsto c_1  \mapsto \setof{\tuple{b, \vf}},
  2  \mapsto c_0  \mapsto \setof{\tuple{a, \vf}},
  2  \mapsto c_1  \mapsto \setof{\tuple{b, \vt}}
] $.
To evaluate $\trace$ on $\aut_{\MONID_0}$ (from \rfig{fig:dmon:aut}), we use $\Delta_{\MONID_0}'^*(q_{{\MONID_0}_0}, \trace, 1)$.
To do so, we first evaluate $\Delta_{\MONID_0}'(q_{{\MONID_0}_0}, \trace, 1)$.
We notice that the expressions depend on $\MONID_1$, therefore we need to evaluate $\Delta_1'^*(q_{{\MONID_1}_0}, \trace, 1)$.
All expressions have no monitor labels, thus we construct $\mem^1_{{\MONID_1}} =  \cons(\tuple{b, \vf}, \fid) = [b \mapsto \vf]$, and notice that $\seval(\neg b, \mem^1_{{\MONID_1}}) = \vt$ and therefore it can move to state $q_{{\MONID_1}_2}$ associated with verdict $\vf$.
Notice that $\Delta_{\MONID_1}'($ $\Delta_{\MONID_1}'(q_{{\MONID_1}_0}, \trace, 1), \trace, 2) = q_{{\MONID_1}_2}$ with $\verdictf_{{\MONID_1}}(q_{{\MONID_1}_2}) = \vf$.
We can construct $\mem^1_{{\MONID_0}} =  \cons(\tuple{a, \vf},$ $\fid) \memadd [\MONID_1 \mapsto \vf] = [a \mapsto \vf, \MONID_1 \mapsto \vf]$.
We then have $\seval(\neg \MONID_1 \land \neg a_0, \Mem^1_{{\MONID_0}}) = \vt)$ and $\aut_{\MONID_0}$ is in state $(\MONID_0, q_0)$.
By doing the same for $t = 2$, we obtain $\Mem^2_{{\MONID_0}} = [a \mapsto \vf, \MONID_1 \mapsto \vt]$, we then evaluate $\seval(\MONID_1 \land \neg a_0) = \vt$.
This indicates that $\Delta_{\MONID_0}'(\Delta_{\MONID_0}'(q_{{\MONID_0}_0}, \trace, 1), \trace, 2) = q_{{\MONID_0}_1}$ and the final verdict is $\vt$.
\end{example}

%
%
%
%
%
%
%

%
\section{Properties for Decentralized Specifications}
\label{sec:prop}
%
A key advantage of using decentralized specifications is to make the association of monitors with components explicit.
Since monitors have been explicitly modeled as a set of automata with dependencies between each other, we can now determine properties on decentralized specifications.
In this section, we revisit the concept of \emph{monitorability}, characterize it for automata, define it for decentralized specifications, and describe an algorithm for deciding monitorability.
Furthermore, we explore \emph{compatibility}, that is the ability of a decentralized specification to be deployed on a given architecture.

\subsection{Decentralized Monitorability}\label{sec:dmon:monitorability}

An important notion to consider when dealing with runtime verification is that of monitorability~\cite{PnueliZ06,Falcone10}.
In brief, monitorability of a specification determines whether or not an RV technique is applicable to a specification.
That is, a monitor synthesized for a non-monitorable specification is unable to check if the execution complies or violates the specification for all possible traces.
Consider the automaton shown in \rfig{fig:monitorability:aut}, one could see that there is no state labeled with a final verdict.
In this case, we can trivially see that no trace allows us to reach a final verdict.
We also notice similar behavior when monitoring LTL expressions with the pattern $\mathbf{G}\mathbf{F}(p)$ with $p$ is an atomic proposition.
The LTL expression requires that at all times $\mathbf{F}(p)$ holds $\vt$, while $\mathbf{F}(p)$ requires that $p$ eventually holds $\vt$.
As such, at any given point of time, we are unable to determine a verdict, since if $p$ is $\vf$ at the current timestamp, it can still be $\vt$ at a future timestamp, and thus $\mathbf{F}(p)$ will be $\vt$ for the current timestamp.
And if $\mathbf{F}(p)$ is $\vt$ at the current timestamp, the $\mathbf{G}$ requires that it be $\vt$ for all timestamps, so in the future there could exist a timestamp which falsifies it.
Consequently, when monitoring such an expression, a monitor will always output $\vna$, as it cannot determine a verdict for any given timestamp.
In this section, we first characterize monitorability in terms of automata and \MemRep{} for both centralized and decentralized specifications.
Then, we provide an effective algorithm to determine monitorability.

	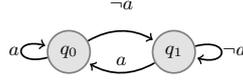
\begin{figure}[t] %
		\centering%
		\scalebox{0.8}{\begin{tikzpicture}[aut]
	\node[location] (m0q0) at (0,0) {$q_{0}$};
	\node[location, right=of m0q0] (m0q1) {$q_{1}$};

	\tconnectt[loop left]{m0q0}{m0q0}{$a$}{m0t0}{yshift=-0.3cm}
  \tconnectt[loop right]{m0q1}{m0q1}{$~\neg a$}{m0t1}{xshift=1.2,yshift=-0.3cm}
		\tconnect[bend left]{m0q0}{m0q1}{$\neg a$}{t1}
		\tconnect[bend left]{m0q1}{m0q0}{$a$}{t2}

\end{tikzpicture}} %
		\vspace{-1em}
		\caption{A trivial nonmonitorable specification}%
		\label{fig:monitorability:aut}%
		\vspace{-1em}
	\end{figure}%

%
\subsubsection{Characterizing Monitorability}
%
\paragraph{Centralized monitorability of properties.}
Monitorability in the sense of \cite{PnueliZ06} is defined on traces.
A monitorable property is where for all finite traces $t$ (a sequence of events) in the set of all (possibly infinite) traces,
there exists a continuation $t'$ such that monitoring $t \cdot t'$ results in a true or false verdict.
Informally, it can be seen as whether or not continuing to monitor the property after reading $t$ can still yield a final verdict.
We note that this definition covers all possible traces, it establishes monitorability to be oblivious of the input trace.
That is, one can determine whether a property is monitorable irrespective of the input trace.

\paragraph{Centralized monitorability in automata.}
We can generalize monitorability to reach ``true'' or ``false'' verdict to the notion of reaching a \emph{final verdict}, and extend it to automata.
For automata, monitorability can be analyzed in terms of reachability and states.
Given a specification $\aut = \tuple{Q, q_0 \in Q, \delta, \verdictf}$,
a state $q \in Q$ is monitorable (denoted $\monitorable(q)$) iff $\verdictf(q') \in \verdictb$ or $\exists q' \in Q$ such that $\verdictf(q') \in \verdictb$ and $q'$ is reachable from $q$.
Specification $\aut$ is said to be monitorable (denoted $\monitorable(\aut)$) iff $\forall q \in Q : \> \monitorable(q)$.
Defining monitorability using reachability is consistent with \cite{PnueliZ06}.
After reading a finite trace $t$ and reaching $q$ ($q = \Delta^*(q_0 , t)$), there exists a continuation $t'$ that leads the automaton to a state $q'$ ($q' = \Delta^*(q , t')$), such that $\verdictf(q') \in \verdictb$.
We note that a specification is monitorable according to this definition iff, in an automaton, all paths from the initial state $q_0$ lead to a state with a final verdict.
As such, it is sufficient to analyze the automaton to determine monitorability irrespective of possible traces (see \rsec{sec:monitorability:compute}).

The automaton presented in Figure~\ref{fig:cmon:aut} to express $\mathbf{F}(a \lor b)$ is monitorable, as both $q_0$ and $q_1$ are monitorable.
At $q_0$, it is possible to reach $q_1$ labeled by the final verdict $\vt$.
We note that monitorability is a necessary but not sufficient condition for termination (with a final verdict).
An infinite trace consisting of the event $\setof{\tuple{a,\vf}, \tuple{b,\vf}}$ never lets the automaton reach $q_1$.
However, monitorability guarantees the eventuality of reaching a final verdict.
Thus, if a state $q$ is not monitorable, then we know that it is impossible to reach a final verdict from $q$, and can abandon monitoring.

\paragraph{Centralized monitorability with \MemRep{}}
Reachability in automata can be expressed as well using the \MemRep{} data structure.
A path from a state $q$ to a state $q'$ is expressed as an expression over atoms.
We define $\paths(q, q')$ to return all possible paths from $q$ to $q'$.
\[\paths(q, q') = \setof{
          \mathit{e} \mid
          \exists t \in \timestamp:  \sys^t(t, q') = \mathit{e}
          \land \,  \sys^t = \smove([0 \mapsto q\mapsto \vt], 0, t)}
  \]
%
Each expression is derived similarly as would an execution in the $\MemRep{}$ (\rdef{def:cmon:ehe}).
We start from state $q$ and use a logical timestamp starting at 0 incrementing it by 1 for the next reachable state.
A state $q$ is monitorable iff $\exists \mathit{e}_f \in \paths(q, q_f)$, such that (1) $ \mathit{e}_f$ is satisfiable;
	(2) $\verdictf(\fsimp( \mathit{e}_f)) \in \verdictb$.
The first condition ensures that the path is able to lead to the state $q_f$, as an unsatisfiable path will never evaluate to true. The second condition ensures that the state is labeled by a final verdict.
An automaton is thus monitorable iff all its states are monitorable.
We note that $\paths(q, q')$ can be infinite if the automaton contains cycles, however path expressions could be ``compacted" using the pumping lemma.
Using \MemRep{} we can frame monitorability as a satisfiability problem which can benefit from additional knowledge on the truth values of atomic propositions.
For the scope of this paper, we focus on computing monitorability on automata in \rsec{sec:monitorability:compute}.

\paragraph{Decentralized monitorability.}
In the decentralized setting, we have a set of monitors $\mons$.
The labels of automata include monitor ids ($\APmons$).
We recall that the evaluation of a reference $\ell\in \APmons$ consists in running the remainder of the trace on $\aut_{\ell}$ starting from the initial state $q_{\ell_0}$.
As such, for any dependency on a monitor $\aut_\ell$, we know that $\ell$ evaluates to a final verdict iff $\monitorable(\aut_\ell)$.
We notice that monitorability of decentralized specification is recursive, and relies on the inter-dependencies between the various decentralized specifications.
This is straightforward for \MemRep{}, since a path is an expression.
For a path $\mathit{e}_f$, the dependent monitors are captured in the set $\mondep(\mathit{e}_f)$.
The additional condition on the path is thus: $\forall \aut_\ell \in \mondep(\mathit{e}_f): \> \monitorable(\aut_\ell)$.


\subsubsection{Computing Monitorability} \label{sec:monitorability:compute}

\paragraph{Centralized specification}
We compute the monitorability of a centralized specification $\aut$, with respect to a set of final verdicts $\verdictb$\footnote{While we use $\verdictb$, this can be extended without loss of generality to an arbitrary set $\mathbb{B}_\mathrm{f}$.}.
Computing monitorability consists in checking that all states of the automaton are co-reachable from states with final verdicts.
As such, it relies on a traversal of the graph starting from the states that are labeled with final verdicts.
To do so, we use a variation of the work-list algorithm.
We begin by adding all states labeled by a final verdict to the work list.
These states are trivially monitorable.
Conversely, any state that leads to a monitorable state is monitorable.
As such, for each element in the work list, we add its predecessors to the work list.
We maintain a set of marked states ($\func{Mark}$), that is, states that have already been processed, so as to avoid adding them again.
This ensures that cycles are properly handled.
The algorithm stabilizes when no further states can be processed (i.e., the work list is empty).
All marked states ($\func{Mark}$) are therefore monitorable.
To check if an automaton is monitorable, we need all of its states to be monitorable.
As such we verify that $|\func{Mark}| = |Q|$.
The number of edges between any pair of states can be rewritten to be at most 1 (as explained in Section~\ref{sec:cmon:pre}).
As such, one has to traverse the graph once, the complexity being linear in the states and edges (i.e., $O(|Q| + |\delta|)$).
Hence in the worst case, an automaton forms a complete graph, and we have ${{|Q|}\choose{2}}$ edges.
The worst case complexity is quadratic in the number of states (i.e., $O(|Q| + \frac{1}{2}|Q| (|Q|-1))$).%

\cblock{%
\begin{algorithm}[t]\scriptsize
 \caption{Centralized Automata Monitorability}
 \label{alg:monitorable:aut}
\begin{algorithmic}[1]
\Procedure{CAMonitorable}{$\aut$, $\verdictb$}
  \State $\func{WL} \gets \func{queue}()$
  \State $\func{Mark} = \emptyset$ \Comment{Define the set of processed states}
	\ForAll{$q \in Q$} \Comment{States with a final verdict}
		\If{$\verdictf(q) \in \verdictb$}
			\State $\func{push}(\func{WL}, q)$ \Comment{Add each state to the worklist}
      \State $\func{Mark} \gets \func{Mark} \cup \setof{q}$ \Comment{Mark the state as processed}
		\EndIf
	\EndFor
	\While{$\neg \func{empty}(\func{WL})$} \Comment{Process}
		\State $\func{q'} \gets \func{pop}(\func{WL})$
		\ForAll{$q_p \in \setof{q \in Q \mid \exists e \in \expr_\AP: \> \delta(q, e) = q'}$} \Comment{Consider all co-reachable states}
			\If{$q_p \not\in \func{Mark}$} \Comment{Consider only states which we have not seen yet}
				\State $\func{push}(\func{WL}, q_p)$ \Comment{Add co-reachable state to the work-list}
				\State $\func{Mark} \gets \func{Mark} \cup \setof{q_p}$ \Comment{Mark the state as processed}
			\EndIf
		\EndFor
	\EndWhile
	\State \Return $|\func{Mark}| = |Q|$ \Comment{Have all states been processed?}
\EndProcedure
\end{algorithmic}
\end{algorithm}
}%


\paragraph{Decentralized specifications}
In the case of decentralized specifications, the evaluation of paths (using $\seval$) in an automaton depends on other monitors (and thus other automata).
To compute monitorability, we first build the monitor dependency set for a given monitor $\aut_\ell$ (noted $\MDS(\aut_\ell)$) associated with a monitor label $\ell$.
\[
	\MDS(\aut_\ell) = \biguplus_{\setof{e  \in \expr\, \mid\, \exists q, q' \in Q_\ell: \> \delta_\ell(q, e) = q'}} \mondep(e)
\]
The monitor dependency list for a monitor contains all the references to other monitors across all paths in the given automaton ($\aut_\ell$), by examining all the transitions.
It can be obtained by a simple traversal of the automaton.

Second, we construct the monitor dependency graph (MDG), which describes the dependencies between monitors.
The monitor dependency graph for a set of monitors $\mons$ is noted $\MDG(\mons) = \tuple{\mons, \func{DE}}$ where $\func{DE}$ is the set of edges which denotes the dependency edges between the monitors, defined as:
	$\func{DE} = \setof{\tuple{\aut_{\ell}, \aut_{\ell'}} \in \mons \times \mons \mid \aut_{\ell'} \in \MDS(\aut_\ell)}$.
A monitor $\aut_{\MONID_i}$ depends on another monitor $\aut_{\MONID_j}$ iff $\MONID_j$ appears in the expressions on the transitions of  $\aut_{\MONID_i}$.
For a decentralized specification to be monitorable, the two following conditions must hold:
$\MDG(\mons)$ has no cycles; and
$\forall \ell \in \mons: \> \func{CAMonitorable}(\aut_\ell)$.
The first condition ensures that no cyclical dependency exists between monitors.
The second condition ensures that all monitors are individually monitorable.
We note, that both conditions are decidable.
Furthermore, detecting cycles in a graph can be done in linear time with respect to the sum of nodes and edges, by doing a depth-first traversal with back-edge detection, or by finding strongly connected components~\cite{Tarjan72}.
Thus in worst case, it is quadratic in $|\mons|$.
Monitorability is therefore quadratic in the number of monitors and states in the largest automaton.
%
\subsection{Compatibility} \label{sec:prop:compat}
%
A key advantage of decentralized specifications is the ability to associate monitors to components.
This allows us to associate the monitor network with the actual system architecture constraints.

The monitor network is a graph $\func{N} = \tuple{\mons, E}$, where $\mons$ is the set of monitors, and $E$ representing the communication edges between monitors.
The monitor network is typically generated by a monitoring algorithm during its \emph{setup} phase (See Section~\ref{sec:analysis:algs}).
For example, $\func{N}$ could be obtained using the construction $\MDG(\mons)$ presented in Section~\ref{sec:monitorability:compute}.
The system is represented as another graph $S = \tuple{\comps,E'}$, where $\comps$ is the set of components, and $E'$ is the set of communication channels between components.

\paragraph{Defining compatibility.}
We now consider checking for \emph{compatibility}.
Compatibility denotes whether a monitoring network can be actually deployed on the system.
That is, it ensures that communication between monitors is possible when those are deployed on the components.
We first consider the reachability in both the system and monitor graphs as the relations $\reach{S} : \comps \rightarrow 2^\comps$, and $\reach{M} : \mons \rightarrow 2^\mons$, respectively.
Second, we recall that a monitor may depend on other monitors and also on observations local to a component.
If a monitor depends on local observations, then it provides us with constraints on where it should be placed.
We identify those constraints using the partial function $\constraint : \mons \rightarrow \comps$.
We can now formally define compatibility.
Compatibility is the problem of deciding whether or not there exists a \emph{compatible assignment}.

\begin{definition}[Compatible assignment] \label{def:compatibility}
A compatible assignment is a function $\compat: \mons \rightarrow \comps$ that assigns monitors to components while preserving the following properties:
\begin{enumerate}
    \item $\forall m_1, m_2 \in \mons: m_2 \in \reach{M}(m_1) \implies \compat(m_2) \in \reach{S}(\compat(m_1))$;
    \item $\forall m \in \fdom(\constraint): \constraint(m) = \compat(m)$.
\end{enumerate}
\end{definition}
The first proposition ensures that reachability is preserved.
That is, it ensures that if a monitor $m_1$ communicates with another monitor $m_2$ (i.e. $m_2 \in \reach{M}(m_1)$), then $m_2$ must be placed on a component reachable from where $m_1$ is placed (i.e. $ \compat(m_2) \in \reach{S}(\compat(m_1))$).
The second proposition ensures that dependencies on local observations are preserved.
That is, if a monitor $m$ depends on local observations from a component $c \in \comps$ (i.e. $\constraint(m) = c$), then $m$ must be placed on $c$ (i.e. $\constraint(m) = \compat(m)$).

\paragraph{Computing compatibility.}
We next consider the problem of finding a \emph{compatible assignment} of monitors to components.
Algorithm~\ref{alg:compatible:brute} finds a compatible assignment for a given monitor network ($\tuple{\mons, E}$), system ($\tuple{\comps, E'}$), and an initial assignment of monitors to components ($\constraint$).
The algorithm can be broken into three procedures: procedure $\funca{verifyCompatible}$ verifies that a (partial) assignment of monitors to components is compatible,
 procedure $\funca{compatibleProc}$ takes as input a set of monitors that need to be assigned and explores the search space (by iterating over components), and finally,
 procedure $\funca{compatible}$ performs necessary pre-computation of reachability, verifies that the constraint is first compatible, and starts the search.

We verify that an assignment of monitors to components ($\mrm{s} : \mons \rightarrow \comps$) is compatible using algorithm $\funca{verifyCompatible}$ (Lines 1-8).
We consider each assigned monitor ($m \in \fdom(\mrm{s})$).
Then,  we constrain the set of reachable monitors from $m$ to those which have been assigned a component ($M' = \reach{M}(m) \cap \fdom(\mrm{s})$).
Using $M'$, we construct a new set of components using $\mrm{s}$ (i.e., $C' = \setof{\mrm{s}(m') \in \comps \mid m' \in M'}$).
Set $C'$ represents the components on which reachable monitors have been placed.
Finally, we verify that the components in the set $C'$ are reachable from where we placed $m$ (i.e., $C' \subseteq \reach{S}(\mrm{s}(m))$).
If that is not the case, then the assignment is not compatible (Line 4).
To iterate over all the search space, that is, all possible assignments of monitors to components, procedure $\funca{compatibleProc}$ (Lines 9-24) considers a set of monitors to assign ($\mrm{M}$),
selects a monitor $m \in \mrm{M}$ (Line 13), and iterates over all possible components, verifying that the assignment is compatible (Lines 14-22).
If the assignment is compatible, it iterates over the remainder of the monitors (i.e., $\mrm{M} \setminus \setof{m}$), until it is empty (Line 16).
If the assignment is not compatible, it discards it and proceeds with another component.
For each monitor we seek to find at least one compatible assignment.
One can see that the procedure eventually halts (as we exhaust all the monitors to assign), and is affected exponentially based on the number of monitors to assign $|\mons \setminus \fdom(\constraint)|$ (Line 31) with a branching factor determined by the possible values to assign ($|\comps|$, Line 14).
It is important to note that the number of monitors to assign is in practice particularly small.
The number of monitors to assign includes monitors that depend only on other monitors and not local observations from components, as the dependency on local observations requires that a monitor be placed on a given component (that is, it will be in $\fdom(\constraint)$).


\begin{example}[Compatibility]
  \rfig{fig:compat} presents a simple monitor network of 3 monitors, and a system graph of 4 components.
  We consider the following constraint: $\constraint = [m_0 \mapsto c_0, m_2 \mapsto c_2]$.
  For compatibility, we must first verify that $\constraint$ is indeed a compatible (partial) assignment, then consider placing $m_1$ on any of the components (i.e., both properties of Definition~\ref{def:compatibility}).
  Procedure $\funca{compatible}$ computes the set of reachable nodes for both the monitor network and the system.
  They are presented in \rfig{fig:compat:reachm} and \rfig{fig:compat:reachs}, respectively.
  We then proceed with line 28 to verify the constraint ($\constraint$) using procedure $\funca{verifyCompatible}$.
  We consider both $m_0$ and $m_2$.
  For $m_0$ (resp. $m_2$) we generate the set (Line 3) $\setof{c_0}$ (resp. $\setof{c_2}$), and verify that it is indeed a subset of $\reach{S}(c_0)$ (resp. $\reach{S}(c_2))$.
  This ensures that the constraint is compatible.
  We then proceed to place $m_1$ by calling $\funca{compatibleProc}(\constraint, \setof{m_1}, \setof{c_0, c_1, c_2, c_3}, \reach{M} , \reach{S})$.
  While procedure $\funca{compatibleProc}$ will attempt all components, we will consider for the example placing $m_1$ on $c_1$.
  On line 15, the partial function $s'$ will be $\constraint \memadd [m_1 \mapsto c_1]$.
  We now call $\funca{verifyCompatible}$ to verify $s'$.
  We consider both $m_0$, $m_1$, and $m_2$.
  For $m_0$ (resp. $m_1$, $m_2$) we generate the set $\setof{c_0, c_1}$ (resp. $\setof{c_1}, \setof{c_2, c_1}$).
  We notice that for $m_0$, $\setof{c_0, c_1}$ is indeed a subset of $\reach{S}(c_0)$.
  This means that $m_0$ is able to communicate with $m_1$.
  However,  it is not the case for $m_2$, the set $\setof{c_2, c_1}$ is not a subset of  $\reach{S}(c_2) = \setof{c_2, c_3}$.
  The monitor $m_2$ will not be able to communicate with $m_1$ if $m_1$ is placed on $c_1$.
  Therefore, assigning $m_1$ to $c_1$ is incompatible.
  Example of compatible assignments for $m_1$ are $c_2$ and $c_3$ as both of those components are reachable from $c_2$.
  Procedure $\funca{compatibleProc}$ continues by checking other components, and upon reaching  $c_2$ or $c_3$  stops and returns that there is at least one compatible assignment.
  Therefore, the monitor network (\rfig{fig:compat:net}) is compatible with the system (\rfig{fig:compat:sys}).
\end{example}

\begin{figure}[t]
  \centering
  \subfloat[Monitor Network\label{fig:compat:net}]{%
      {\scalebox{0.8}{\begin{tikzpicture}[aut]
	\node[location] (m1) at (0,0) {$m_1$};
	\node[location, left=of m1] (m0) {$m_0$};
	\node[location, right=of m1] (m2) {$m_2$};

	\tconnect{m0}{m1}{}{m0m1}{}
	\tconnect{m2}{m1}{}{m2m1}{}

\end{tikzpicture}}}
    }
    \subfloat[System\label{fig:compat:sys}]{%
    {\scalebox{0.8}{\begin{tikzpicture}[aut]
	\node[location] (c0) at (0,0) {$c_0$};
	\node[location, right=of c0] (c1) {$c_1$};
	\node[location, right=of c1] (c2) {$c_2$};
	\node[location, right=of c2] (c3) {$c_3$};

	\tconnect{c0}{c1}{}{t0}{}
	\tconnect{c1}{c2}{}{t1}{}
	\tconnect{c2}{c3}{}{t2}{}

\end{tikzpicture}}}
    }
    \subfloat[$\reach{M}$\label{fig:compat:reachm}]{%
    \begin{tabular}{|r|l|}
      \hline  $m_0$ & $\setof{m_0, m_1}$\\
      \hline  $m_1$ & $\setof{m_1}$\\
      \hline  $m_2$ & $\setof{m_2, m_1}$\\
      \hline
      \end{tabular}
    }
    \subfloat[$\reach{S}$\label{fig:compat:reachs}]{%
    \begin{tabular}{|r|l|}
      \hline  $c_0$ & $\setof{c_0, c_1, c_2, c_3}$\\
      \hline  $c_1$ & $\setof{c_1, c_2, c_3}$\\
      \hline  $c_2$ & $\setof{c_2, c_3}$\\
      \hline  $c_3$ & $\setof{c_3}$\\
      \hline
      \end{tabular}
    }
    \caption{Example Compatibility}
    \label{fig:compat}
\end{figure}
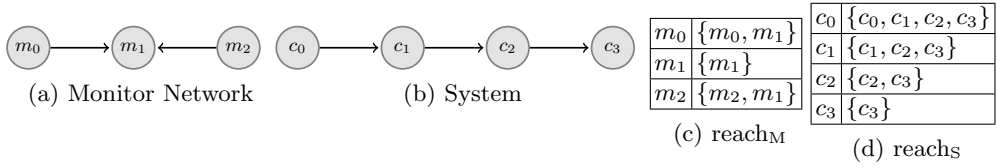

\begin{algorithm}[t]
 \caption{Computing Compatibility}
 \label{alg:compatible:brute}
\begin{algorithmic}[1]
\Procedure{verifyCompatible}{$\mrm{s}$, $\reach{M}$, $\reach{S}$} \Comment{Verify the assignment $\mrm{s}$}
	\ForAll{$m \in \fdom(\mrm{s})$} \Comment{Consider only assigned monitors}
		\If{$\setof{\mrm{s}(m') \mid m' \in (\reach{M}(m) \cap \fdom(\mrm{s}))} \not\subseteq \reach{S}(\mrm{s}(m))$ }\Comment{Check reachability}
			\State \Return $\mathit{false}$
		\EndIf
	\EndFor
	\State \Return $\mathit{true}$
\EndProcedure
\Procedure{compatibleProc}{$\mrm{s}$, $\mrm{M}$, $\mrm{C}$ , $\reach{M}$, $\reach{S}$} \Comment{Explore assignments}
  \If{$\mrm{M} = \emptyset$} \Comment{No monitors left to assign}
    \State \Return $\tuple{\mathit{true}, \mrm{s}}$ \Comment{Successfully assigned all monitors}
  \EndIf
  \State $m \gets \mrm{pick}(M)$ \Comment{Pick a monitor from those left to assign}
  \ForAll{$c \in \mrm{C}$} \Comment{Explore assigning monitor to all possible components}
    \State $s' \gets  \mrm{s} \memadd [m \mapsto c]$ \Comment{Add assignment to the existing solution}
    \If{$\funca{verifyCompatible}(s', \reach{M} , \reach{S})$} \Comment{Is it compatible?}
      \State $\tuple{res, sol} \gets \funca{compatibleProc}(s', \mrm{M} \setminus \setof{m}, \mrm{C}, \reach{M}, \reach{S}$) \Comment{Recurse on the rest}
      \If{$res$} \Comment{Found a compatible assignment for all the rest of $\mathrm{M}$}
        \State \Return $\tuple{res, sol}$
      \EndIf
    \EndIf
  \EndFor
  \State \Return $\tuple{\mathit{false}, []}$ \Comment{No compatible assignment found}
\EndProcedure
\Procedure{compatible}{$\tuple{\mons, E}$, $\tuple{\comps, E'}, \constraint$}
  \State $\reach{M} \gets \funca{computeReach}(\tuple{\mons, E})$ \Comment{Precompute reachability}
  \State $\reach{S} \gets \funca{computeReach}(\tuple{\comps, E'})$
  \If{$\neg \funca{verifyCompatible}(\constraint, \reach{M} , \reach{S})$} \Comment{Check constraint first}
    \State \Return $\tuple{\mathit{false}, []}$ \Comment{Constraint not satisfied}
  \EndIf
	\State \Return $\funca{compatibleProc}(\constraint, \mons \setminus \fdom(\constraint), \comps, \reach{M}, \reach{S})$ \Comment{Begin exploring}
\EndProcedure
\end{algorithmic}
\end{algorithm}

\section{Analysis} \label{sec:analysis}
%
%
We compare decentralized monitoring algorithms in terms of computation, communication and memory overhead.
We first consider the parameters and the cost for the basic functions of the \MemRep{}.
Then, we adapt the existing algorithms to use \MemRep{} and analyze their behavior.
We use $\fenc{e}$ to denote the size necessary to encode an element of the set $E$.
For example, $\fenc{\AP}$ is the size needed to encode an element of set $\AP$.
%
\subsection{Data Structure Costs} \label{sec:analysis:ds}
We consider the cost of using a memory or an \MemRep{}.
To do so, we first address the cost to store partial functions and merge them.

\paragraph{Storing partial functions.}
Since memory and $\MemRep{}$ are partial functions, to assess their required memory storage and iterations, we consider only the elements defined in the function.
The size of a partial function $f$, denoted $|f|$, is the size to encode all $x = f(x)$ mappings.
We recall that $|\fdom(f)|$ the number of entries in $f$.
The size of each mapping $x = f(x)$ is the sum of the sizes $|x| + |f(x)|$.
Therefore $|f| = \sum_{x\in\fdom(f)} |x| + |f(x)|$.
\paragraph{Merging.}
Merging two memories or two \MemRep{}s is linear in the size of both structures in both time and space.
In fact, to construct $f \dagger_{\mathrm{op}} g$, we first iterate over each $x \in \fdom(f)$, check whether $x \in \fdom(g)$, and if so assign $\mathrm{op}(f(x), g(x))$, otherwise assign $f(x)$.
Finally we assign $g(x)$ to any $x \in \fdom(g) \cap \overline{\fdom(f)}$.
This results in $|\fdom(f \dagger_{\mathrm{op}} g)| = |\fdom(f) \cup \fdom(g)|$ which is at most $|\fdom(f)| + |\fdom(g)|$.
\paragraph{Information delay.}
An $\MemRep{}$ associates an expression with a state for any given timestamp.
When an expression $\vars{e}$ associated with a state $q_\sstable$ for some timestamp $t_\sstable$ is evaluated to \vt, we know that the automaton is in $q_\sstable$ at $t_\sstable$.
We call $q_\sstable$ a `\emph{known}' (or stable) state.
Since we know the automaton is in $q_\sstable$, prior information is no longer necessary, therefore it is possible to discard all entries in $\sys$ with $t < t_\sstable$.
We parametrize the number of timestamps needed to reach a new known state from an existing known state as the information delay $\paramdelay$.
This can be seen as a garbage collection strategy~\cite{Wuu,CRDT} for the memory and $\MemRep{}$.
\paragraph{\MemRep{} encoding.}
For the \MemRep{} data structure, we consider in three functions: $\smove$, $\seval$, and $\sreach$\footnote{$\verdictAt$ is simply a $\sreach$ followed by a $O(1)$ lookup} (see \rsec{sec:cmon:ehe}).
Function $\smove$ depends on the topology of the automaton.
We quantify it using the maximum size of the expression that labels a transition in a normalized automaton (see Remark~\ref{rmk:normalized-aut}) $\paramlabel$, and the number of states in the automaton $|Q|$).
From a known state each application of $\smove$ considers all possible transitions and states that can be respectively taken and reached, for each outbound transition, the label itself is added.
Therefore, the rule is expanded by $\paramlabel$  per outbound state for each move beyond $t_\sstable$.
For each timestamp, we need for each state an expression.
The maximum size of the $\MemRep{}$ is therefore:
\[
	|\sys^{\paramdelay}| =  \paramdelay \times |Q| \times \sum_{1}^{\paramdelay} \paramlabel = \paramdelay^2 \times |Q| \times \paramlabel.
\]
For a given expression $\vars{e}$, we use $|\vars{e}|$ to denote the size of $\vars{e}$, i.e., the number of atoms in $\vars{e}$.
Given a memory $\mem$, the complexity of function $\seval(\vars{e}, \mem)$ is the cost of $\fsimp(\rw(\vars{e}, \mem))$.
Function $\rw(\vars{e}, \mem)$ looks up each atom in $\vars{e}$ in $\mem$ and attempts to replace it by its truth-value.
The cost of a memory lookup is $\Theta(1)$, and the replacement is linear in the number of atoms in $\vars{e}$.
It effectively takes one pass to syntactically replace all atoms by their values, therefore the cost of $\rw$ is $\Theta(|\vars{e}|)$.
However, applying the function $\fsimp()$ requires solving the Minimum Equivalent Expression problem which is $\Sigma_2^p$-complete~\cite{CircuitMin}, it is exponential in the size of the expression, making it the most costly function.
$|\vars{e}|$ is bounded by $\paramdelay \paramlabel$.
Function $\sreach()$ requires evaluating every expression in the \MemRep{}.
For each timestamp we need at most $|Q|$ expressions, and the number of timestamps is bounded by $\paramdelay$.
\paragraph{Memory.}
The memory required to store $\mem$ depends on the trace, namely the amount of observations per component.
Recall that once a state is known, observations can be removed, the number of timestamps is bounded by $\paramdelay$.
The size of the memory is then:
\[
\sum_{t=i}^{i+\paramdelay} |\trace(c,t)| \times (\fenc{\timestamp} \times \fenc{\AP} \times \fenc{\verdictb}).
\]
%
\subsection{Analyzing Existing Algorithms} \label{sec:analysis:algs}
%
We now shift the focus to the algorithms and their usage of the data structures.
We begin by presenting an overview of the abstract phases performed by decentralized monitoring algorithms.
We then elaborate on our approach to model their behavior.
Finally, we present the analysis for each of the algorithms adapted from~\cite{DecentMon}.

\paragraph{Overview.}
A decentralized monitoring algorithm consists of two steps: setting up the monitoring network, and monitoring.
In the first step, an algorithm initializes the monitors, defines their connections, and attaches them to the components.
We represent the connections between the various monitors using a directed graph $\tuple{\mons, E}$ where $E = 2^{\mons \times \mons}$ defines the edges  describing the sender-receiver relationship between monitors.
For example, the network $\tuple{\setof{\MONID_0, \MONID_1}, \setof{\tuple{\MONID_1, \MONID_0}}}$ describes a network consisting of two monitors $\MONID_0$ and $\MONID_1$ where $\MONID_1$ sends information to $\MONID_0$.
In the second step, an algorithm proceeds with monitoring, wherein each monitor processes observations and communicates with other monitors.

We consider the existing three algorithms: Orchestration, Migration and Choreography~\cite{DecentMon} adapted to use \MemRep{}.
We note that these algorithms operate over a global clock, therefore the sequence of steps can be directly mapped to the timestamp.
We choose an appropriate encoding of $\atoms$ to consist of a timestamp and the atomic proposition ($\atoms = \timestamp \times \AP$).
These algorithms are originally presented using an LTL specification instead of automata, however, it is possible to obtain an equivalent Moore automaton as described in \cite{LTL3Tools}.
\paragraph{Approach.}
A decentralized monitoring algorithm consists of one or more monitors that use the $\MemRep{}$ and memory data structures to encode, store, and share information.
By studying  $\paramdelay$, we derive the size of the $\MemRep{}$ and the memory a monitor would use (see \rsec{sec:analysis:ds}) .
Knowing the sizes, we determine the computation overhead of a monitor, since we know the bound on the number of simplifications a monitor needs to make ($\paramdelay |Q|)$, and we know the bounds on the size of the expression to simplify ($\paramdelay \paramlabel$).
Once the cost per monitor is established, the total cost for the algorithm can be determined by aggregating the costs per monitors.
This can be done by summing to compute total cost or by taking the maximum cost in the case of concurrency following the Bulk Synchronous Parallel (BSP)~\cite{Valiant90} approach.
\paragraph{Orchestration.}
The orchestration algorithm ({\algorch}) consists in setting up a main monitor which will be in charge of monitoring the entire specification.
However since that monitor cannot access all observations on all components, orchestration introduces one monitor per component to forward the observations to the main monitor.
Therefore, for our setup,  we consider the case of a main monitor $\MONID_0$ placed on component $c_0$ which monitors the specification and $|\comps| - 1$ forwarding monitors that only send observations to $\MONID_0$ (labeled $\MONID_k$ with $k \in [1, |\comps|]$).
We consider that the reception of a message takes at most $\vars{d}$ rounds.
The information delay $\paramdelay$ is then constant, $\paramdelay = \vars{d}$.
The number of messages sent at each round is $|\comps| - 1$, i.e., the number of forwarding monitors sending their observations.
The size of a message is linear in the number of observations for the component, for a forwarding monitor labeled with $\MONID_k$, the size of the message is $|\trace(t, c_k)| \times (\fenc{\timestamp} \times \fenc{\AP} \times \fenc{\verdictb})$.

\paragraph{Migration.}
The migration algorithm ({\algmigr}) initially consists in rewriting a formula and migrating from one or more component to other  components to fill in missing observations.
We call the monitor rewriting the formula the active monitor.
Our $\MemRep{}$ encoding guarantees that two monitors receiving the same information are in the same state.
Therefore, monitoring with Migration  amounts to rewriting the $\MemRep{}$ and migrating it across components.
Since all monitors can send the \MemRep{} to any other monitor, the monitor network is a strongly-connected graph.
In {\algmigr}, the delay depends on the choice of function $\migrchc$, which determines which component to migrate to next upon evaluation.
By using a simple function $\migrchc$, which causes a migration to the component with the atom with the smallest timestamp, it is possible to view the worst case as an expression where for each timestamp we depend on information from all components, therefore $|\comps| - 1$ rounds are necessary to get all the information for a timestamp ($\paramdelay = |\comps| - 1$).
We parametrize Migration by the number of active monitors at a timestamp $m$.
The presented function $\migrchc$ in~\cite{DecentMon}, selects at most one other component to migrate to.
Therefore, after the initial choice of $m$, subsequent rounds can have at most $m$ active monitors.

We illustrate {\algmigr} in Algorithm~\ref{alg:migration}.
The state of a migration monitor consists of a variable  $\mrm{isActive}$ that determines whether or not the monitor is active, and $\sys$ that is an \MemRep{} encoding the same automaton shared by all monitors.
At each round the monitor receives a timestamp $t$ and a set of observations $o$.
Line 2 displays the memory update with observations for that round.
Lines 3 to 10 describe the reception of $\MemRep{}$s from other monitors.
Upon receiving an $\MemRep{}$, the monitor state is set to active (Line 7).
An active monitor will then update its $\MemRep{}$ by first ensuring that it is expanded to the current timestamp using $\smove$ (Line 11),
then rewriting and evaluating each entry (Lines 12-17).
The number of entries in the $\MemRep{}$ depends on $\paramdelay$.
If any of the entries is evaluated to a final verdict (Line 14), then the verdict is found and we terminate.
While the verdict is not found, the migration algorithm first removes all unnecessary entries in the $\MemRep{}$ (Line 18).
Unnecessary entries are entries for which the state is known, the last known state is only kept, all previous timestamps are removed.
After removing unnecessary entries, we determine a new monitor to continue monitoring using the function $\migrchc$ (Lines 19-22).
The initial choice of active monitors is bounded by $m \leq |\comps|$.
Since at most $m - 1$ other monitors can be running, there can be $(m - 1)$ merges.
The size of the  resulting $\MemRep{}$ is  $m \times |\sys^\paramdelay| = m (|\comps| - 1)^2 |Q| \paramlabel$.
In the worst case, the upper bound on the size of $\MemRep{}$ is $(|\comps| - 1)^3 |Q| \paramlabel$.
The number of messages is bounded by the number of active monitors $m$.
The size of each message is however the size of the $\MemRep{}$, since {\algmigr} requires the entire $\MemRep{}$ to be sent.

\begin{algorithm}[!tp]
 \caption{Migration}
 \label{alg:migration}

\begin{algorithmic}[1]
\Procedure{Migration}{$t, o$}
	\State $\mem \gets \mem \memadd \cons(o, \fts_t)$ \Comment{Add observations to memory}
	\While{Received $\sys'$} \Comment{Received an $\MemRep{}$ from another monitor}
		\If{$\mrm{isActive}$}  \Comment{If the monitor is active, the monitor \MemRep{} has information}
			\State $\sys \gets \sys \sysadd \sys'$\Comment{Merge information with existing information}
		\Else
			\State $\sys \gets \sys'$; $\mrm{isActive} \gets \vt$ \Comment{Monitor becomes active after it receives an \MemRep{}}
		\EndIf
	\EndWhile
	\If{$isActive$}
		\State $t' \gets  \mathrm{getEnd}(\sys)$; $\sys \gets \smove(\sys, t', t)$	 \Comment{Build $\MemRep{}$ up to current timestamp}
		\ForAll{$t_v \in \fdom(\sys)$} \Comment{Go through all $\MemRep{}$ timestamps}
			\State $v \gets \verdictAt(\sys, \mem, t_v)$  \Comment{Evaluate the entries associated with timestamp $t_v$}
			\If{$v \in \verdictb$} \Comment{Found a final verdict}
					\State Report $v$ and terminate
			\EndIf
		\EndFor
		\State $\sys \gets \mathrm{dropResolved}(\sys)$ \Comment{Purge $\MemRep{}$ of non-needed entries}
		\State $c_k \gets \migrchc(\sys)$ \Comment{Determine the next component}
		\If{$c_k \neq c$} \Comment{Is the next component not local}
			\State $\mrm{isActive} \gets \vf$; Send $\sys$ to $\MONID_k$ \Comment{Send to relevant monitor, stop monitoring}
		\EndIf
	\EndIf
\EndProcedure
\end{algorithmic}

\end{algorithm}

\paragraph{Choreography.}
Choreography ({\algchor}) presented in~\cite{DecentMon,BauerF12} splits the initial LTL formula into subformulas and delegates each subformula to a component.
Thus {\algchor} can illustrate how it is possible to monitor decentralized specifications.
Once the subformulas are determined by splitting the main formula~\footnote{Details of the generation is provided in~\rapp{sec:app:choreo}.}, we adapt the algorithm to generate an automaton per subformula to monitor it.
To account for the verdicts from other monitors, the set of possible atoms is extended to include the verdict of a monitor identified by its id.
Therefore, $\atoms = (\timestamp \times \AP) \cup (\mons \times \timestamp)$.
Monitoring is done by replacing the subformula by the id of the monitor associated with it.
Therefore, monitors are organized in a tree, the leafs consisting of monitors without any dependencies, and dependencies building up throughout the tree to reach the main monitor that outputs the verdict.
Since each monitor is in charge of evaluating a subformula, the monitors communicate the evaluation of the formula as a verdict $\verdictb$ when it is resolved.
Furthermore, monitors may instruct other monitors to stop monitoring as they are no longer necessary.
The two messages are referred to as $\msgverdict$ and $\msgkill$, respectively.
For each monitor labeled $\ell \in \APmons$ we determine the set $\chorcoref_\ell \in 2^\APmons$ which contains the labels of monitors that send their verdicts to monitor $\aut_\ell$.
%
The information delay for a monitor is thus dependent on its depth in the network tree.
The depth of a monitor labeled $\ell$ that depends on the set of monitors $\chorcoref_{\ell}$,  is computed recursively as follows:
\[
	\func{depth}(\ell) = \funcparts{
	1 & \pif \monitorable(\aut_\ell) \land \chorcoref_\ell = \emptyset, \\
		1 + \func{max}(\setof{\func{depth}(\ell') \mid \ell' \in \chorcoref_\ell}) & \pif \monitorable(\aut_\ell) \land \chorcoref_\ell \neq \emptyset, \\
		\infty & \pelse.
	}
\]
%
A monitor synthesized by a non-monitorable specification will never emit a verdict, therefore its depth is $\infty$.
A leaf monitor has no dependencies, its depth is 1.
Since the depth controls the information delay ($\paramdelay$), it is possible in the case of choreography to obtain a large $\MemRep{}$ depending on the specification.
In effect, the worst case the size of the \MemRep{} can be linear in the size of the trace $\paramdelay = |\trace|$, as it will be required to store the $\MemRep{}$ until the end of the trace.
As such properties of the specification such as monitorability (see Section~\ref{sec:dmon:monitorability}) impact greatly the delay, and thus performance.
In terms of communication, the number of monitors generated determines the number of messages that are exchanged.
%
By using the naive splitting function (presented in~\cite{DecentMon}), the number of monitors depends on the size of the LTL formula.
Therefore, we expect the number of messages to grow with the number of atomic propositions in the formula.
By denoting $|E|$ the number of edges between monitors, we can say that the number of messages is linear in $|E|$.
The size of the messages is constant, it is the size needed to encode a timestamp, id and a verdict in the case of $\msgverdict$, or only the size needed to encode an id in the case of $\msgkill$.

\gettable{
	\caption{Scalability of Existing Algorithms.}
	\label{tbl:analysis:algs}
	\begin{tabular}{|l|c|c|c|}
		\hline \textbf{Algorithm} & \boldmath$\paramdelay$ & \textbf{\# Msg} & \boldmath$|\mathrm{Msg}|$\\
		\hline \hline Orchestration & $\Theta(1)$ &  $\Theta(|\comps|)$ & $O(\AP_c)$\\
		\hline Migration & $O(|\comps|)$ & $O(m)$ & $O(m|\comps|^2)$\\
		\hline Choreography & $O(\func{depth}(\mroot) + |\trace|)$ & $\Theta(|E|)$ & $\Theta(1)$\\
		\hline
	\end{tabular}
	}
\paragraph{Discussion.}
We summarize the main parameters that affect the algorithms in \rtbl{tbl:analysis:algs}.
%
This comparison could serve as a guide to choose which algorithm to run based on the environment (architectures, networks etc).
For example, on the one hand, if the network only tolerates short message sizes but can support a large number of messages, then {\algorch} or {\algchor} is preferred over {\algmigr}.
On the other hand, if we have heterogeneous nodes, as is the case in the client-server model, we might want to offload the computation to one major node, in this scenario {\algorch}  would be preferable as the forwarding monitor require no computation.
This choice can be further impacted by the network topology.
In a ring topology for instance, one might want to consider using Migration (with $m = 1$), as using {\algorch} might impose further delay in practice to relay all information, while in a star topology, using {\algorch} might be preferable.
In a more hierarchical network, {\algchor} can adapt its monitor tree to the hierarchy of the network.
Since we perform a worst-case analysis, we investigate the trends shown in \rsec{sec:experiment} by simulating the behavior of the algorithms on a benchmark consisting of randomly generated specifications and traces.
Furthermore, we use a real example in \rsec{sec:chiron} to refine the comparison by looking at six different specifications.

\section{The \THEMIS{} Framework} \label{sec:tool}
%
%
\THEMIS{} is a framework to facilitate the design, development, and analysis of decentralized monitoring algorithms; developed using Java and AspectJ~\cite{AspectJ} ($\sim$5700 LOC).\footnote{The \THEMIS{} framework is further described and demonstrated in the tool-demonstration paper~\cite{isstademo} and on its Website~\cite{themisweb}.}
It consists of a library and command-line tools.
The library provides all necessary building blocks to develop, simulate, instrument, and execute decentralized monitoring algorithms.
The command-line tools provide basic functionality to generate traces, execute a monitoring run and execute a full experiment (multiple parametrized runs).

The purpose of \THEMIS{} is to minimize the effort required to design and assess decentralized monitoring algorithms.
\THEMIS{} provides an API for monitoring and necessary data structures to load, encode, store, exchange, and process observations, as well as manipulate specifications and traces.
These basic building blocks can be reused or extended to modify existing algorithms or design new more intricate algorithms.
To assess the behavior of an algorithm, \THEMIS{} provides a base set of metrics (such as messages exchanged and their size, along with computations performed), but also allows for the definition of new metrics by using the API or by writing custom AspectJ instrumentation.
These metrics can be used to assess existing algorithms as well as newly developed ones.
Once algorithms and metrics are developed, it is possible to use existing tools to perform monitoring runs or full experiments.
Experiments are used to define sets of parameters, traces and specifications.
An experiment is effectively a folder containing all other necessary files.
By bundling everything in one folder, it is possible to share and reproduce the experiment\footnote{Experiments provided in this paper are provided at~\cite{themisartifact}, earlier experiments are provided at~\cite{themisweb}}.
After running a single run or an experiment, the metrics are stored in a database for postmortem analysis.
These can be queried, merged or plotted easily using third-party tools.
After completing the analysis, algorithms and metrics can be tuned so as to refine the design as necessary.

The \THEMIS{} framework has been improved since~\cite{themisissta,isstademo} to support fully distributed and multi-threaded support for monitoring by adding the tool \code{Node} that acts as a runtime.
One or more nodes can be deployed on a given platform.
A node receives information (via commands) to deploy components, and monitors on the current platform.
Each component contains one or multiple peripheries.
A periphery is an input stream to the component, that generates observations.
Peripheries follow a stream interface, and waits on a \code{next()} call to generate the next observations.
Monitors are attached to components, and receive the observations that components receive.
Thus, a node follows a publish subscribe model.
Components can be seen as topics, where a monitor registers to a topic.
Peripheries produce a stream of observations for components.
Peripheries can include reading traces from files, over network sockets, or be generated in a stream.
Nodes can communicate with other nodes in a distributed manner, through sockets.
The implementation of a node defines the high level assumptions of monitoring, for example, our round-based monitoring approach is implemented as a node.
For our implementation of a node, reading peripheries to generate component observations is done in parallel.
Once all peripheries have executed their $\code{next()}$ call to read the next events on the stream, the observations are aggregated in an event that is sent to monitors associated with the component.
Monitors execute in parallel, once all monitors have completed for a given round, the next round begins.
This behavior could be altered to ignore rounds, and simply monitor as soon as information is available by using a different node implementation.
Measures have been updated to  be thread-safe and work at the node-level.

Furthermore, the simplification logic (operation $\seval$) for the data structure \MemRep{} has been greatly improved to call the simplifier less and be more aggressive with the simplifications.
This yields on average, a  much smaller \MemRep{}, more details are provided in~\rapp{sec:app:newthemis}.

\section{Comparing Algorithms with \THEMIS{}}
\label{sec:experiment}
%
%
We use \THEMIS{} to compare adapted versions of existing algorithms ({\algorch}, {\algmigr}, and {\algchor} - \rsec{sec:analysis}) and study the behavior of the \MemRep{} data structure to validate the trends presented in the analysis in \rsec{sec:analysis}.
Furthermore, since the analysis presented worst-case scenarios, we look at the usefulness of simulation to determine the advantages or disadvantages of certain algorithms in specific scenarios.
The \THEMIS{} tool, the data for both scenarios used in this paper, the scripts used to process the data, and the full documentation for reproducing the experiments is found at~\cite{themisartifact}.
\paragraph{Overview of scenarios.}
We additionally consider a round-robin variant of {\algmigr}, {{\algmigrr}}, and use that for analyzing the behavior of the migration family of algorithms as it has a predictable heuristic (function $\migrchc$).
We compare the algorithms under two scenarios.
The first scenario explores synthetic benchmarks, that is, we consider random traces and specifications.
This allows us to account for different types of behavior.
The second scenario explores a specific example associated with a common pattern in programming.
For that, we consider a publish-subscribe system, where multiple publishers subscribe to a channel (or topic), the channel publishes events to the subscribers.
We use the Chiron user interface example~\cite{chiron,chironweb}, along with the specifications formalized for it~\cite{ltlpatterns}.
\paragraph{Monitoring metrics.}
The first considered metric is that of information delay ($\paramdelay$) (\rsec{sec:analysis:ds}).
The information delay impacts the size of the $\MemRep{}$ and therefore the computation, communication costs to send an \MemRep{} structure, and also the memory required to store it.
To compute the average information delay, we first consider the timestamp difference when an $\MemRep{}$ is resolved (i.e., it indicates a state).
We sum these differences across the entire run and count the number of resolutions.
As such, we acquire the average number of timestamps stored in an \MemRep{}.
We notice that it is possible for delay to fall below 1, as some traces can cause some monitors to emit a verdict at the very first timestamp.
By considering our analysis in \rsec{sec:analysis}, we split our metrics into two main categories: computation and communication.
The \MemRep{} structure requires the evaluation and simplification of a Boolean expression which is costly (see \rsec{sec:cmon:ehe}).
To measure computation, we can count the number of expressions evaluated (using memory lookup), and the number of calls to the simplifier.
For this experiment we consider the calls to the simplifier.
Since algorithms may have more than one monitor active, we consider for a given round the monitor with the most simplifications.
We sum the maximum number of simplifications per round across all the rounds, and then  normalize by the number of rounds.
This allows us to determine the slowest monitor per round, as other monitors are executing in parallel.
Therefore, we determine the bottleneck.
We refer to this metric as critical simplifications.
This can be similarly done by considering the number of expressions evaluated.
Since monitors can execute in parallel, we introduce \emph{convergence} as a metric to capture load balancing across a run of length $n$, where:
\[
  \func{conv}(n) = \dfrac{1}{n} \sum\limits_{t=1}^{n}\left(\sum\limits_{c \in \comps} \left(\frac{s^t_c}{{s}^t} - \frac{1}{|\comps|}\right)^2\right) \text{, with } {s}^t = \sum\limits_{c\in\comps} s^t_c.
\]
At a round $t$, we consider all simplifications performed on all components ${s}^t$ and for a given component $s^t_c$.
Then, we consider the ideal scenario where computations have been spread evenly across all components.
Thus, the ideal ratio is $\frac{1}{|\comps|}$.
We compute the ratio for each component ($\frac{s^t_c}{{s}^t}$), then its distance to the ideal ratio.
Distances are added for all components across all rounds then normalized by the number of rounds.
The higher the convergence the further away we are from having all computations spread evenly across components.
Convergence can also be measured similarly on evaluated expressions.
We consider communication using three metrics: number of messages, total data transferred, and the data transferred in a given message.
The number of messages is the total messages sent by all monitors throughout the entire run.
The data transferred consists of the total size of messages sent by all monitors throughout the entire run.
Both the number of messages and the data transferred are normalized using the run length.
Finally, we consider the data transferred in a given message to verify the message sizes.
To do so, we normalize the total data transferred using the number of messages.
%
\subsection{Synthetic Scenario}
\label{sec:experiment:synthetic}
%
\paragraph{Experimental setup.}
We generate the specifications as random LTL formulas using \vars{randltl} from \vars{Spot}~\cite{SPOT} then converting the LTL formulae to automata using \vars{ltl2mon}~\cite{LTL3Tools}.
We generate traces by using the \vars{Generator} tool in \THEMIS{} which generates synthetic traces using various probability distributions (provided by COLT\footnote{COLT provides a set of Open Source Libraries for High Performance Scientific and Technical Computing in Java.\cite{COLT}}).
For all algorithms we considered the communication delay to be 1 timestamp.
That is, messages sent at $t$ are available to be received at most at $t+1$.
In the case of migration, we set the active monitors to 1 ($m = 1$).
For our experiment, we use 200 traces of 60 events per component, we associate with each component 2 observations.
Traces are generated using 4 probability distributions (50 traces for each probability distribution).
The used distributions include \emph{normal} ($\mu=0.5, \sigma^2=1$), \emph{binomial} ($n=100, p=0.3$), and two \emph{beta} distributions: \emph{beta-1} ($\alpha=2, \beta=5$), and \emph{beta-2} ($\alpha=5, \beta=1$).
The varied distributions provide different probability to assign $\vt$ and $\vf$ to observations in the traces, as such we achieve varied coverage\footnote{
An observation is assigned $\vt$ if the generated number is strictly greater than $0.5$, and is otherwise $\vf$.
For the binomial distribution, we consider $p=0.3$ the probability of obtaining $\vt$.
}.
We vary the number of components between 3 and 5, and ensure that for each number we have 100 formulae that reference all components.
We were not able to effectively use a larger number of components since most formulae become sufficiently large that generating an automaton from them using \code{ltl2mon} fail.
The generated formulae were fully constructed of atomic propositions, there were no terms containing $\vt$ or $\vf$\footnote{To generate formulae with basic operators, string \vars{$\vf$=0,$\vt$=0,xor=0,}. \vars{M=0,W=0,equiv=0,implies=0,ap=6,X=2,R=0} is passed to \vars{randltl}.}
When computing sizes, we use a normalized unit to separate the encoding from actual implementation strategies.
Our assumptions on the sizes follow from the bytes needed to encode data (for example: 1 byte for a character,  4 for an integer).
We normalized our metrics using the length of the run, that is, the number of rounds taken to reach the final verdict (if applicable) or timeout, as different algorithms take different numbers of rounds to reach a verdict.
In the case of timeout, the length of the run is 65 (length of the trace, and 5 additional  timestamps to timeout).
%


\begin{figure}[t!]
 \subfloat[][Average delay]{\includegraphics[width=.5\textwidth]{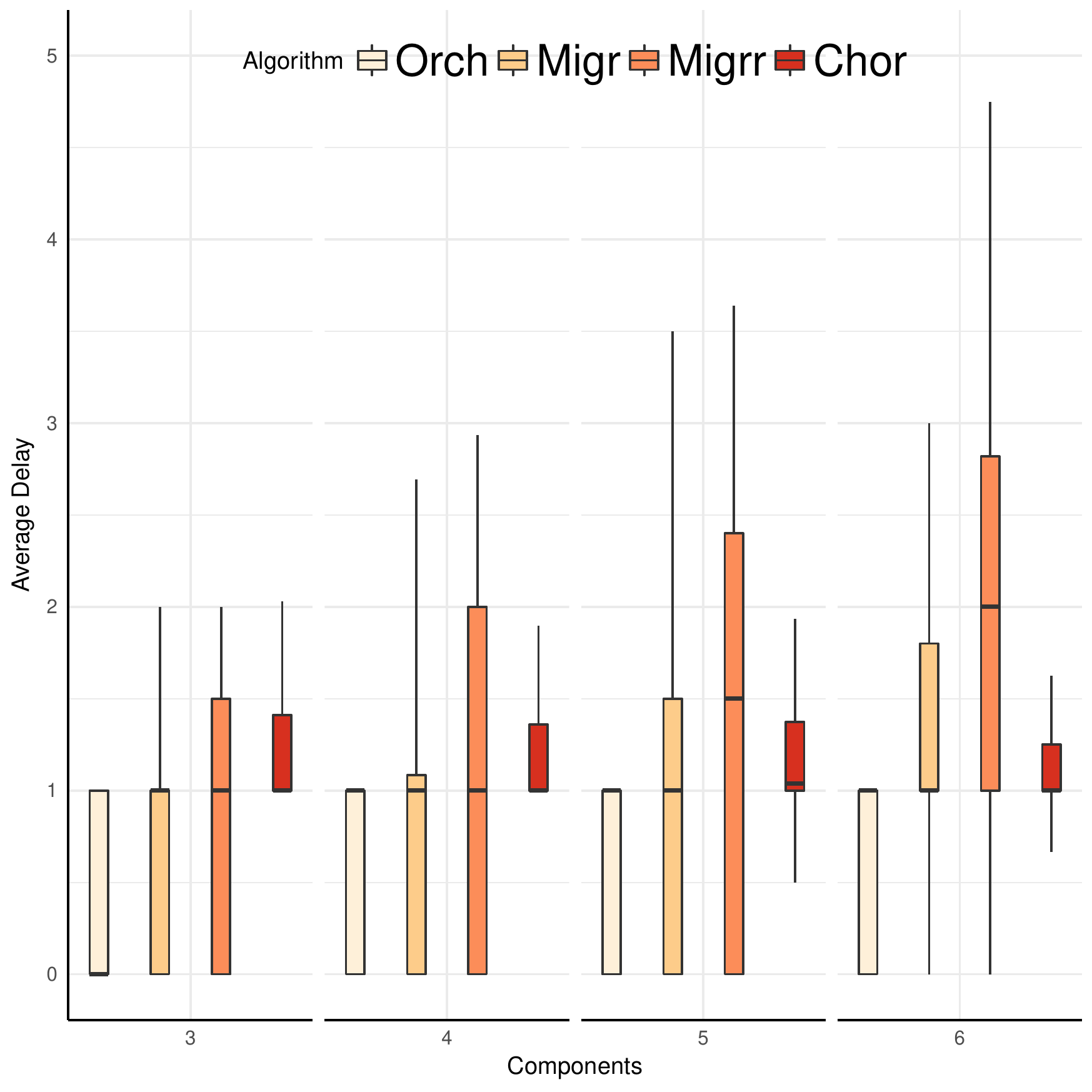} \label{fig:syn:delay}}
 \subfloat[][Critical simplifications]{\includegraphics[width=.5\textwidth]{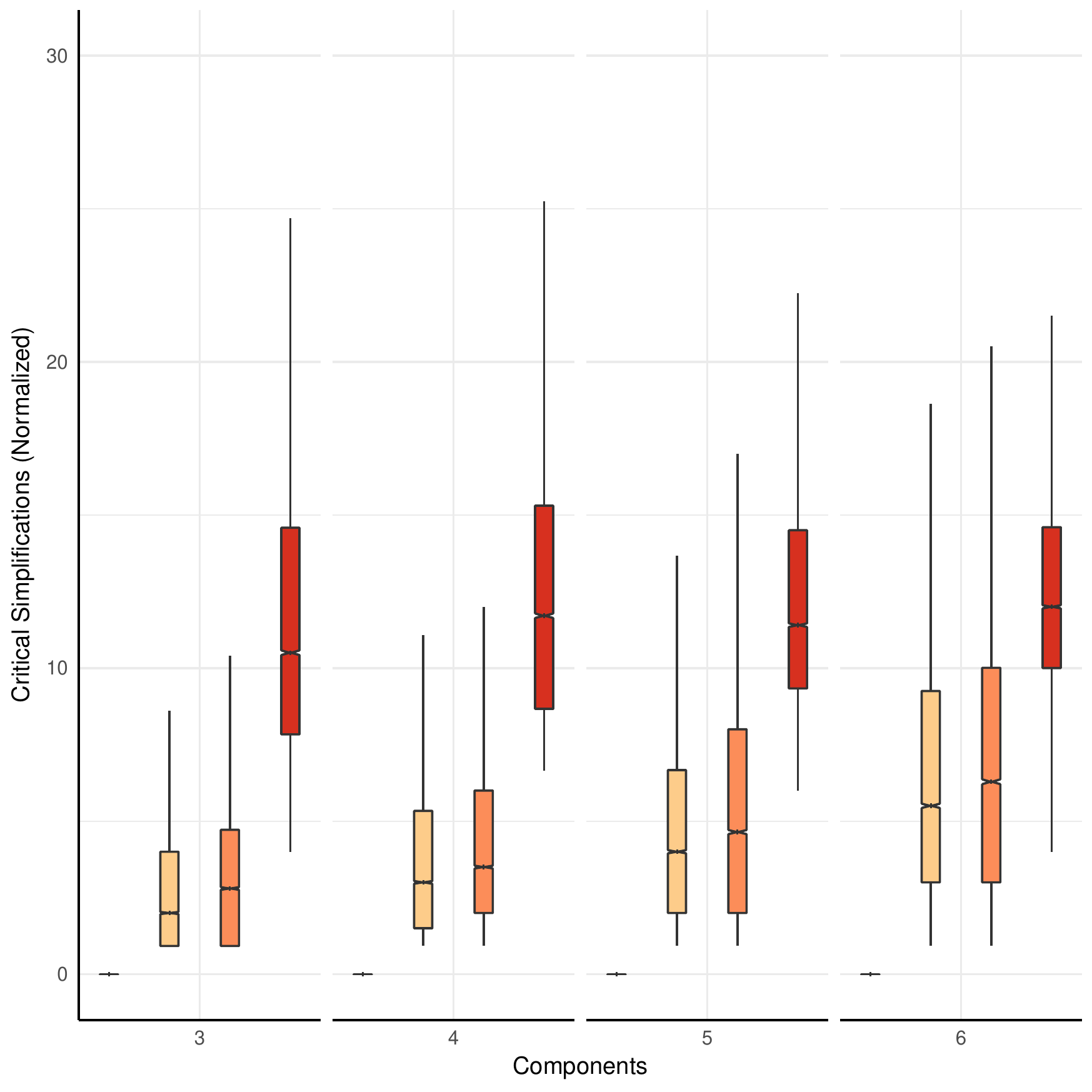} \label{fig:syn:simp}}\\
 \subfloat[][Convergence]{\includegraphics[width=.5\textwidth]{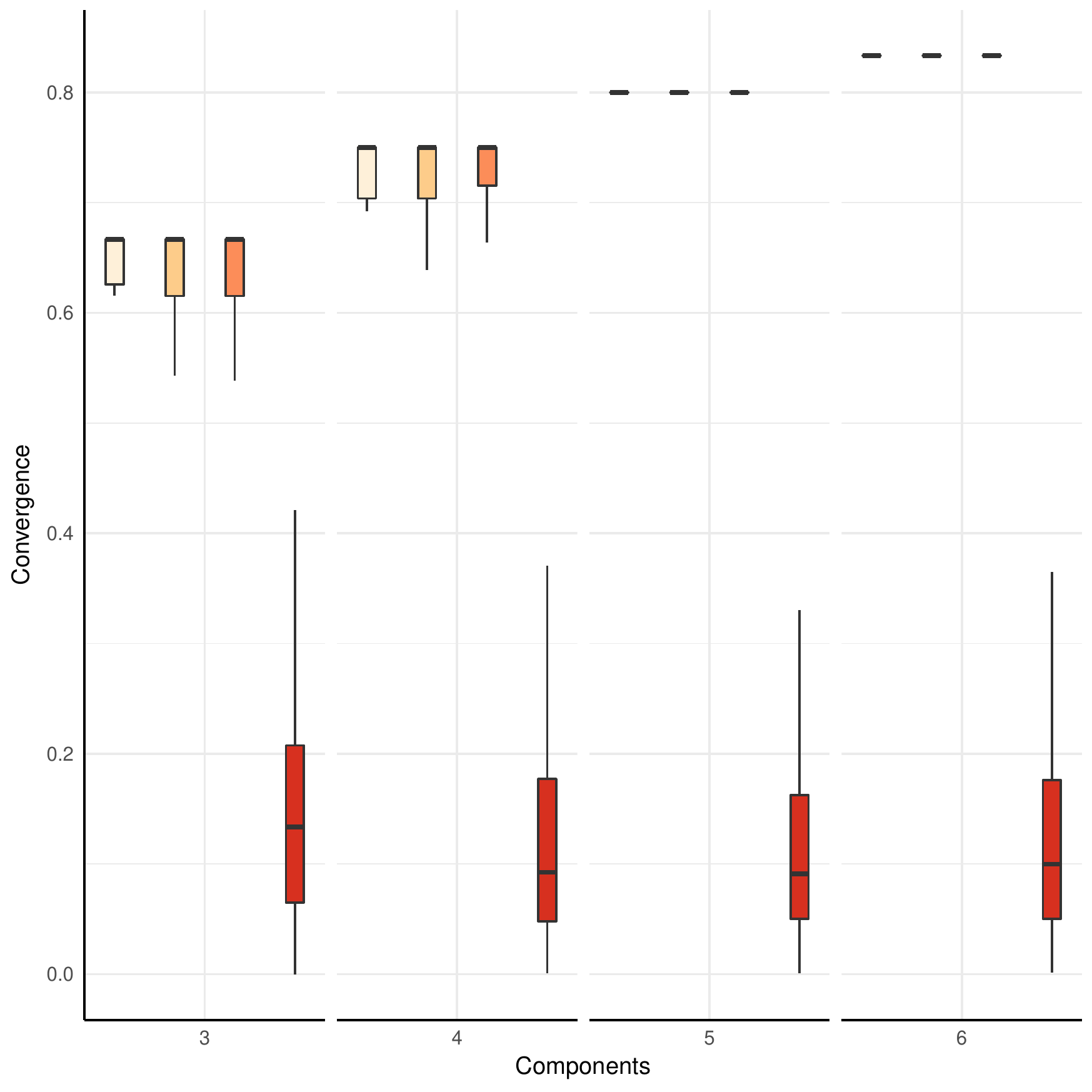} \label{fig:syn:conv}}
  \subfloat[][Number of messages]{\includegraphics[width=.5\textwidth]{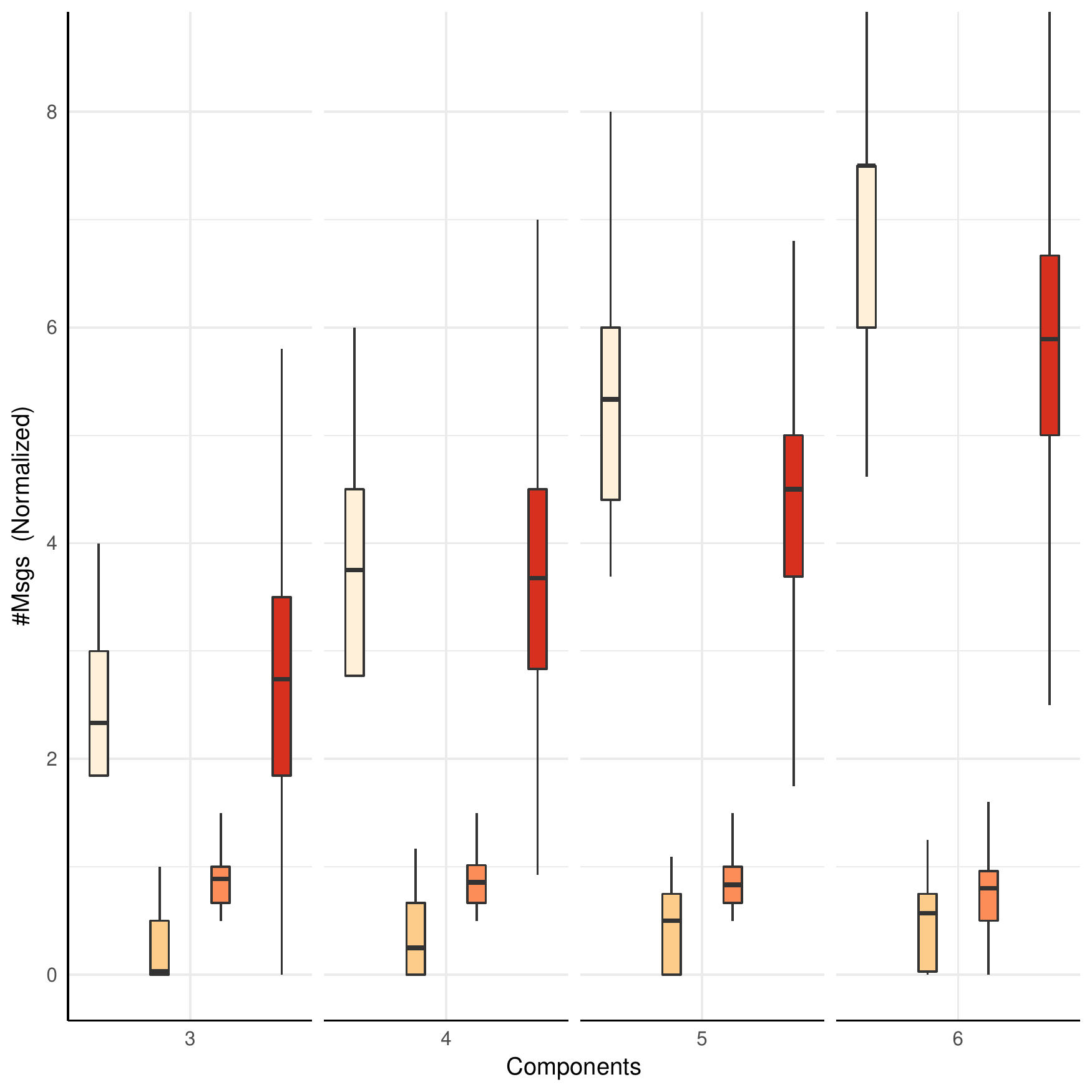} \label{fig:syn:msgnum}}
  \caption{Comparison of delay, computation and number of messages. Algorithms are presented in the following order: \algorch{}, \algmigr{}, \algmigrr{}, \algchor{}.}
  \label{fig:syn:computation}
\end{figure}

\begin{figure}[t!]
 \subfloat[][Total data transferred]{\includegraphics[width=.5\textwidth]{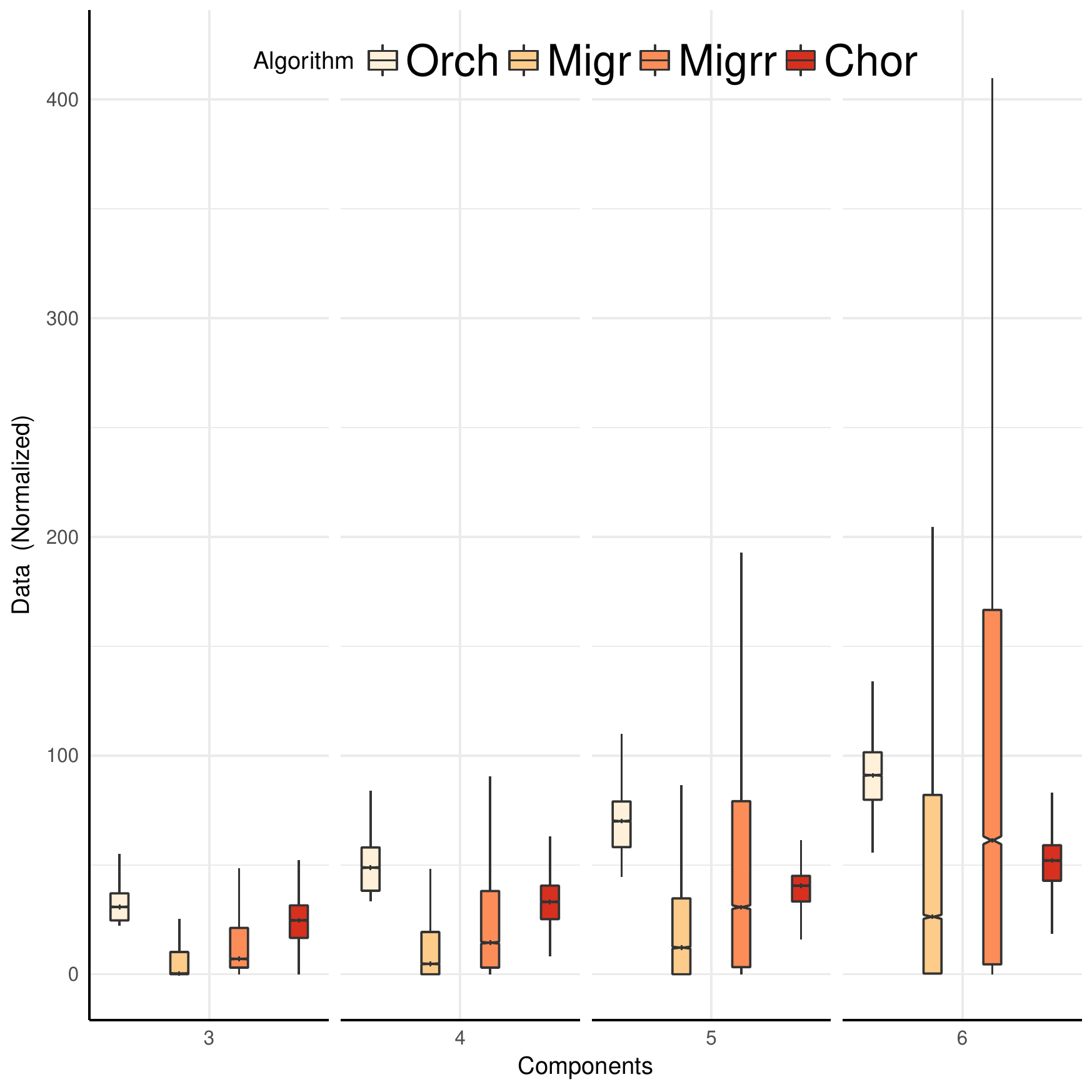} \label{fig:syn:data}}
 \subfloat[][Data per message]{\includegraphics[width=.5\textwidth]{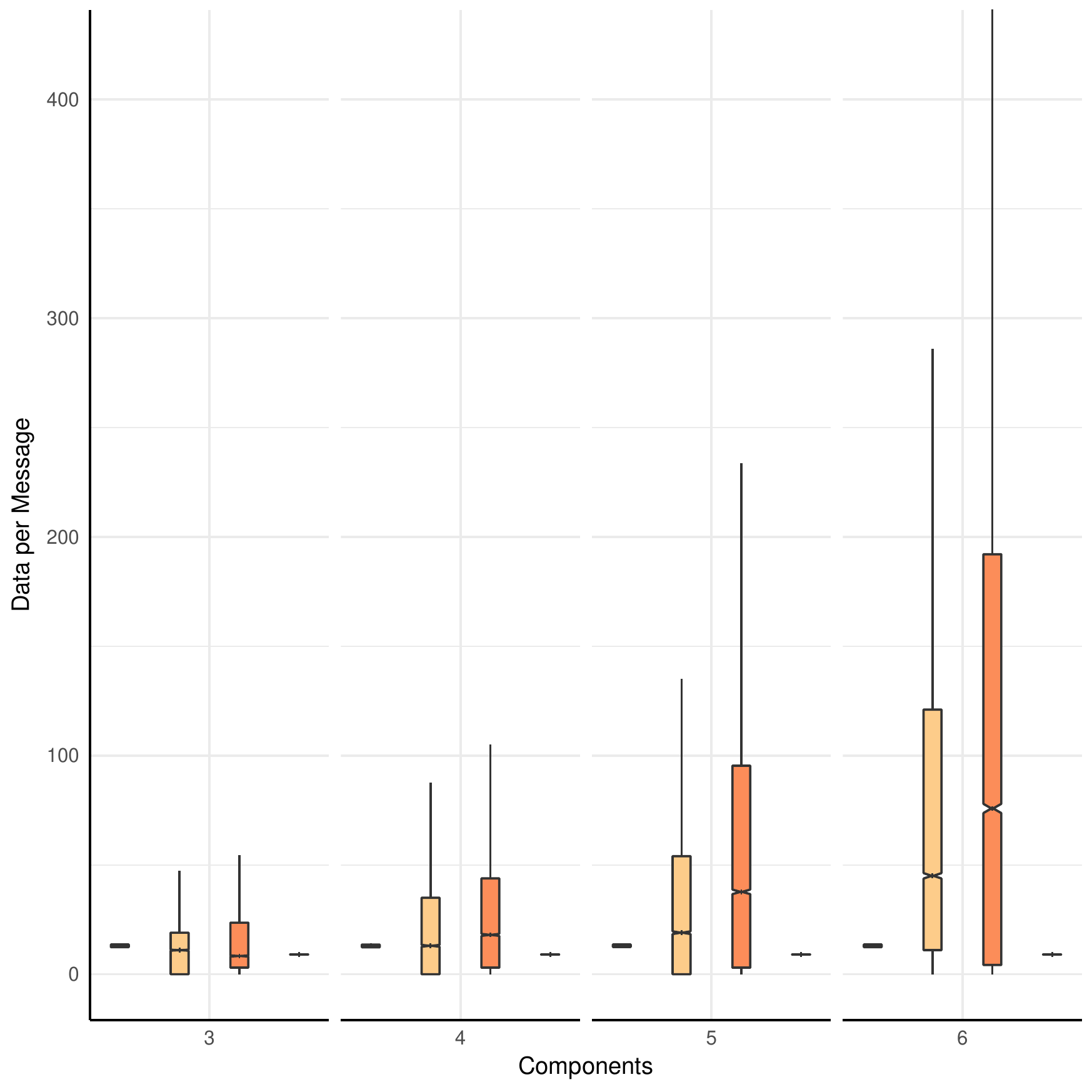} \label{fig:syn:datanorm}}
  \caption{Data Transfer}
  \label{fig:syn:datafig}
\end{figure}

\paragraph{Comparing algorithms.}
Figures~\ref{fig:syn:computation} and \ref{fig:syn:datafig} present the outcome of the proposed metrics for the algorithms.
We inspect the behavior of information delay in \rfig{fig:syn:delay} by computing the average information delay.
As expected, orchestration never exceeds a delay of 1.
For migration, the delay depends on the heuristic used, as mentioned in \rsec{sec:analysis:algs}, its worst case is the number of components.
Migration can still have a lower delay than orchestration in some cases (as observed for $|C| \geq 4$).
This observation is due to the initial monitor placement, as in our case we chose the first component to be always where we place the main orchestration monitor (component \emph{A}), while for migration, the heuristic function ($\migrchc{}$) decides which monitor starts.
As such, in a specification where the verdict can be resolved at the first timestamp, migration has an advantage.
For \algchor{}, the delay is at least 1, as the network depth affects the delay.
Furthermore, we notice that the delay for \algchor{} is not particularly affected with the number of components.
We know that its worst-case will depend on traces in cases of non-monitorability, we inspect that further in \rsec{sec:chiron}.
\rfig{fig:syn:simp} shows the average maximum computation done by a monitor for a given round.
By looking at computation, we notice that \algorch{} performs no simplifications.
This is the case as expressions in the \MemRep{} do not become sufficiently complex to require simplification.
We recall that for orchestration, the memories of all local observations are sent to the main monitor within one timestamp.
And as such, by memory lookup, the expression is immediately evaluated without the need to simplify.
We notice that for the average case, \algmigr{} performs a small amount of simplifications, and \algchor{} still executes a reasonable number of simplifications.
\rfig{fig:syn:conv} shows the convergence for the algorithms.
Since \algchor{} is the only algorithm that performs computations at different components in a given round, we notice that the convergence is much lower.

For communication, we first consider the number of messages transferred normalized by the length of the run.
We notice that for algorithms \algorch{} and \algchor{} the number of messages increases with the number of components.
Since \algchor{} depends on the edges that connect monitors, it scales better with the number of components than \algorch{} (the depth of the network is usually smaller than the number of components).
In contrast, we notice that for \algmigr{} and \algmigrr{}, the number of messages is independent from the number of components, as it depends on the number of active monitors.
\rfig{fig:syn:data} presents the total data transferred normalized by the run length.
We notice by examining algorithm \algorch{} that sending all observations can be costly.
Algorithms \algmigr{} and \algmigrr{} are capable of sending much less data on average, but have variable behavior, and scale poorly, we notice an increase as $|\comps|$ increases.
 Algorithm \algchor{} performs better than \algorch{}, and scales much better with component size.
 We notice that while \algmigr{} and \algmigrr{} send less messages than the other algorithms, and have better scaling in the number of messages transferred, they can still, in total, send more data depending on the traces and specification.
 We notice that the 75\% quartile for \algmigrr{} still exceeds that of \algorch{}.
 Since total data transferred includes both the number of messages and their sizes, we present the size of the message in \rfig{fig:syn:datanorm} by dividing the total data transferred by the number of messages.
 We observe that for \algorch{} and \algchor{} the size of a message is constant, not very variable and does not depend on $|\comps|$,
 while for \algmigr{} and \algmigrr{} we observe quadratic scaling.
 We recall from \rsec{sec:analysis:algs} that the migration algorithms send the \MemRep{} which grows quadratic in the size of the information delay.

\paragraph{Comparing variants.}
Using the same dataset, we look at another use-case of \THEMIS{}; that of comparing variants of the same algorithm.
In this case, we focus on differences between \algmigr{} and \algmigrr{}.
The heuristic of \algmigr{} improves on the round-robin heuristic of \algmigrr{} by choosing to transfer the \MemRep{} to the component that can observe the atomic proposition with the earliest timestamp in the \MemRep{} (referred to as earliest obligation~\cite{DecentMon}).
Using the simple heuristic, we notice a drop in the delay starting from $|\comps| > 4$ (\rfig{fig:syn:delay}).
The simple heuristic of earliest obligation seems to reduce on average the delay of the algorithm, interestingly, it maintains a mean of 1.
Furthermore, we observe a drop in both messages transferred (\rfig{fig:syn:data}) and size of messages (\rfig{fig:syn:datanorm}).
Consequently, this constitutes a drop in the total data transferred (\rfig{fig:syn:data}).
We note that the message size is also the size of the \MemRep{}.
The drop in the number of messages sent is explained by the decision not to migrate when the soonest observation can be observed by the same component, while for \algmigrr{}, the round-robin heuristic causes the \MemRep{} to always migrate.
However, this does not lead to a much lower number of simplifications (\rfig{fig:syn:simp}).
Using \THEMIS{} to compare the variants shows us that the earliest obligation heuristic reduces the size of the \MemRep{}, and thus, the size of the message, but also the number of messages sent.
However, it does not seem to impact computation as the number of simplifications remains similar.

\paragraph{Discussion.}
The observed behavior of the simulation aligns with the initial analysis described in \rsec{sec:analysis}.
We observe that the \MemRep{} presents predictable behavior in terms of size and computation.
The delay presented for each algorithm indeed depends on the listed parameters in the analysis.
With the presented bounds on \MemRep{}, we can determine and compare the algorithms that use it.
Therefore, we can theoretically estimate the situations where algorithms might be (dis)advantaged.
However, both Figures~\ref{fig:syn:computation} and \ref{fig:syn:datafig} show that for most metrics, we observe a large variance (as evidenced by the interquartile difference).
As such, we caution that while the analysis presents trends where algorithms have the advantage, it is still necessary to address the specifics, hence the need for simulation.

\paragraph{Trace variance.}
In \rtbl{tbl:exp:traces}, we examine the variance by observing metrics {\wrt} probability distributions used to generate the traces.
To exclude the variance due to the number of components, we fix $|\comps| = 6$, as it provides the highest variance.
For each metric, we present the mean and the standard deviation (between parentheses).
All metrics are normalized by the length of the run.
The metrics in order of columns are: average information delay ($\paramdelay$), average number of messages (\#Msgs), total data transferred (Data), average maximum simplifications per monitor ($\mathrm{S}$), and
convergence based on expressions evaluated ($\mathrm{Conv_E}$).
We observe that by changing the probability distribution, the metrics vary significantly.
This is particularly prominent for the information delay (especially in the case of \algchor{}), and data transferred.
We explore the differences in the algorithms in \rsec{sec:chiron} by considering real examples with existing formalized specifications.

\begin{table}[t]
  \caption{Variation of average delay, number of messages, data transfer, critical simplifications and convergence with traces generated using different probability distributions for each algorithm. Number of components is $|C| = 6$. Table cells include the mean and the standard deviation (in parentheses).}
  \label{tbl:exp:traces}
\centering
\scalebox{0.9}{\begin{tabular}{|l|c|c|cc|cc|cc|cc|}
  \hline
  \textbf{Alg.} & \textbf{Trace} & \boldmath$\paramdelay$ & \textbf{\#Msgs} & \textbf{Data} & \boldmath$\mathrm{S_{crit}}$ & \boldmath$\mathrm{Conv_E}$ \\
\hline \hline \multirow{4}{*}{{\algorch}}
 & normal & 0.69 (0.46) & 7.13 (1.38) & 94.87 (18.46) & 0.00 (0.00) & 0.83 (0.01) \\
 & binomial & 0.69 (0.46) & 7.15 (1.39) & 93.10 (18.13) & 0.00 (0.00) & 0.83 (0.01) \\
 & beta-1 & 0.70 (0.46) & 6.98 (1.47) & 96.11 (20.17) & 0.00 (0.00) & 0.83 (0.02) \\
 & beta-2 & 0.69 (0.46) & 6.91 (1.71) & 83.37 (20.66) & 0.00 (0.00) & 0.82 (0.02) \\
\hline \multirow{4}{*}{{\algmigr}}
 & normal & 1.72 (1.42) & 0.50 (0.32) & 110.13 (276.43) & 7.01 (5.01) & 0.82 (0.03) \\
 & binomial & 1.67 (1.40) & 0.49 (0.32) & 95.44 (221.65) & 6.86 (4.91) & 0.82 (0.04) \\
 & beta-1 & 1.82 (1.44) & 0.53 (0.32) & 133.15 (313.44) & 7.14 (5.16) & 0.82 (0.04) \\
 & beta-2 & 1.53 (1.36) & 0.47 (0.38) & 56.48 (114.54) & 5.95 (4.24) & 0.82 (0.03) \\
\hline \multirow{4}{*}{{\algmigrr}}
 & normal & 2.64 (1.93) & 0.70 (0.35) & 177.48 (358.61) & 7.50 (5.18) & 0.83 (0.03) \\
 & binomial & 2.59 (1.95) & 0.69 (0.36) & 171.64 (318.94) & 7.49 (5.21) & 0.83 (0.03) \\
 & beta-1 & 2.82 (1.94) & 0.74 (0.34) & 210.02 (452.83) & 7.49 (5.23) & 0.82 (0.03) \\
 & beta-2 & 2.55 (2.08) & 0.66 (0.41) & 162.28 (287.28) & 6.93 (4.90) & 0.82 (0.02) \\
\hline \multirow{4}{*}{{\algchor}}
   & normal & 2.02 (1.97) & 5.92 (1.60) & 52.54 (14.23) & 12.68 (3.63) & 0.13 (0.10) \\
   & binomial & 1.93 (1.86) & 5.95 (1.61) & 52.76 (14.33) & 12.55 (3.70) & 0.13 (0.10) \\
   & beta-1 & 2.59 (4.58) & 5.80 (1.64) & 51.54 (14.48) & 13.29 (4.33) & 0.14 (0.12) \\
   & beta-2 & 2.95 (7.26) & 5.81 (1.79) & 51.52 (15.93) & 13.50 (9.91) & 0.14 (0.14) \\
  \hline
\end{tabular}}
\end{table}



\subsection{The Chiron User Interface}
\label{sec:chiron}
\paragraph{Overview.}
Moving away from synthetic benchmarks, we consider properties that apply to patterns of programs and specifications.
In this section, we compare the algorithms by looking at a real example that uses the publish-subscribe pattern.
To that extent, we consider the Chiron user interface example~\cite{chiron}.
Chiron consists of artists responsible of rendering parts of a user interface, that register for various events via a dispatcher.
A dispatcher receives events from an abstract data type (ADT) and forwards them to the registered artists.
We chose Chiron for two practical reasons.
Firstly its example source code (in ADA), and its specifications are available online~\cite{chironweb}.
The specification is completely formalized and utilizes various LTL patterns described in~\cite{ltlpatterns,ltlpatternsweb}.
Thus, it covers a multitude of patterns for writing specifications.
Secondly, the Chiron system can be easily decomposed into various components, we consider four components, the dispatcher (A), the two artists (B,C) and the main thread (D).
The main thread is concerned with observing termination of the program.
\paragraph{Experimental setup.}
Table~\ref{tbl:chiron:specs} lists the subset of the Chiron specification we considered.
For each property,  column ID references the original property name in~\cite{chironweb},
 column $\verdict$ references the expected verdict at the end of the trace\footnote{In the case where the expected verdict is $\vna$, the specification is designed to falsify the property, as such if no falsification is found, we will terminate with verdict $\vna$.},
and column pattern identifies the LTL pattern corresponding to the formula.
We modify the Chiron example program~\cite{chironweb} to output a trace of the program, and consider the specifications listed in Table~\ref{tbl:chiron:specs}\footnote{We exclude specification 7 as we were unable to generate an automaton using \code{ltl2mon} for it. This is due to the formula either being too complex, or non-monitorable.}.
For example, we consider the specification shown in Listing~\ref{chiron:spec}.
It states that artists are only notified when the dispatcher receives an event.
That is, the dispatcher does not send events to the artists without receiving them properly from the ADT.
Since we monitor offline, we generate the trace by inserting a global monitor that contains information about all relevant atomic propositions.
The program is then instrumented to notify the monitor of events.
Specifications and traces are then provided as input to \THEMIS{} to process with the existing algorithms.
The details on the atomic propositions and their assignment to components can be found in Appendix~\ref{sec:app:chiron}.
We randomized the events dispatched in the Chiron example, and generated 100 traces of length 279.
We targeted generating traces under 300 events.
This corresponds to the ADT dispatching 91 events, with the addition of events to register, and unregister artists.

\begin{lstlisting}[basicstyle=\ttfamily\footnotesize,frame=single, caption={Example Chiron Specification},label={chiron:spec},float]
!(notify_client_event_a1_e1 || notify_client_event_a2_e1)
  U (notify_artists_e1 ||
     []!(notify_client_event_a1_e1 || notify_client_event_a2_e1))
\end{lstlisting}

\gettable{
\caption{Monitored Chiron specifications. CRC stands for Constrained Response Chain.}
\label{tbl:chiron:specs}
\centering
\begin{tabularx}{\textwidth}{|l|c|l|X|}
  \hline \textbf{ID} & \boldmath$\verdict$ & \textbf{Pattern} & \textbf{Description} \\
  \hline\hline
   1 & $\vna$ & Absence & An artist never registers for an event if she is already registered for that event, and an artist never unregisters for an event unless she is already registered for that event. \\
  \hline 2 & $\vna$ & CRC (2-1) & If an artist is registered for an event and dispatcher  receives this event, it will not receive another event before  passing this one to the artist.  \\
  \hline 3 & $\vt$ & Precedence & Dispatcher does not notify any artists of an event until it receives this event from the ADT.\\
  \hline 5 & $\vna$ & Absence & Dispatcher does not block ADT if no one is registered (this means that if no artists are registered for events of kind 1, dispatcher does nothing upon receiving an event of this kind from the ADT). \\
  \hline 7* & $\vna$ & CRC (3-1)  & The order in which artists register for events of kind 1 is
             the order in which they are notified of an event of this kind by
             the dispatcher. In other words, if artist1 registers for event2
             before artist2 does, then once dispatcher receives event2 from
             the ADT, it will first send it to artist1 and then to artist2.\\
 \hline 15a & $\vna$  & Universal & The program never terminates with an artist registered. \\
 \hline 15b & $\vt$ & Response & An artist always unregisters before the program terminates. Given that you can't register for the same event twice, we need only check that unregisters respond to registers \\
  \hline
\end{tabularx}
}

\paragraph{Comparing algorithms.}
\rfig{fig:chiron:computation} presents the means for the metrics of average delay, convergence, and both critical and maximum simplifications\footnote{We note that since we broke down the metrics per specification, we have little variation in the data, for details and standard deviations refer to Appendix~\ref{sec:app:data}.}.
We immediately notice for \algchor{} the high average delay for specifications 2 and 15b (133.86, and 116.52, respectively).
In these two cases, the heuristic to generate the monitor network for choreography has split the network inefficiently, and introduced a large delay due to dependencies.
We recall that the heuristic used for choreography consider LTL formulae.
For a given formula it counts the number of references to atomic propositions of a given component.
The monitor tasked with monitoring the formula will then be associated with the highest component.
To generate a decentralized specification, the heuristic starts with an LTL formula and splits it into two subformulas for each binary operator, then one of the subformulas is chosen to remain on the current component while the other is delegated to the component with the most references to atomic propositions.
We see in this case that simply counting references and breaking ties using the lexicographical order of the component name can yield inefficient decompositions.
Furthermore, we notice that while \algorch{} maintains the lowest delay, other algorithms can still yield comparable delays.
In the case of specification 15a we observe that \algorch{}, \algmigr{}, and \algchor{} have similar delay (1.0).
While \algmigr{} may outperform \algchor{} for specifications 2 and 15b, it is the opposite for specifications 1, 3 and 5.
This highlights that the network decomposition of monitors (i.e., the setup phase) is an important consideration when designing decentralized monitoring algorithms.

\rfig{fig:chiron:conv} presents the convergence (computed using the number of expressions evaluated).
We see that the poor decomposition also yields imbalanced workloads on the monitors.
In the case of specifications 2 and 15b, we observe a convergence of 0.67 to 0.71 for \algchor{}, respectively.
The observed convergence is comparable  with that of \algorch{} and \algmigr{}.
Furthermore, it is still possible to improve on the load balance for specifications 3 and 15a, as the convergence is high (0.33 and 0.47).

\rfig{fig:chiron:simpc} illustrates critical simplifications, we see that \algchor{} has a higher cost compared to \algmigr{} in terms of computation.
We also notice that \algmigr{} performs better than \algmigrr{} for all specifications.
The heuristic of migrating the formula based on the atomic proposition with the earliest timestamp  (earliest obligation) does indeed improve computation costs.
More importantly, we notice that the highest delay for \algchor{} is for specifications 2 and 15b.
To inspect that, we look at the maximum delay induced in a given monitor for an entire run, and consider the mean across all traces to obtain the worst-case maximum simplifications in \rfig{fig:chiron:simpm}.
Indeed, we notice a peak in the maximum number of simplification in a given round for specifications 2 and 15b.
Particularly, we notice that while comparable in other specifications (e.g., for specification 15a, we have 20  max simplifications for \algchor{} as opposed to 8.64 and 10.00 for \algmigr{} and \algmigrr{}, respectively), the maximum number of simplifications for \algchor{} increases to 2,798 (compared to 12 for \algmigrr{}), and 3,387 (compared to 16.86 for \algmigrr{}) for specifications 2 and 15b, respectively.
In this particular case, we see how delay can impact the maximum number of simplifications.
\begin{figure}[t!]
 \subfloat[][Average delay]{\includegraphics[width=.5\textwidth]{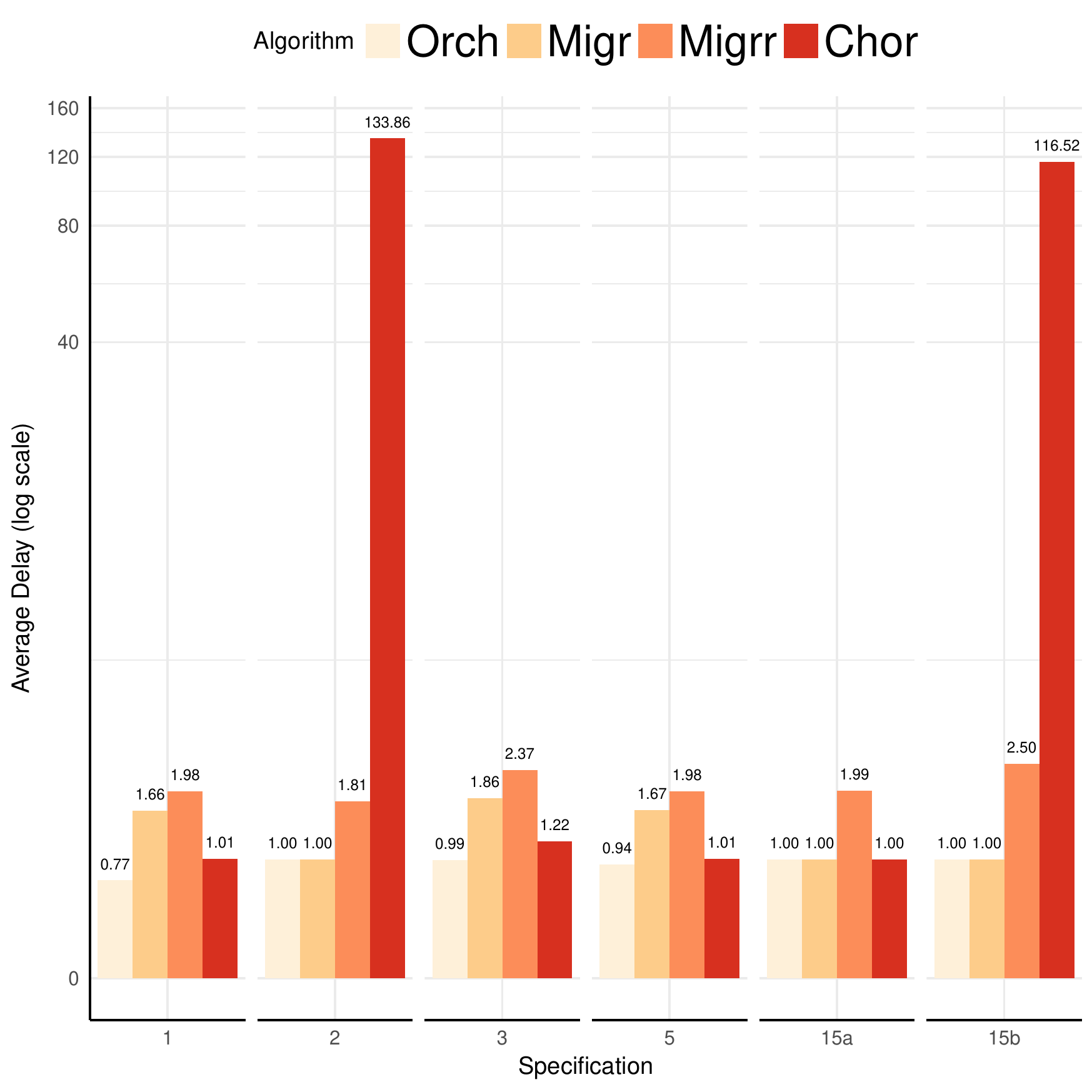} \label{fig:chiron:delay}}
 \subfloat[][Convergence]{\includegraphics[width=.5\textwidth]{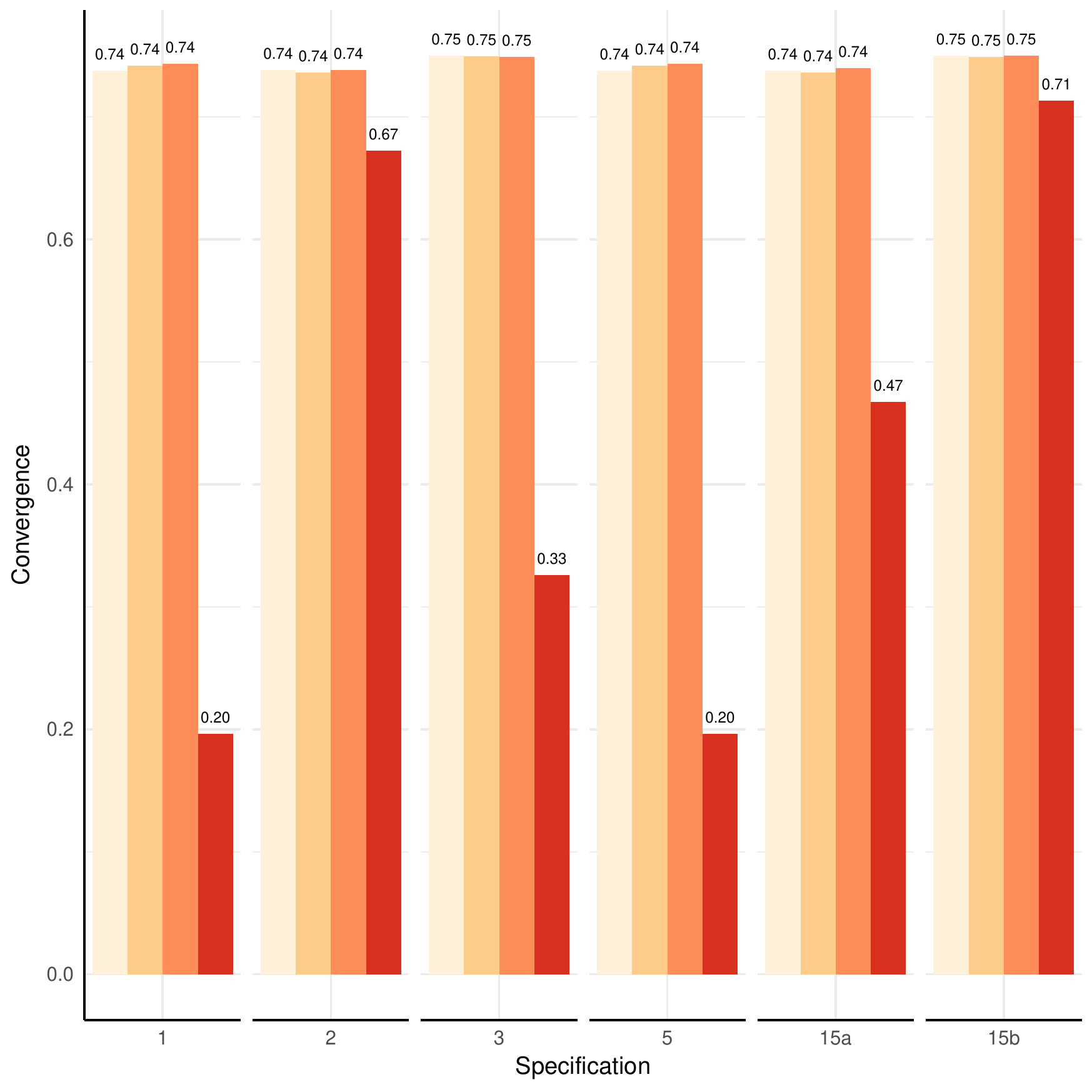} \label{fig:chiron:conv}}\\
 \subfloat[][Critical simplifications]{\includegraphics[width=.5\textwidth]{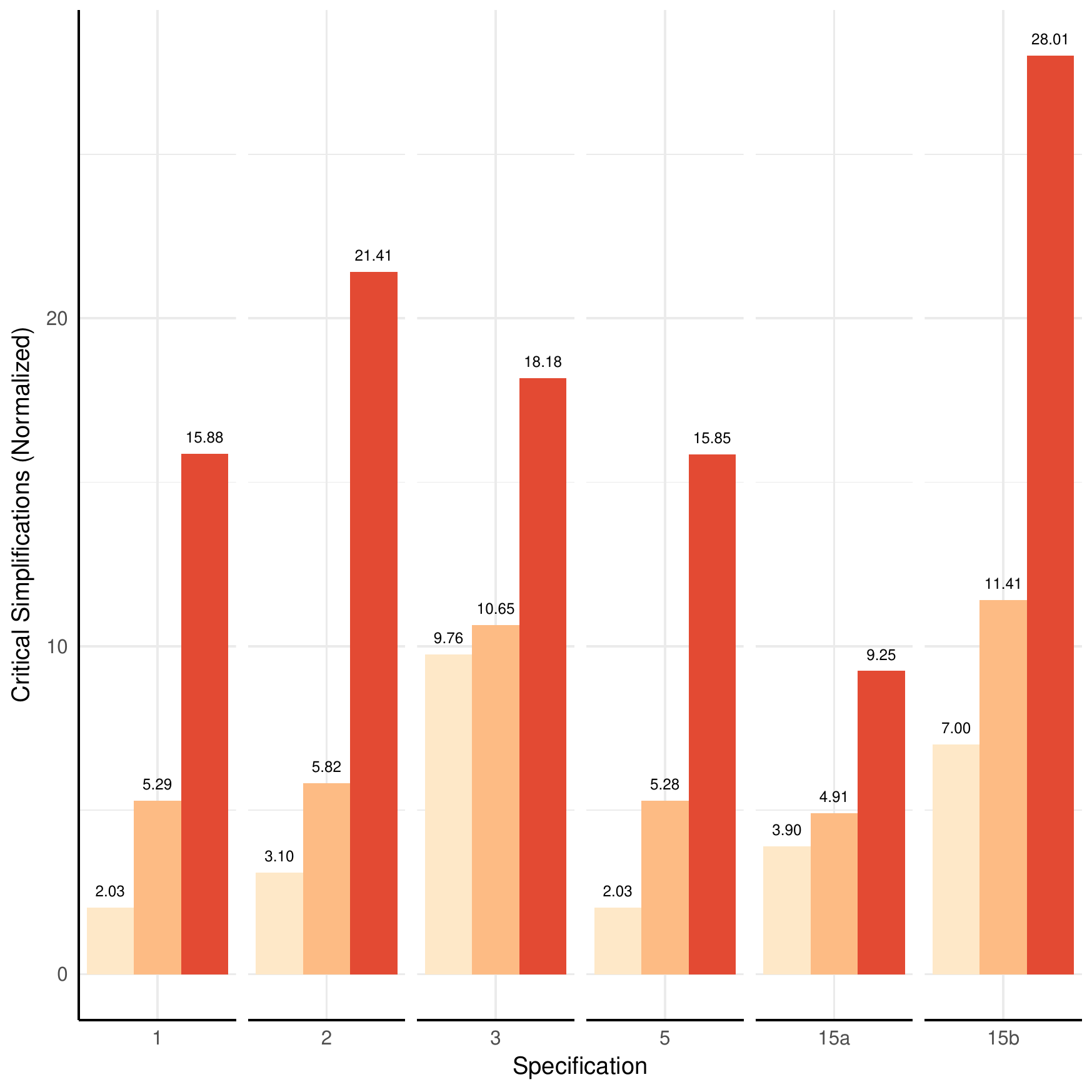} \label{fig:chiron:simpc}}
  \subfloat[][Maximum simplifications (worst case)]{\includegraphics[width=.5\textwidth]{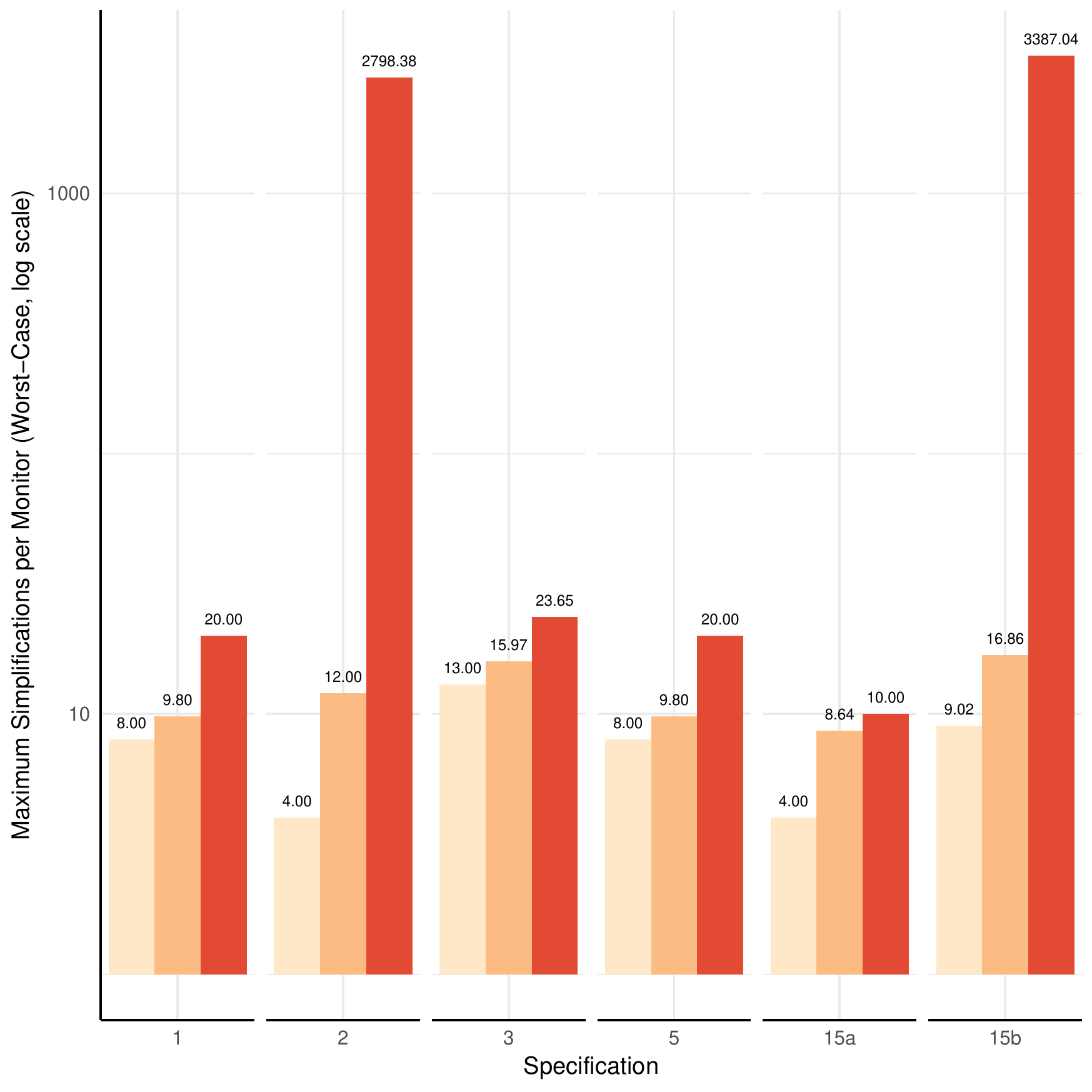} \label{fig:chiron:simpm}}
  \caption{Comparison of delay, convergence and number of simplifications. Algorithms are presented in the following order: \algorch{}, \algmigr{}, \algmigrr{}, \algchor{}. \algorch{} is omitted in the simplifications count as it is zero.}
  \label{fig:chiron:computation}
\end{figure}

We now consider communication costs by observing the number of messages transferred in \rfig{fig:chiron:msgnum}.
We see that \algmigr{} and \algmigrr{} perform well compared to the other two algorithms, with \algmigr{} performing consistently better than \algmigrr{}.
We note that the analysis of \algmigr{} indicates that the number of messages per round will be in the worst case the number of active monitors (in our case that is 1).
One can see in specification 5 that \algmigr{} sends only 0.02 messages on average per round, compared to \algmigrr{} with 1.01, followed by \algorch{} with 2.95, and finally \algchor{} with 4.89.
We see that \algorch{} outperforms \algchor{} in the case of specifications 1, 2, 3 and 5, where generally \algorch{} sends 1-2 messages less.
We note that this pattern is in line with the trends shown in \rfig{fig:syn:msgnum}.
We see for $|\comps| = 4$ that \algorch{} and \algchor{} overlap, with \algmigrr{} outperforming both, and \algmigr{} outperforming all other algorithms.
Interestingly, we find that in the case of specification 15a in  \rfig{fig:chiron:msgnum}, we see for \algchor{} a number of messages  (0.98) slightly higher than \algmigr{} (0.97) and lower than \algmigrr{} (1.01), consistent with the lower whiskers in  \rfig{fig:syn:msgnum}.
Similarly, when considering the total data transferred in \rfig{fig:chiron:msgdata}, we see as a trend across specifications \algmigr{} being particularly good, while still being slightly outperformed by \algchor{} in specifications 15a and 15b.
Furthermore, we notice that \algmigrr{} performs poorly and indeed sends more data than \algorch{} in most cases, indicating that a heuristic can indeed be instrumental in the success of designing the family of migration algorithms.
We notice also that the trends from \rfig{fig:syn:data} seem to apply in most cases, \algorch{} sends a lot more data than \algmigr{} and \algchor{}, with \algmigrr{} possibly surpassing \algorch{}.

\begin{figure}[t!]
 \subfloat[][Number of messages (normalized)]{\includegraphics[width=.5\textwidth]{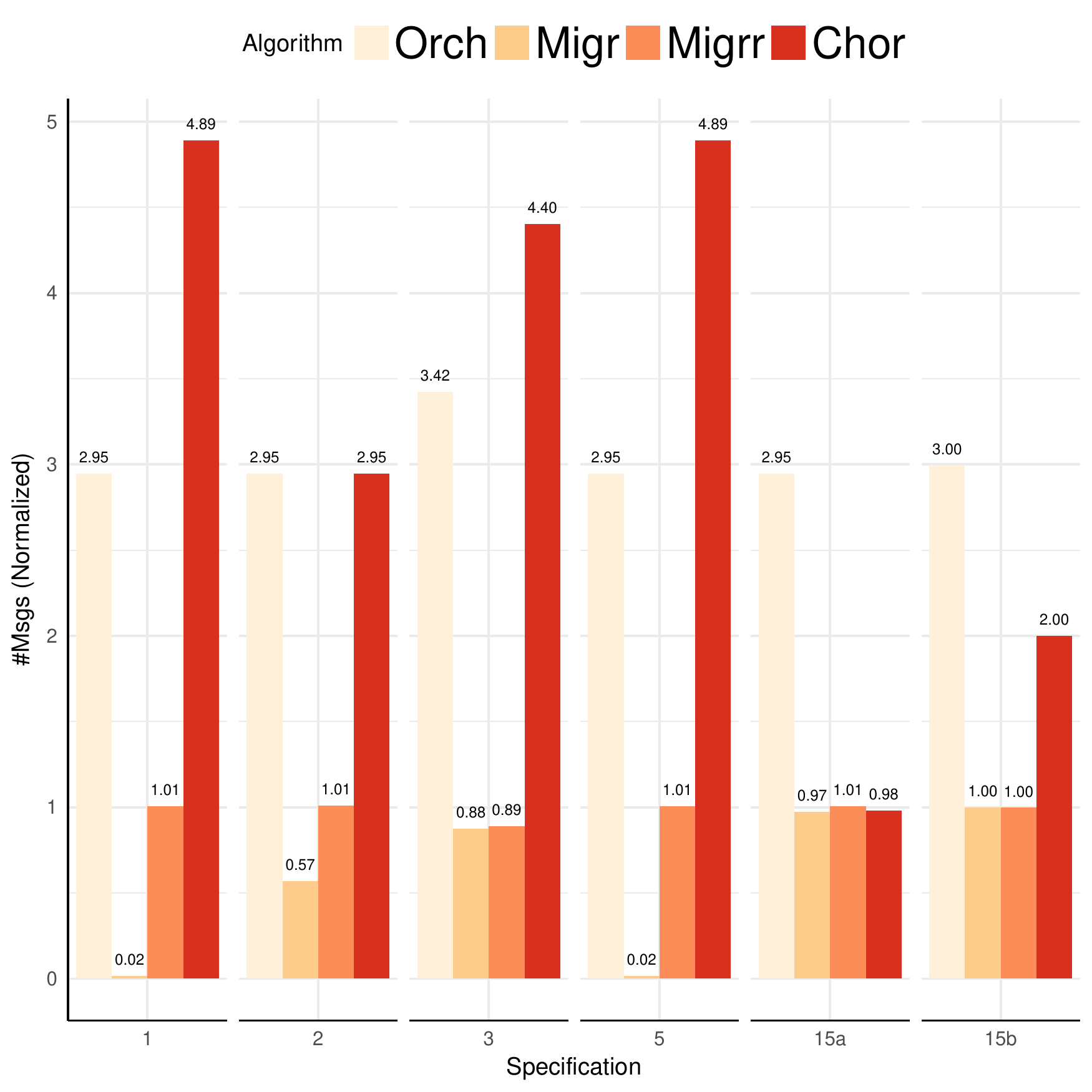} \label{fig:chiron:msgnum}}
 \subfloat[][Total data transferred (normalized)]{\includegraphics[width=.5\textwidth]{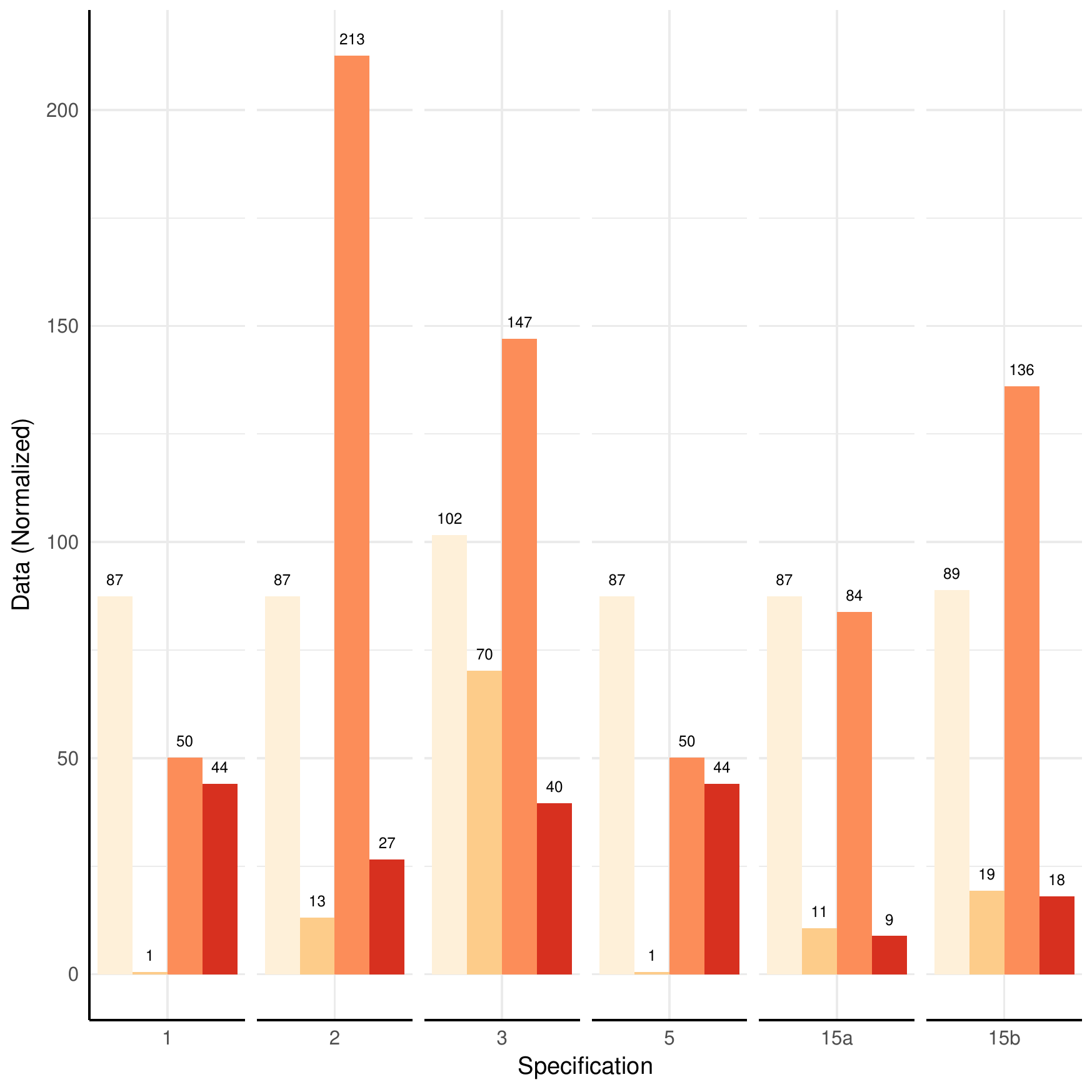} \label{fig:chiron:msgdata}}
  \caption{Data Transfer}
  \label{fig:chiron:datafig}
\end{figure}

\section{Future Directions}
\label{sec:fw}
%
By introducing decentralized specifications, we separate the monitor topology from the monitoring algorithm.
As such, we address future directions that result from analyzing the topology of monitors, and thus, define properties on such topologies, studying the monitoring by improving metrics, and applying decentralized specifications to the problem of runtime enforcement~\cite{Falcone10}.

\cblock{
On the one hand, we can consider optimizing the topology of monitors to be best suited for system architecture, this provides a formulation  of compatibility (\rsec{sec:prop:compat}) as an optimization problem.
We can also compare decentralized specifications to ensure that two specifications emit the same verdict for all possible traces, we elaborate on this property as \emph{verdict equivalence}.
Additionally, it is possible to explore the problem of \emph{specification synthesis}, that is, the problem of generating decentralized specifications using other specifications as reference, and ensuring they have specific properties, such as monitorability (\rsec{sec:dmon:monitorability}), compatibility, and verdict equivalence.
On the other hand, we can explore new metrics  in \THEMIS{}  to compare decentralized monitoring algorithms.
And finally, to explore the problem of runtime enforcement~\cite{Falcone10} when dealing with decentralized specifications.
}

\paragraph{Optimized compatibility.}
The first direction is to extend the notion of \emph{compatibility} (\rsec{sec:prop:compat}) to not only decide whether or not a specification is applicable to the architecture of the system, but also use the architecture to optimize the placement.
That is, one can generate a decentralized specification that balances computation to suit the system architecture, or optimize specific algorithms for specific layouts of decentralized systems (as discussed in \rsec{sec:analysis}).

\paragraph{Verdict equivalence.}
We can also compare decentralized specifications to ensure that two specifications emit the same verdict for all possible traces, we elaborate on this property as \emph{verdict equivalence}.
We consider two decentralized specifications $\decent$ and $\decent'$, constructed with two sets of monitors $\mons$ and $\mons'$ (as per Section~\ref{sec:dmon}).
Let the root monitors be $\mroot$, and $\mroot'$, respectively.
We recall the notation from \rsec{sec:cmon:pre}, for a given monitor label $\ell$, $q_{\ell_0}$, $\Delta_\ell$ and $\verdictf_\ell$ indicate the initial state, transition relation and the verdict function for a given monitor (automaton).
One way to assess equivalence, is to verify, that for all traces, both specifications yield similar verdicts.
It suffices to evaluate the trace on the transition function starting from the root monitor, and check the verdict of the reached state.
That is, two decentralized specifications  $\decent$ and $\decent'$ are \emph{verdict equivalent}
iff $\forall t \in \traces: \>
    \verdictf_{\mroot}(\Delta_{\mroot}'^*(q_{\mroot_0} , t, 1)) =
    \verdictf_{\mroot'}(\Delta_{\mroot'}'^*(q_{{\mroot'}_0} , t, 1)).
 $
The verdict equivalence property establishes the basis for comparing two specifications that eventually output the same verdicts for the same traces.
For all possible traces ($\forall t \in \traces$),
we first evaluate the trace on the root monitor of $\decent$ (i.e., $\Delta_{\mroot}'^*(q_{\mroot_0} , t)$),
and similarly we evaluate the same trace on the root monitor of $\decent'$ (i.e., $\Delta_{\mroot'}'^*(q_{{\mroot'}_0} , t)$).
The states reached for both of the automata executions need to be labeled by the same verdict.
While both specifications yield the same verdict for a given trace, one could also extend this formulation to add bounds on delay .
\paragraph{Specification Synthesis.}
\label{sec:prop:synthesis}
Another interesting problem to explore is that of \emph{specification synthesis}.
Specification synthesis considers the problem of generating a decentralized specification, using various inputs.
Typically, we would expect another specification as reference, and possibly the system architecture.
For example, given a centralized specification, we generate a decentralized specification, by splitting the specification into subspecifications and assigning the subspecifications to monitors.
Generating a decentralized specification using a centralized one as reference is used in some algorithms such as choreography~\cite{DecentMon}\footnote{For more details see \rapp{sec:app:choreo}}.
Starting from an LTL formula, the formula is split into subformulas hosted on the various components of the system (this is detailed further in Section~\ref{sec:analysis:algs}).
Given a decentralized specification $\decent$, and a system graph $\tuple{\comps,E'}$, the problem consists in generating a specification $\decent'$.
The variants of the synthesis problem depend on the properties that $\decent'$ must have, we list (non-exhaustively) example properties:
\begin{enumerate}
    \item $\decent'$ is monitorable (Section~\ref{sec:dmon:monitorability});
    \item $\decent'$ is compatible with $\tuple{\comps, E'}$ (Section~\ref{sec:prop:compat});
        \item $\decent'$ and $\decent$ are verdict equivalent.
\end{enumerate}
Synthesis problems could also be expanded to handle optimization techniques, with regards to specifications.
The specification determines the computation and communication needed by the monitors.
As such, it is possible to optimize, the size of automata, and references so as to fine tailor load and overhead for a given system architecture.

\paragraph{Extending \THEMIS{} metrics.}
Moreover, one could consider creating new metrics for \THEMIS{} to analyze more aspects of decentralized monitoring algorithms.
We see that this is important, as in two specifications out of the five when using Chiron  traces (\rsec{sec:chiron}), the choreography algorithm using a simple heuristic generated an inefficient decentralized specification.
New metrics would be automatically instrumented on all existing algorithms and experiments could be easily replicated to compare them.
%
%

\section{Conclusions}
\label{sec:conclusion}
%
%
%

We present a general approach to monitoring decentralized specifications.
A specification is a set of automata associated with monitors that are attached to various components.
We provide a general decentralized monitoring algorithm defining the major steps needed to monitor such specifications.
We make a clear distinction between the topology of monitors and the behavior of each monitor.
We elaborate on twp properties associated with decentralized specifications: compatibility, and monitorability.
In addition, we present the \MemRep{} data structure which allows us to (i) aggregate monitor states with strong eventual consistency, (ii) remain sound {\wrt} the execution of the monitor, and (iii) characterize the behavior of the algorithm at runtime.
We then map three existing algorithms: Orchestration, Migration and Choreography to our approach using our data structures.
We develop and use \THEMIS{} to implement algorithms and analyze their behavior by designing new metrics.
We implement four algorithms in \THEMIS{} under our model and data-structures: orchestration (\algorch{}), migration using earliest obligation (\algmigr{}), migration using round-robin (\algmigrr{}), and choreography (\algchor{}).
Using \THEMIS{} and the designed metrics, we explore simulations of the four algorithms on two scenarios and validate the trends observed in the analysis.
In the first scenario, we use the synthetic benchmark comprising of random specifications and traces.
In the second scenario, we use a real example (Chiron) with existing formalized specifications.

\bibliographystyle{splncs03}
\bibliography{biblio}

\appendix
\section{Proof} \label{sec:app:proofs}
%
%

\begin{proof}[Proof of \rprop{prop:ehe:invariant}]\label{proof:ehe:invariant}
The proof is by induction on the number of timestamps in the \MemRep{}, i.e., $n = |\timestamps{\sys}|$.
Without loss of generality, we can assume the automaton being encoded is normalized (see Remark~\ref{rmk:normalized-aut}), that is, all shared edges between any two states are replaced by one edge which is labeled by the disjunction of their labels.\\

\noindent
One could see that the base case only contains the initial state of an automaton, i.e., $\sys^0 = [0 \mapsto q_0 \mapsto \vt]$, and as such the proposition holds.\\

\noindent
Let us consider $n = 2$, we have $\sys^1 =  \smove(\sys^0, 0, 1)$.
To compute $\smove$, we first consider $\snext(\sys^0, 0)$ which considers all states reachable from $q_0$ as the only tuple in $\sys^0$ is $\tuple{0, q_0}$, i.e.,
$\snext(\sys^0, 0) = \setof{q' \in Q \mid \exists e \in \expr: \delta(q_0, e) = q'}$, we know that only one such $e$ can evaluate to $\vt$ for any memory encoded with the identity encoder ($\fid$), since the automaton is deterministic.
Let us collect all such states and their expressions as $P = \setof{\tuple{q', e} \in Q \times \expr \mid \exists e \in \expr: \delta(q_0, e) = q'}$.
We note that $\sys^1(0,q_0) = \vt$ is the only entry for timestamp 0.
The property holds trivially for that entry.
We now consider the entries in $\sys^1$ for timestamp 1.
Each of tuple $\tuple{q', e} \in P$ corresponds to the expression $\sys^1(1, q')$, constructed with $\sto(\sys^0, 0, q', \fts_1) = \sys^0(0, q_0) \land \fts_1(e)$.
We note that $\fts_1$ only adds the timestamp 1 to each atomic proposition.
As such, for any given memory encoded with $\fts_1$ only one such expression can be evaluated to $\vt$.\\

\noindent
Inductive step:
We assume that the property holds on $\sys^{n-1}$ for some $n \in \mathbb{N}$, that is:\\
$\forall \mem \in \Mem, \forall t \in \timestamps{\sys^{n-1}}, \exists q \in Q: (\seval(\sys^{n-1}(t,q), \mem) = \vt) \implies (\forall q' \in Q \setminus \{ q \} \implies \seval(\sys^{n-1}(t,q'), \mem) \neq \vt)$.
Let us prove that the property holds for $\sys^n$.\\

\noindent
The approach is similar to that of $n=2$ using the recursive structure of the $\MemRep{}$ to generalize. We decompose $\sys^n$ as follows:

\begin{align*}
  \sys^n = \sys^{n-1} \sysadd  \smashoperator{\biguplus\limits^{\lor}_{q' \in \snext(\sys^{n-1},n)}} \setof{n \mapsto q' \mapsto \sto(\sys^{n-1},n-1,q', \fts_{n})} && \text{(definition of $\smove$)} \\
\end{align*}

\noindent
We know that $\timestamps{\sys^n} = \timestamps{\sys^{n-1}} \cup \setof{n}$.
The induction hypothesis states that the property holds for all entries in $\sys^{n-1}$ (i.e. for $t \in \timestamps{\sys^{n-1}}$), we consider the entries for timestamp $n$ only.
Since $\biguplus\limits^{\lor}$ applies $\sysadd$ on the entire set, and it is associative and commutative we consider the expression for a given state after all the merges, without consideration of order of merges.
As such the states associated with timestamp $n$ are computed using $\snext(\sys^{n-1},n)$. We have:

\begin{align*}
  \forall q' \in \snext(\sys^{n-1},n): \sys^n(n, q')  &=  \sto(\sys^{n-1},n-1,q', \fts_{n}) && \text{(definition of $\smove$)} \\
                 &=  	\smashoperator{\bigvee_{\setof{\tuple{q,e'} \mid  \> \delta(q, e') = q'}}}( \sys^{n-1}(n - 1, q) \land \fts_n(e')  )  && \text{(1)}
\end{align*}

\noindent
(1) follows from the definition of $\sto$.
If we examine the disjunction we notice using the induction hypothesis that there can only be a unique $q_\mrm{u} \in Q$ with $\sys^{n-1}(n - 1, q_\mrm{u})$ that evaluates to $\vt$ at timestamp $n-1$.
As such, the conjunction can only hold for one such $q_\mrm{u}$.
Consequently, we can rewrite (1) by simplifying the disjunction and considering only states reachable from $q_\mrm{u}$, as the rest cannot evaluate to $\vt$.
Let us collect all such states and expressions in the set
$P_\mrm{u} = \setof{\tuple{q', e'} \mid q' \in \snext(\sys^{n-1},n) \land  \exists e' \in \expr_\AP: \delta(q_\mrm{u}, e') = q'}$.
The only entries that can still evaluate to $\vt$ are:
\begin{align*}
    \forall \tuple{q', e'} \in P_\mrm{u}:  \sys^n(n, q') &=  \sys^{n-1}(n - 1, q_\mrm{u}) \land \fts_n(e')\\
      &= \fts_n(e')
\end{align*}
Since the automaton is deterministic, we know that we have one unique expression $e_\mrm{u}$ that can evaluate to $\vt$, given any memory encoded with $\fid$.
Since $\fts_n$ only adds the timestamp $n$ to the atomic propositions without changing the expression, we deduce that only $\fts_n(e_\mrm{u})$ evaluates to $\vt$.
As such, there is a unique expression that can evaluate to $\vt$ for any given memory encoded with $\fts_n$.
Furthermore, we know that the expression has only been encoded with $\fts_n$ so when memories encoded with different timestamps or encoders are merged, they do not affect the evaluation of $\fts_n(e_\mrm{u})$.
As such, we have a unique entry $\sys^n(n, q'_\mrm{u})$ s.t. $\delta(q_\mrm{u}, e_\mrm{u})$ that can evaluate to $\vt$. Therefore:\\

\noindent
\begin{align*}
\forall \mem \in \Mem, \forall t \in &\timestamps{\sys^{n}}, \exists q \in Q: \\
 &(\seval(\sys^n(t,q),\mem) = \vt) \implies (\forall q' \in Q : q' \neq q \implies \seval(\sys^n(t,q'), \mem) \neq \vt)
\end{align*}
\qed
\end{proof}

\begin{lemma}[Evaluation modulo encoding]\label{sec:app:proofs:l1}
  Given a trace $\trace$ of length $i$ and a reconstructed global trace $\rho(\trace) = \mathit{evt}_1 \cdot \hdots \cdot \mathit{evt}_i$, we consider two memories $\mem^i_\aut$ and $\mem^i$ generated under different encodings.
  We consider $\mem^i_\aut = \cons(\mathit{evt}_i, \fid)$, and  $\mem^i = \biguplus^{2}_{t \in [1, i]} \setof{\cons(\mathit{evt}_{t}, \fts_{t})}$.
  We show that an expression encoded using different encodings evaluates the same for the memories, that is:
  \[
  \forall e \in \expr_\AP:
    \seval(\fid(e), \mem_\aut^{i}) \Leftrightarrow \seval(\fts_{i}(e), \mem^{i}).
  \]
\end{lemma}
\begin{proof}[Proof of Lemma~\ref{sec:app:proofs:l1}]

  \noindent
  \noindent
  We first note that for the first evaluation $\seval(\fid(e), \mem_\aut^{i})$, we rely only on the event $\mathit{evt}_i$ since  $\mem^i_\aut = \cons(\mathit{evt}_i, \fid)$.
  This is not the case for $\seval(\fts_{i}(e), \mem^{i})$ as $\mem^i = \biguplus^{2}_{t \in [1, i]} \setof{\cons(\mathit{evt}_{t}, \fts_{t})}$.
  However, we notice that for the second evaluation we evaluate the expression $\fts_{i}(e)$, that is, where the expression where all atomic propositions have been encoded by the timestamp $i$.
  Therefore, let us denote the memory with the timestamp $i$ by $\mem' = \cons(\mathit{evt}_{i}, \fts_{i})$.
  We can rewrite $\mem^i$ as follows:
  \[\begin{array}{rll}
      \mem^i &= \cons(\mathit{evt}_{i}, \fts_{i}) &\memadd \biguplus^{2}_{t \in [1, k]} \setof{\cons(\mathit{evt}_{k}, \fts_{k})} \\
             &= \mem' &\memadd \biguplus^{2}_{t \in [1, k]} \setof{\cons(\mathit{evt}_{k}, \fts_{k})}.
  \end{array}\]
  We know that all entries $\tuple{k, a} \in \fdom(\mem^i)$ with $k < i$ do not affect at all the evaluation of an expression encoded with $\fts_i$.
  As such we have: \[
    \forall e \in \expr_\AP: \seval(\fts_{i}(e), \mem^{i}) \equiv  \seval(\fts_{i}(e), \mem')
   \]
  We now show that the two memories $\mem^i_\aut$ and $\mem'$ contain simply an encoding of the same atomic propositions.
  We have by construction the following:
  \[
  \begin{array}{rll}
    \forall a \in \fdom(\mem_\aut^{i}) : & \tuple{i, a} \in \fdom(\mem^{i})   &\land \> \mem_\aut^{i}(a) = \mem'(\tuple{i, a})\\
    \forall \tuple{i, a'} \in \fdom(\mem') : & a' \in \fdom(\mem_\aut^{i}) & \land  \> \mem'(\tuple{i, a'}) = \mem_\aut^{i}(a')
  \end{array}
  \]
  As such we have:  $\forall e \in \expr_\AP: \seval(\fid(e), \mem_\aut^{i})
  \equiv \seval(\fts_{i}(e), \mem')
  \equiv \seval(\fts_{i}(e), \mem^{i})
  $.

\end{proof}

\begin{proof}[Proof of \rprop{prop:cmon:soundness}]\label{proof:ehe:soundess}
Given a trace $\trace$ of length $i$ and a reconstructed global trace $\rho(\trace) = \mathit{evt}_1 \cdot \hdots \cdot \mathit{evt}_i$, the proof is done by induction on the length of the trace $|\rho(\trace)|$.
We omit the label $\ell$ for clarity.\\

\noindent
Base case: $|\rho(\trace)| = 0, \rho(\trace) = \emptytrace$

\noindent
\[ \Delta^*(q_0, \emptytrace) =   q_0  = \sreach(\sys^0, [\,], 0)\]
\[ \sys^0 = \smove([0 \mapsto q_0 \mapsto \vt], 0, 0) = [0 \mapsto q_0 \mapsto \vt] \]
We only have expression $\vt$ which is mapped to $q_0$ at $t = 0$.
Expression $\vt$ requires no memory to be evaluated.\\

\noindent
Inductive step:
We assume that the property holds for a trace of length $i$ for some $ i \in \mathbb{N}$, that is $\Delta^*(q_0, \mathit{evt}_1 \cdot \hdots \cdot \mathit{evt}_i) = \sreach(\sys^i, \mem^i, i) = q_i$.
Let us prove that the property holds for any trace of length $i+1$.\\

\noindent
We now consider the transition functions in the automaton:
\begin{align*}
 q_{i+1} &= \Delta^*(q_0, \mathit{evt}_1 \cdot \hdots \cdot \mathit{evt}_{i+1})\\
         &= \Delta(\Delta^*(q_0, \mathit{evt}_1 \cdot \hdots \cdot \mathit{evt}_i), \mathit{evt}_{i+1})  \text{ (\rdef{def:cmon:aut-semantics})}\\
         &= \Delta(q_i, \mathit{evt}_{i+1})                                        \text{ (Induction Hypothesis)}\\
         &\Leftrightarrow  \exists \vars{e} \in \expr_\AP : \delta(q_i, expr) = q_{i+1} \land \seval(e, \mem_\aut^{i+1}) = \vt\text{ (1)}\\
\end{align*}
We note that, since the automaton is deterministic, there is a unique $q_{i+1}$ such that $q_{i+1} = \Delta(q_i, \mathit{evt}_{i+1})$.

We now consider the \MemRep{} operations to reach $q_{i+1}$ from $q_i$.
\begin{align*}
  q_i &= \sreach(\sys^i, \mem^i, i) \\
      &\Leftrightarrow e = \sys^i(i,q_i) \mbox{\text{ with }} \seval(e, \mem^i) = \vt && \text{(2)}\\
      & \quad \land \forall q_i' \in Q: q_i' \neq q_i \implies  \seval(\sys^i(i, q'_i)) \neq \vt && \text{(\rprop{prop:ehe:invariant})}\\
      &\Leftrightarrow \sto(\sys^i, i, q_{i+1}, \fts_{i+1}) = \vt && \text{(3)}\\
\end{align*}

\noindent
(3) From the induction hypothesis,  we know that $\sys^i(i, q_i) = \vt$, thus:
\begin{align*}
  \sto(\sys^i, &i, q_{i+1}, \fts_{i+1}) \\
                                         &= \bigvee\limits_{\setof{\tuple{q, e'} \mid \delta(q, e')= q_{i+1}'}}(\sys^i(i, q) \land \fts_{i+1}(e') ) \\
                                       &= \bigvee\limits_{\setof{\tuple{q, e''} \mid \delta(q, e'')= q_{i+1}' \land q \neq q_i}}(\sys^i(i, q) \land \fts_{i+1}(e'') ) \\
                                       &\quad  \lor \bigvee\limits_{\setof{\tuple{q_i,e'''} \mid \delta(q_i, e''') = q_{i+1}'}} (\fts_{i+1}(e''')).\\
\end{align*}
We split the disjunction to consider the expressions that only come from state $q_i$, we now show that one such expression evaluates to $\vt$.
We know from (1), that one such expression can be taken in the automaton:
\begin{align*}
   \exists \mathit{e} \in \expr_\AP &: \delta(q_i, \mathit{e}) = q_{i+1} \land \seval(\mathit{e}, \mem_\aut^{i+1}) = \vt && \text{(1)}\\
     &\Leftrightarrow  \seval(\fts_{i+1}(\mathit{e}), \mem^{i+1}) = \vt && \text{(4)}\\
     &\Leftrightarrow \sto(\sys^i, i, q_{i+1}, \fts_{i+1}) = \vt        && \text{(5)}
\end{align*}
(4) is obtained using \rlemma{sec:app:proofs:l1} and  $\fid(\mathit{e}) = \mathit{e}$.\\
(5) follows from the disjunction.

\noindent
Using the same approach, we can show that  $\forall q' \in \snext(\sys^i, i): q' \neq q_{i+1} \implies \sto(\sys^i, i, q', \fts_{i+1}) \neq \vt$, since the first part of the conjunction does not evaluate to $\vt$, and we know that the second part cannot evaluate to $\vt$ by (2).\\
Finally, $\sto(\sys^i, i, q_{i+1}, \fts_{i+1}) = \vt$ iff $\sreach(\sys^{i+1}, \mem^{i+1}, i+1) = q_i.$\\
\end{proof}
%
\section{Chiron System Atomic Propositions}
\label{sec:app:chiron}
%
We broke down the Chiron system based on analysis of the examples provided in~\cite{chironweb,ltlpatterns}, using the various specifications rewritten in~\cite{ltlpatternswebspecs}.
Table~\ref{tbl:app:chiron} displays various associations we used to generate our traces and events.
Column \textbf{C} assigns an ID to the component.
Column \textbf{Name} lists the logical module of the system we considered as a component.
Column \textbf{Original (AP)} lists the atomic proposition provided by the authors of Chiron, and then edited by~\cite{ltlpatterns}.
Column \textbf{AP} maps the atomic proposition to our traces.
Column \textbf{Comments} includes comments on the atomic propositions.

{\begin{table}[h!]\scriptsize
\caption{Chiron Atomic Propositions and Components}
\label{tbl:app:chiron}
\begin{tabularx}{\textwidth}{|c|c|c|c|X|}
\hline \textbf{C} & \textbf{Name} & \textbf{Original (AP)} & \textbf{AP} & \textbf{Comments} \\
\hline\multirow{12}{*}{ A } & \multirow{12}{*}{ Dispatcher } & registered\_event\_a1\_e1 & a0 & \multirow{4}{*}{ Holds true after an artist has completed registration} \\
& & registered\_event\_a1\_e2 & a1 & \\
& & registered\_event\_a2\_e1 & a2 & \\
& & registered\_event\_a2\_e2 & a3 & \\
\cline{3-5}& & notify\_a1\_e1 & a4 & \multirow{4}{*}{ Holds true on starting to dispatch an event to an artist} \\
& & notify\_a1\_e2 & a5 & \\
& & notify\_a2\_e1 & a6 & \\
& & notify\_a2\_e2 & a7 & \\
\cline{3-5}& & lst\_sz0\_e1 & a8 & \multirow{4}{*}{ Tracking the size of the list (state) } \\
& & lst\_sz0\_e2 & a9 & \\
& & lst\_gt2\_e1 & a10 & \\
& & lst\_gt2\_e2 & a11 & \\
\hline\multirow{6}{*}{ B } & \multirow{6}{*}{ Artist1 } & notify\_client\_a1\_e1 & b0 & \multirow{2}{*}{ Artist receives a notification } \\
& & notify\_client\_a1\_e2 & b1 & \\
\cline{3-5}& & register\_event\_a1\_e1 & b2 & \multirow{2}{*}{ Artist requests to register for an event } \\
& & register\_event\_a1\_e2 & b3 & \\
\cline{3-5}& & unregister\_event\_a1\_e1 & b4 & \multirow{2}{*}{ Artist requests to unregister } \\
& & unregister\_event\_a1\_e2 & b5 & \\
\hline\multirow{6}{*}{ C } & \multirow{6}{*}{ Artist2 } & notify\_client\_a2\_e1 & c0 & \multirow{6}{*}{ See Artist1 } \\
& & notify\_client\_a2\_e2 & c1 & \\
& & register\_event\_a2\_e1 & c2 & \\
& & register\_event\_a2\_e2 & c3 & \\
& & unregister\_event\_a2\_e1 & c4 & \\
& & unregister\_event\_a2\_e2 & c5 & \\
\hline D & Main & term & d0 & Main program terminates \\
\hline
\end{tabularx}

\end{table}
}

\section{Changes in \THEMIS{}}\label{sec:app:newthemis}
%
Figure~\ref{fig:tool:ehe} shows the data transferred for the migration algorithms, which is associated with the size of the $\MemRep{}$.
We reduced the size of the $\MemRep{}$ by making calls to the simplifier only for complex simplifications, and implement the basic Boolean simplification while traversing the expression to replace atomic propositions by looking up the memory (in the operation $\rw$).
Since we have less calls, we apply a more aggressive simplification that is more costly\footnote{We use \texttt{ltlfilt} from~\cite{SPOT} with \texttt{--boolean-to-isop} to rewrite Boolean subformulas as irredundant sum of products.}, but also reduces the size of the expressions.
The x-axis indicates the algorithm's variant and the number of components, where {\algmigr} (resp. {\algmigr}r) stands for earliest the variant obligation (resp. round-robin).
The y-axis is presented in logarithmic scale.
We notice a significant drop in the size of \MemRep{}, dropping from 769 in the ISSTA'17 version for Migrr-5 to 73.12.

\begin{figure}[h!]
  \centering
  \includegraphics[width=.85\textwidth]{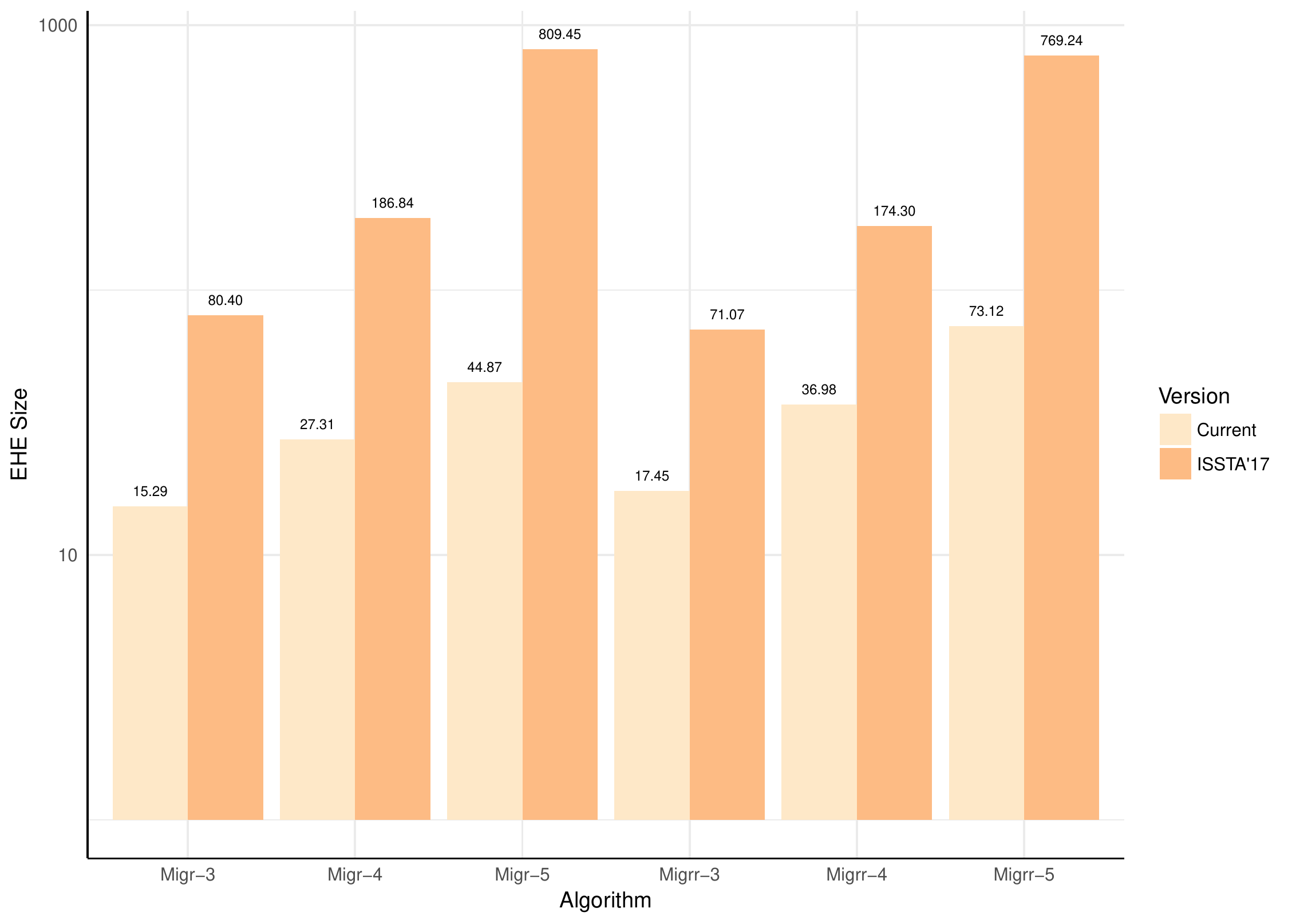}
  \caption{Size of \MemRep{}}
  \label{fig:tool:ehe}
\end{figure}

%
%
%
%
%
%
%

\section{Detailed Comparison}
\label{sec:app:data}
Tables~\ref{tbl:exp:syndata} and \ref{tbl:exp:chirondata} present the detailed comparison for the synthetic scenario and Chiron, respectively.
The metrics presented are (in order of columns): average information delay ($\paramdelay$), normalized average number of messages (\#Msgs), normalized data transferred (Data), maximum simplifications done by any given monitor per run, averaged across all runs ($\mathrm{S_{max}}$), normalized critical simplifications ($\mathrm{S_{crit}}$), and convergence based on expressions evaluated ($\mathrm{Conv_E}$). For more details on the metrics, see \rsec{sec:experiment}.

\begin{table}[h!]
  \caption{Synthetic Benchmark. Cells contains mean and standard deviation in parentheses.}
  \label{tbl:exp:syndata}
\centering
\scalebox{0.85}{\begin{tabular}{|l|c|c|cc|cc|c|}
  \hline
  \textbf{Alg.} & \textbf{$|\comps|$} & \boldmath${\paramdelay}$ & \textbf{\#Msgs} & \textbf{Data}  & \boldmath$\mathrm{S_{max}}$ & \boldmath$\mathrm{S_{crit}}$ & \boldmath$\mathrm{Conv_E}$ \\

\hline\hline \multirow{4}{*}{{\algorch}}
 &   3 & 0.48 (0.50) & 2.44 (0.61) & 31.84 (8.07) & 0.00 (0.00) & 0.00 (0.00) & 0.65 (0.02) \\
 &   4 & 0.53 (0.50) & 3.85 (0.94) & 50.05 (12.39) & 0.00 (0.00) & 0.00 (0.00) & 0.74 (0.02) \\
 &   5 & 0.64 (0.48) & 5.30 (1.16) & 69.20 (15.55) & 0.00 (0.00) & 0.00 (0.00) & 0.79 (0.02) \\
 &   6 & 0.69 (0.46) & 7.04 (1.50) & 91.86 (20.02) & 0.00 (0.00) & 0.00 (0.00) & 0.83 (0.02) \\

\hline \multirow{4}{*}{{\algmigr}}
 &   3 & 0.58 (0.58) & 0.27 (0.32) & 8.46 (15.32) & 4.72 (4.41) & 3.08 (2.66) & 0.65 (0.02) \\
 &   4 & 0.71 (0.67) & 0.32 (0.34) & 17.45 (35.87) & 6.10 (6.17) & 4.03 (3.75) & 0.73 (0.03) \\
 &   5 & 0.96 (0.71) & 0.43 (0.34) & 30.41 (56.68) & 7.41 (6.18) & 4.97 (3.76) & 0.79 (0.03) \\
 &   6 & 1.19 (0.86) & 0.50 (0.34) & 98.80 (244.94) & 10.09 (8.32) & 6.74 (4.87) & 0.82 (0.04) \\
\hline \multirow{4}{*}{{\algmigrr}}
 &   3 & 0.76 (0.69) & 0.78 (0.33) & 14.51 (18.40) & 5.62 (4.99) & 3.51 (2.93) & 0.65 (0.02) \\
 &   4 & 1.02 (0.90) & 0.76 (0.36) & 31.76 (51.55) & 7.64 (7.16) & 4.58 (4.04) & 0.74 (0.03) \\
 &   5 & 1.39 (1.04) & 0.75 (0.35) & 62.83 (91.89) & 9.70 (7.88) & 5.70 (4.25) & 0.79 (0.03) \\
 &   6 & 1.72 (1.19) & 0.70 (0.37) & 180.35 (360.25) & 12.56 (9.76) & 7.35 (5.14) & 0.82 (0.03) \\

\hline \multirow{4}{*}{{\algchor}}
 &   3 & 1.47 (1.99) & 2.79 (1.10) & 24.98 (9.85) & 60.22 (242.88) & 12.27 (6.55) & 0.16 (0.12) \\
 &   4 & 1.36 (1.52) & 3.84 (1.23) & 34.36 (10.94) & 44.71 (184.05) & 12.95 (5.98) & 0.13 (0.12) \\
 &   5 & 1.41 (1.55) & 4.63 (1.37) & 41.17 (12.16) & 44.06 (223.15) & 12.68 (6.06) & 0.12 (0.11) \\
 &   6 & 1.29 (1.38) & 5.87 (1.66) & 52.09 (14.77) & 38.35 (215.27) & 13.01 (6.01) & 0.13 (0.12) \\

  \hline
\end{tabular}}
\end{table}

\begin{table}[h!]
  \caption{Metrics for Chiron traces. Cells contains mean and standard deviation in parentheses.}
  \label{tbl:exp:chirondata}
\centering
\scalebox{0.85}{\begin{tabular}{|l|c|c|cc|cc|c|}
  \hline
  \textbf{Alg.} & \textbf{Spec} & \boldmath$\paramdelay$ & \textbf{\#Msgs} & \textbf{Data} & \boldmath$\mathrm{S_{max}}$ & \boldmath$\mathrm{S_{crit}}$ & \boldmath$\mathrm{Conv_E}$ \\
\hline \hline\multirow{6}{*}{{\algorch}}
 & 1 & 0.77 (0.42) & 2.95 (0.00) & 87.46 (0.00) & 0.00 (0.00) & 0.00 (0.00) & 0.74 (0.00) \\
   & 2 & 1.00 (0.00) & 2.95 (0.00) & 87.46 (0.00) & 0.00 (0.00) & 0.00 (0.00) & 0.74 (0.00) \\
   & 3 & 0.99 (0.10) & 3.42 (0.02) & 101.62 (1.00) & 0.00 (0.00) & 0.00 (0.00) & 0.75 (0.00) \\
   & 5 & 0.94 (0.24) & 2.95 (0.00) & 87.46 (0.00) & 0.00 (0.00) & 0.00 (0.00) & 0.74 (0.00) \\
   & 15a & 1.00 (0.00) & 2.95 (0.00) & 87.46 (0.00) & 0.00 (0.00) & 0.00 (0.00) & 0.74 (0.00) \\
   & 15b & 1.00 (0.00) & 3.00 (0.01) & 88.90 (0.16) & 0.00 (0.00) & 0.00 (0.00) & 0.75 (0.00) \\

\hline \multirow{6}{*}{{\algmigr}}
 & 1 & 1.66 (0.03) & 0.02 (0.00) & 0.52 (0.00) & 8.00 (0.00) & 2.03 (0.00) & 0.74 (0.00) \\
 & 2 & 1.00 (0.00) & 0.57 (0.00) & 13.09 (0.10) & 4.00 (0.00) & 3.10 (0.01) & 0.74 (0.00) \\
 & 3 & 1.86 (0.00) & 0.88 (0.01) & 70.23 (1.00) & 13.00 (0.00) & 9.76 (0.05) & 0.75 (0.00) \\
 & 5 & 1.67 (0.00) & 0.02 (0.00) & 0.52 (0.00) & 8.00 (0.00) & 2.03 (0.00) & 0.74 (0.00) \\
 & 15a & 1.00 (0.00) & 0.97 (0.00) & 10.71 (0.04) & 4.00 (0.00) & 3.90 (0.00) & 0.74 (0.00) \\
 & 15b & 1.00 (0.00) & 1.00 (0.00) & 19.36 (0.35) & 9.02 (2.00) & 7.00 (0.03) & 0.75 (0.00) \\

\hline \multirow{6}{*}{{\algmigrr}}
 & 1 & 1.98 (0.00) & 1.01 (0.01) & 50.19 (0.33) & 9.80 (1.37) & 5.29 (0.07) & 0.74 (0.00) \\
 & 2 & 1.81 (0.02) & 1.01 (0.01) & 212.55 (3.17) & 12.00 (0.00) & 5.82 (0.05) & 0.74 (0.00) \\
 & 3 & 2.37 (0.03) & 0.89 (0.02) & 147.00 (3.97) & 15.97 (0.22) & 10.65 (0.09) & 0.75 (0.01) \\
 & 5 & 1.98 (0.01) & 1.01 (0.00) & 50.16 (0.30) & 9.80 (1.35) & 5.28 (0.06) & 0.74 (0.00) \\
 & 15a & 1.99 (0.01) & 1.01 (0.01) & 83.80 (0.34) & 8.64 (0.94) & 4.91 (0.01) & 0.74 (0.00) \\
 & 15b & 2.50 (0.01) & 1.00 (0.00) & 136.05 (0.27) & 16.86 (0.35) & 11.41 (0.01) & 0.75 (0.00) \\

\hline \multirow{6}{*}{{\algchor}}
 & 1 & 1.01 (0.00) & 4.89 (0.00) & 44.02 (0.00) & 20.00 (0.00) & 15.88 (0.32) & 0.20 (0.01) \\
 & 2 & 133.86 (0.17) & 2.95 (0.00) & 26.52 (0.00) & 2798.38 (211.11) & 21.41 (0.74) & 0.67 (0.02) \\
 & 3 & 1.22 (0.04) & 4.40 (0.13) & 39.64 (1.16) & 23.65 (0.87) & 18.18 (0.93) & 0.33 (0.02) \\
 & 5 & 1.01 (0.00) & 4.89 (0.00) & 44.02 (0.00) & 20.00 (0.00) & 15.85 (0.33) & 0.20 (0.01) \\
 & 15a & 1.00 (0.00) & 0.98 (0.00) & 8.84 (0.00) & 10.00 (0.00) & 9.25 (0.07) & 0.47 (0.01) \\
 & 15b & 116.52 (1.19) & 2.00 (0.00) & 18.00 (0.00) & 3387.04 (316.16) & 28.01 (1.13) & 0.71 (0.00) \\

  \hline
\end{tabular}}
\end{table}

\newpage
\section{Choreography Setup Phase}
\label{sec:app:choreo}
Choreography as presented in~\cite{DecentMon} splits the initial LTL formula into subformulae and delegates each subformula to a monitor on a component.
Thus choreography presents a complicated \emph{setup} phase.
In this section, we present the \emph{setup} phase.
As such, we present the generation of the decentralized specification from a start LTL formula.

~

Choreography begins by taking the main formula, then deciding to split it into subformulae.
Each monitor will monitor the subformula, notify other monitors of its verdict, and when needed \emph{respawn} .
Recall from the definition of $\Delta'$ (see \rdef{def:sem-decent}), that monitoring is recursively applied to the remainder of a trace starting at the current event.
That is, initially we monitor from $e_0$ to $e_n$ and then from $e_1$ to $e_n$ and so forth.
To do so, it is necessary to reset the state of a monitor appropriately, this process is called in~\cite{DecentMon} a \emph{respawn}.
Once the subformulae are determined, we generate an automaton per subformula to monitor it.
Then, we construct the network of monitors in the form of a tree, in which the root is the main monitor.
Verdicts for each subformula are then propagated in the hierarchy until a verdict can be reached by the root monitor.

A choreography monitor is a tuple $\tuple{id,\aut_{\varphi_{id}}, \chorref_{id}, \chorcoref_{id}, \chorresp_{id}}$ where:
\begin{itemize}
	\item $id$ denotes the monitor unique identifier (label);
	\item $\aut_{id}$ the automaton that monitors the subformula;
	\item $\chorref_{id} : 2^{\mons}$ the monitors that this monitor should notify of a verdict;
	\item $\chorcoref_{id} : 2^{\mons}$ the monitors that send their verdicts to this monitor;
	\item $\chorresp_{id} : \verdictb$ specifies whether the monitor should \emph{respawn};
\end{itemize}
To account for the verdicts from other monitors, the set of possible atoms is extended to include the verdict of a monitor identified by its id.
Therefore, $\atoms = (\timestamp \times \AP) \cup (\mons \times \timestamp)$.
Monitoring is done by replacing the subformula by the id of the monitor associated with it.

Before splitting a formula, it is necessary to determine the component that hosts its monitor.
The component score is computed by counting the number of atomic propositions associated with a component in the subformula.
\begin{align*}
\scor(\varphi, c)
				&: LTL \times \comps \rightarrow \timestamp\\
				&= \texttt{ match $\varphi$ with }\\
				 & \begin{array}{ll}
				\mid a \in AP 				  & \rightarrow \twopartdef{1}{\clookup(a) = c} {0} \\
				\mid \text{ op } \phi	 	  & \rightarrow \scor(\phi, c)\\
				\mid \phi \text{ op } \phi'	  & \rightarrow \scor(\phi, c) + \scor(\phi', c)\\
\end{array}
\end{align*}
The chosen component is determined by the component with the highest score, using $\cchoose : LTL \rightarrow \comps $:
\[
	\cchoose(\varphi) = \underset{c \in \comps}{\operatorname{arg max}}  (\scor(\varphi, c))
\]

In order to setup the network of monitors, firstly the LTL expression is split into subformulae and the necessary monitors are generated to monitor each subformula.
The tree of monitors is generated by recursively splitting the formula at the binary operators.
We present the \emph{setup} phase as a tree traversal of the LTL formula to generate the monitor network, merging nodes at each operator, which is a different flavor of the generation procedure in~\cite{DecentMon}.
Given the two operands, we choose which operands remains in the host component, and (if necessary) which would be placed on a different component.
Therefore, we add the constraint that at least one part of the LTL expression must still remain in the same component.
Given two formulas $\varphi$ and $\varphi'$ and an initial base component $c_{b}$ we determine the two components that should host $\varphi$ and $\varphi'$ with the restriction that one of them is $c_{b}$:
\begin{align*}
	c_1 = \cchoose(\varphi) &, \quad c_2 = \cchoose(\varphi')\\
	s_1 = \scor(\varphi, c_b) &, \quad s_2 = \scor(\varphi', c_b)\\
	\csplit(\varphi, \varphi', c_b)
							 & = \funcparts{
		\tuple{c_b, c_b} & \pif c_1 = c_2 = c_b \\
		\tuple{c_1, c_b} & \pif (c_1 \neq c_b) \land (c_2 = c_b \lor s_2 > s_1)\\
		\tuple{c_b, c_2} & \pelse
}
\end{align*}
\ralg{alg:choreography-setup} displays the procedure to split the formula.
For each binary operator, we determine which of the operands needs to be hosted in a new component.
The result is a tuple: $\tuple{root, N, E}$ where:
\begin{itemize}
	\item $root$ is the root of the tree;
	\item $N$ is the set of generated monitor data;
	\item $E$ is the set of edges between the monitors.
\end{itemize}
Monitor data is a pair $\tuple{id, spec}$ that represents the id of the monitor and the formula that it monitors.

\begin{itemize}
	\item First, $\cchoose$ determines the $\vars{host}$ component where the root monitor resides.
	\item Second, the AST of the $LTL$ formula is traversed using $\func{netx}$, which splits on binary operators.
		\begin{itemize}
			\item If both formulae can be monitored with the same monitor it does not split.
			\item Otherwise
			\begin{enumerate}
				\item We recurse on the side kept, to further split the formula;
				\item We recurse on the side split, with a new $host$ and $id$;
				\item We merge the subnetworks by:
					\begin{enumerate}
						\item Generating the host monitor with the formula resulting from the recursion;
						\item Connecting the split branch's root monitor to the current host monitor;
						\item Adding the split branch's root monitor to the set of additional monitors;
						\item Merging the set of additional monitors and edges from both branches.
					\end{enumerate}
			\end{enumerate}
		\end{itemize}
\end{itemize}

\begin{algorithm}[!tp]
 \caption{Setting up the monitor tree}
 \label{alg:choreography-setup}
\begin{algorithmic}[1]
	\Procedure{$NET\_CHOR$}{$\varphi, C, M$}
		\State $id \gets 0$
		\State $c_h \gets \cchoose(\varphi)$
		\State $\tuple{root, mons, edges} \gets netx(\varphi, id, c_h)$
		\State \Return $\tuple{\setof{root} \cup mons, edges}$
	\EndProcedure
	\Procedure{$netx$}{$\varphi, id_c, c_h$}
		\If{$\varphi \in AP$}
			\State $m \gets \tuple{\varphi, id_c}$
			\State \Return $\tuple{m, \emptyset, \emptyset}$
		\ElsIf{$\varphi$ matches op $e$} \Comment{Unary Operator}
			\State $o \gets \func{netx}(e, id_c, c_h)$
			\State $m \gets \tuple{\mbox{op } o.f, id_c}$
			\State \Return $\tuple{m, o.N, o.E}$
		\ElsIf{$\varphi$ matches $e$ op $e'$}
			\State $\tuple{c_1, c_2} \gets split(e, e', c_h, M)$
			\If{$c_1 = c_2 = c_sec:app:choreoh$} \Comment{No Split}
				\State $l \gets netx(e, id_c, c_h)$
				\State $r \gets netx(e', id_c, c_h)$
				\State $m \gets \tuple{l.f \mbox{ op } r.f, id_c}$
				\State \Return $\tuple{m, l.N \cup r.N, l.E \cup r.E}$
			\ElsIf{$c_1 = c_h$} \Comment{Split Right Branch}
				\State $id_n \gets \func{newid}()$
				\State $l \gets netx(e , id_c, c_h)$
				\State $r \gets netx(e', id_n, c_2)$
				\State $m \gets \tuple{l.f \mbox{ op } \tuple{id_n}, id_c}$
				\State \Return $\tuple{m, (l.N \cup r.N \cup {r.root}),(l.E \cup r.E \cup \setof{\tuple{id_n, id_c})} }$
			\Else \Comment{Split Left Branch}
				\State $id_n \gets \func{newid}()$
				\State $l \gets netx(e , id_n, c_1)$
				\State $r \gets netx(e', id_c, c_h)$
				\State $m \gets \tuple{\tuple{id_n} \mbox{ op } r.f, id_c}$
				\State \Return $\tuple{m, (l.N \cup r.N \cup {l.root}),(l.E \cup r.E \cup \setof{\tuple{id_n, id_c}}) }$
			\EndIf
		\EndIf
	\EndProcedure
\end{algorithmic}
\end{algorithm}

Once the monitor data tree is created, monitors are created accordingly, generating an automaton for the subformula, where some of its atomic propositions have been replaced with monitor ids.
Each monitor is initialized with the refs and corefs set based on the edges setup.

\begin{remark}[Compacting the network]
The monitor network can further be compacted as follows; monitors with the same subformula are merged into one, and their refs and corefs will be the result of the set union.
However one or more merged monitor will have to replace all occurence of the id of the other monitors in all subformulae of all monitors.
\end{remark}

\end{document}